\documentclass{article}

\usepackage{anysize}
\marginsize{2cm}{2cm}{2cm}{2cm}

\newcommand{\n}{\noindent}

\newcommand{\tr}{\text{tr}}

\newcommand{\mb}{\mathbf}

\newcommand{\x}{\times}
\newcommand{\Z}{\mathbb{Z}}
\newcommand{\N}{\mathbb{N}}
\newcommand{\R}{\mathbb{R}}

\renewcommand{\epsilon}{\varepsilon}

\usepackage[dvipsnames]{xcolor}
\usepackage{amssymb,amsmath,mdframed,ytableau,tikz,extarrows,graphicx,float,enumerate,mathdots}
\usepackage[mathscr]{eucal}
\usepackage[nottoc]{tocbibind}
\usetikzlibrary{calc,arrows,math,snakes}
\usepackage[raggedright]{titlesec}
\usepackage[linktocpage=true,colorlinks=true]{hyperref}

\numberwithin{equation}{section}

 \definecolor{cyan}{RGB}{123,215,255}
 \definecolor{red}{RGB}{184,46,32}

\usepackage{amsthm}
\newtheorem{thm}{Theorem}[section]
\newtheorem{lem}[thm]{Lemma}
\newtheorem{cor}[thm]{Corollary}

\newtheorem{prop}[thm]{Proposition}
\newtheorem{defn}[thm]{Definition}

\newtheorem{ex}[thm]{Example}
\newtheorem{rem}[thm]{Remark}

\title{Lusztig Factorization Dynamics of the Full Kostant-Toda Lattices}
\author{Nicholas M. Ercolani\textsuperscript{1,*,$\dagger$} and Jonathan Ramalheira-Tsu\textsuperscript{1,$\star$}}

\date{}

\begin{document}

\maketitle

\textsuperscript{1}Department of Mathematics, University of Arizona, 617 N. Santa Rita Ave., Tucson, 85721-0089, AZ, USA

\begin{center}
    *Corresponding author(s). E-mail(s): ercolani@math.arizona.edu, jramalheiratsu@math.arizona.edu\\
    ORCID: 0000-0003-2010-4205\textsuperscript{$\dagger$}, 0000-0003-2028-3370\textsuperscript{$\star$}\\
\end{center}

\begin{abstract}
We study extensions of the classical Toda lattices at several different space-time scales. These extensions are from the classical tridiagonal phase spaces to the phase space of full Hessenberg matrices, referred to as the Full Kostant-Toda Lattice. Our formulation makes it natural to make further Lie-theoretic generalizations to dual spaces of Borel Lie algebras. Our study brings into play factorizations of Loewner-Whitney type in terms of canonical coordinatizations due to Lusztig. Using these coordinates we formulate precise conditions for the well-posedness of the dynamics at the different space-time scales. Along the way we derive a novel, minimal box-ball system for the Full Kostant-Toda Lattice that does not involve any capacities or colorings, and which has a natural interpretation in terms of the Robinson-Schensted-Knuth algorithm. We provide as well an extension of O'Connell's ordinary differential equations to the Full Kostant-Toda Lattice.\\

\n\textbf{Keywords:} Toda lattice, box-ball system, Lusztig factorization, integrability
\end{abstract}

\n Acknowledgement: The authors gratefully acknowledge support from the National Science Foundation: DMS-1615921.

\tableofcontents

\section{Introduction} \label{sec:intro}
\subsection{Motivation}

The Toda lattice is a completely integrable system of modern vintage (late 1960’s), which over the ensuing decades has shown remarkable resilience in applicability to a wide range of problems in applied mathematics, geometry and analysis. Though initially posed in a classical mechanical setting, under an appropriate coordinate transformation it acquires the form of a Lax equation on tridiagonal Jacobi matrices. This has opened the door to many other extensions including to discrete and ultra-discrete dynamical analogues (described later in this introduction) as well as Lie theoretic generalizations, due to Kostant, inspired by the Lax formulation. A third, more recent, type of extension comes from lifting these Lax representations from Jacobi matrices to a larger phase space related to Borel Lie algebras. This is the Full Kostant Toda lattice and it too has many potential applications to a variety of mathematical areas. Some of these have begun to emerge very recently (see comments at the end of Section
\ref{section:summary}). This paper will lay a firm and uniform foundation for potential applications of the Full Kostant Toda lattice 
in its several continuous and discretized versions and do this in a way that makes Lie theoretic generalizations natural.

In this introduction we will present some historical and conceptual background that will enable us to informally summarize our main results in Section \ref{section:summary}
and then make some brief connections to related literature. After that we outline the remainder of the paper.

\subsection{The Classical Toda Lattice} \label{history}

The Toda lattice \cite{bib:toda} is a dynamical system on $\mathbb{R}^{2n}$, with coordinates $(p_1,\ldots,p_n,q_1,\ldots,q_n)$. The system is Hamiltonian with respect to the standard symplectic structure on $\mathbb{R}^{2n}$ with Hamiltonian

\begin{equation} \label{hamiltonian}
H(p_1,\ldots,p_n,q_1,\ldots,q_n)=\dfrac{1}{2}\sum_{j=1}^n p_j^2 + \sum_{j=1}^{n-1}e^{q_j-q_{j+1}}.
\end{equation}
(Equations (\ref{eqoneofhamiltod}) - (\ref{eqtwoofhamiltod}) present the associated classical Hamiltonian ODEs.)
Flaschka's transformation \cite{bib:fl} introduces the variable replacement $$(p_1,\ldots,p_n,q_1,\ldots,q_n)\mapsto (a_1,\ldots,a_n,b_1,\ldots,b_{n-1})$$ given by setting $a_j=-p_j$ for $j=1,\ldots,n$, and $b_j=e^{q_j-q_{j+1}}$ for $j=1,\ldots,n-1$. In these variables, the Hamiltonian equations may be re-expressed as a Lax equation

\begin{equation} \label{Lax}
\dfrac{d}{ds}X=[X,\pi_-(X)],~~~~X(0)=X_0
\end{equation}
where 
\begin{equation} \label{hessenbergprojection}
X \doteq \left[\begin{array}{cccc} a_1 & 1\\ b_1 & a_2 & \ddots\\ &\ddots & \ddots & 1\\ &&b_{n-1}&a_n\end{array}\right],
\qquad
\pi_-(X)= \left[\begin{array}{cccc} 0 & \\ b_1 & 0 & \\ &\ddots & \ddots & \\ &&b_{n-1}&0\end{array}\right].
\end{equation}

The formulation as a Lax equation (\ref{Lax}) opens a door to connections with symmetry groups (Lie theory) and integrability which spurred the initial fascination that continues to current applications within both pure and applied mathematics. Everything discussed here may be formulated in a general Lie theoretic framework which is an important aspect of this topic; however, for brevity of exposition we will remain with the above specific formulation.

\subsection{dToda} \label{sec:dtoda}
We now consider the following discrete dynamics on tridiagonal Hessenberg matrices; those of the form (\ref{hessenbergprojection}),

\begin{equation}\label{groupversoflax}
    X(t+1)=\Pi_-^{-1}(X(t))X(t)\Pi_-(X(t)), \,\,\,\, t \in \mathbb{N} \cup \{ 0 \}
\end{equation}
where $\Pi_-$ is the projection onto the lower part of the lower-upper factorization of $X(t)$.
Equation (\ref{Lax}) is in fact an infinitesimal analogue of Equation (\ref{groupversoflax}). Equation (\ref{Lax}) is not the {\it exact} infinitesimal version of (\ref{groupversoflax}); the exact version is given by (\ref{luloglaxen}). However, though these two continuous systems are not the same flows on the stated phase space of tridiagonal matrices, they do commute with one another in the Hamiltonian sense of being in involution \cite{bib:arnold}. Traditional usage \cite{bib:h} refers to (\ref{groupversoflax}) as 
the {\it discrete Toda lattice} and also denotes this system by {\it dToda}, so we will continue with that usage but bear in mind the distinction.\\

\n More precisely, dToda is constructed as follows: if one can factor a matrix $X(t)$, of the tridiagonal Hessenberg form shown in (\ref{hessenbergprojection}), as $X(t)=L(t)R(t)$, where $L(t)$ is lower unipotent and $R(t)$ is upper triangular, then
the time $t+1$ matrix $X(t+1)$ in (\ref{groupversoflax}) is given by $$X(t+1)=R(t)L(t).$$  
This is because $\Pi_-(X(t))=L(t)$, so
$$\Pi_-^{-1}(X(t))X(t)\Pi_-(X(t))
=L(t)^{-1}L(t)R(t)L(t)=R(t)L(t).$$
This way of viewing the dynamics, due to Symes \cite{bib:symes78}, amounts to performing a lower-upper factorization of $X(t)$, then flipping the factors. In Section 
\ref{section:dfktodarecdtoda}, we will specify positivity conditions which ensure that this dynamics may be continued for all discrete time.

Since $X(t)$ is a tridiagonal Hessenberg matrix,  one can see that $L(t)$ is lower bi-diagonal with ones on its diagonal and $R(t)$ is upper bi-diagonal with ones on its superdiagonal. Therefore, the product $X(t+1)=R(t)L(t)$ is itself once again a tridiagonal Hessenberg matrix; i.e., this flow preserves the Toda lattice phase space. In fact it is the {\it stroboscope} of a solution to the Lax equation in continuous time $s$ \cite{bib:watkins}, 
\begin{equation} 
\dfrac{d}{ds}X=[X,\pi_-(\log X)] \label{luloglaxen}
\end{equation}
that commutes with (\ref{Lax}). \begin{rem} \label{rem:symes} In our applications here we will always take $X$ to have positive eigenvalues  so that $\log X$ may be uniquely defined in terms of the principal branch of the logarithm along the positive real axis. From this it follows that this continuous flow is well-defined and commutes with the original Toda flow
\cite{bib:watkins, bib:dlt}. The Hamiltonian for (\ref{luloglaxen}) is $H_{LU} = \mbox{Tr}(X \log X -  X)$. 
\end{rem}

\subsection{The Box-Ball System}
The passage from the time-discrete dToda system to a system that is spatially discrete as well is mediated by a process called {\it ultra-discretization} (see Section \ref{sec:ultradis}). The resultant system is denoted {\it udToda} and has a remarkable presentation in terms of cellular automata. We briefly describe the latter here. 
The (classical) box-ball system (BBS) consists of an infinite number of boxes arranged as a one-dimensional array with a finite number of boxes containing a ball. A simple evolution rule is provided for the box-ball dynamics:\index{Basic Box-Ball Evolution}\index{Box-Ball System}
\begin{enumerate}[(1)]
\item Take the left-most ball that has not been moved and move it to the left-most empty box to its right.
\item Repeat (1) until all balls have been moved precisely once.
\end{enumerate}

\n Since the algorithm requires one to know which balls have been moved, we can, without technically changing the algorithm, introduce a colour-coding based on whether balls have moved or not. Balls will be blue until they have moved, after which they will become red. When all balls are red, the colours should be reset to blue, ready for the next time step. Or, equivalently, a $0$-th step of colouring all balls blue should be prescribed. We will use the latter for a minor benefit in brevity. Below is an example of the evolution with this colour-coding, with each ball movement separated into a sub-step:

\begin{figure}[H]
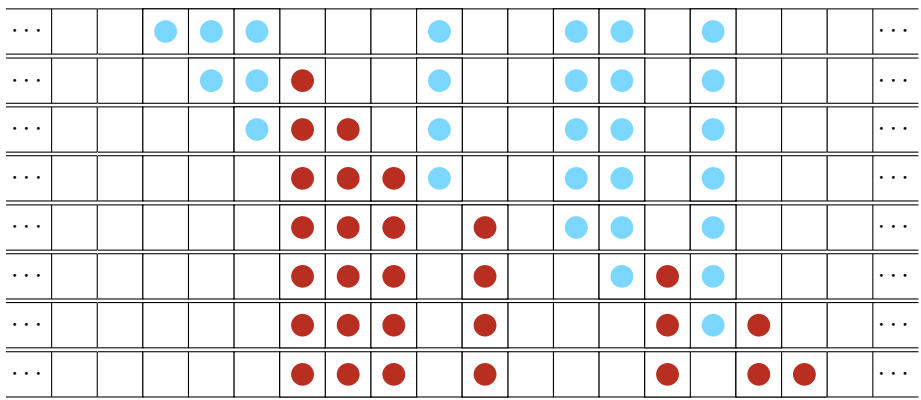

\centering
\tikz[scale=0.6]{
\foreach \x in {0,1,2,3,4,5,6,7,8,9,10,11,12,13,14,15}
{\draw[fill=white]  (\x,3) -- (\x+1,3) -- (\x+1,4) -- (\x,4) -- cycle;			
}
\foreach \x in {1,2,3,7,10,11,13}
{\draw[fill=white]  (\x,3) -- (\x+1,3) -- (\x+1,4) -- (\x,4) -- cycle;			
\fill[cyan] (\x+0.5,3.5) circle (0.25);
}
\foreach \x in {}
{\draw[fill=white]  (\x,3) -- (\x+1,3) -- (\x+1,4) -- (\x,4) -- cycle;			
\fill[red] (\x+0.5,3.5) circle (0.25);
}
\foreach \x in {16}
{\draw[fill=white,white]  (\x,3) -- (\x+2,3) -- (\x+2,4) -- (\x,4) -- cycle;
\draw[-] (\x,3) -- (\x,4);
\draw[-] (\x,3) -- (\x+2,3);
\draw[-] (\x,4) -- (\x+2,4);
\draw[-] (\x+1,3) -- (\x+1,4);
\node at (\x+1.5,3.5) {$\cdots$};
}
\foreach \x in {0}
{\draw[fill=white,white]  (\x,3) -- (\x-2,3) -- (\x-2,4) -- (\x,4) -- cycle;
\draw[-] (\x,3) -- (\x,4);
\draw[-] (\x,3) -- (\x-2,3);
\draw[-] (\x,4) -- (\x-2,4);
\draw[-] (\x-1,3) -- (\x-1,4);
\node at (\x-1.5,3.5) {$\cdots$};
}
}
\tikz[scale=0.6]{
\foreach \x in {0,1,2,3,4,5,6,7,8,9,10,11,12,13,14,15}
{\draw[fill=white]  (\x,3) -- (\x+1,3) -- (\x+1,4) -- (\x,4) -- cycle;			
}
\foreach \x in {2,3,7,10,11,13}
{\draw[fill=white]  (\x,3) -- (\x+1,3) -- (\x+1,4) -- (\x,4) -- cycle;			
\fill[cyan] (\x+0.5,3.5) circle (0.25);
}
\foreach \x in {4}
{\draw[fill=white]  (\x,3) -- (\x+1,3) -- (\x+1,4) -- (\x,4) -- cycle;			
\fill[red] (\x+0.5,3.5) circle (0.25);
}
\foreach \x in {16}
{\draw[fill=white,white]  (\x,3) -- (\x+2,3) -- (\x+2,4) -- (\x,4) -- cycle;
\draw[-] (\x,3) -- (\x,4);
\draw[-] (\x,3) -- (\x+2,3);
\draw[-] (\x,4) -- (\x+2,4);
\draw[-] (\x+1,3) -- (\x+1,4);
\node at (\x+1.5,3.5) {$\cdots$};
}
\foreach \x in {0}
{\draw[fill=white,white]  (\x,3) -- (\x-2,3) -- (\x-2,4) -- (\x,4) -- cycle;
\draw[-] (\x,3) -- (\x,4);
\draw[-] (\x,3) -- (\x-2,3);
\draw[-] (\x,4) -- (\x-2,4);
\draw[-] (\x-1,3) -- (\x-1,4);
\node at (\x-1.5,3.5) {$\cdots$};
}
}
\tikz[scale=0.6]{
\foreach \x in {0,1,2,3,4,5,6,7,8,9,10,11,12,13,14,15}
{\draw[fill=white]  (\x,3) -- (\x+1,3) -- (\x+1,4) -- (\x,4) -- cycle;			
}
\foreach \x in {3,7,10,11,13}
{\draw[fill=white]  (\x,3) -- (\x+1,3) -- (\x+1,4) -- (\x,4) -- cycle;			
\fill[cyan] (\x+0.5,3.5) circle (0.25);
}
\foreach \x in {4,5}
{\draw[fill=white]  (\x,3) -- (\x+1,3) -- (\x+1,4) -- (\x,4) -- cycle;			
\fill[red] (\x+0.5,3.5) circle (0.25);
}
\foreach \x in {16}
{\draw[fill=white,white]  (\x,3) -- (\x+2,3) -- (\x+2,4) -- (\x,4) -- cycle;
\draw[-] (\x,3) -- (\x,4);
\draw[-] (\x,3) -- (\x+2,3);
\draw[-] (\x,4) -- (\x+2,4);
\draw[-] (\x+1,3) -- (\x+1,4);
\node at (\x+1.5,3.5) {$\cdots$};
}
\foreach \x in {0}
{\draw[fill=white,white]  (\x,3) -- (\x-2,3) -- (\x-2,4) -- (\x,4) -- cycle;
\draw[-] (\x,3) -- (\x,4);
\draw[-] (\x,3) -- (\x-2,3);
\draw[-] (\x,4) -- (\x-2,4);
\draw[-] (\x-1,3) -- (\x-1,4);
\node at (\x-1.5,3.5) {$\cdots$};
}
}
\tikz[scale=0.6]{
\foreach \x in {0,1,2,3,4,5,6,7,8,9,10,11,12,13,14,15}
{\draw[fill=white]  (\x,3) -- (\x+1,3) -- (\x+1,4) -- (\x,4) -- cycle;			
}
\foreach \x in {7,10,11,13}
{\draw[fill=white]  (\x,3) -- (\x+1,3) -- (\x+1,4) -- (\x,4) -- cycle;			
\fill[cyan] (\x+0.5,3.5) circle (0.25);
}
\foreach \x in {4,5,6}
{\draw[fill=white]  (\x,3) -- (\x+1,3) -- (\x+1,4) -- (\x,4) -- cycle;			
\fill[red] (\x+0.5,3.5) circle (0.25);
}
\foreach \x in {16}
{\draw[fill=white,white]  (\x,3) -- (\x+2,3) -- (\x+2,4) -- (\x,4) -- cycle;
\draw[-] (\x,3) -- (\x,4);
\draw[-] (\x,3) -- (\x+2,3);
\draw[-] (\x,4) -- (\x+2,4);
\draw[-] (\x+1,3) -- (\x+1,4);
\node at (\x+1.5,3.5) {$\cdots$};
}
\foreach \x in {0}
{\draw[fill=white,white]  (\x,3) -- (\x-2,3) -- (\x-2,4) -- (\x,4) -- cycle;
\draw[-] (\x,3) -- (\x,4);
\draw[-] (\x,3) -- (\x-2,3);
\draw[-] (\x,4) -- (\x-2,4);
\draw[-] (\x-1,3) -- (\x-1,4);
\node at (\x-1.5,3.5) {$\cdots$};
}
}
\tikz[scale=0.6]{
\foreach \x in {0,1,2,3,4,5,6,7,8,9,10,11,12,13,14,15}
{\draw[fill=white]  (\x,3) -- (\x+1,3) -- (\x+1,4) -- (\x,4) -- cycle;			
}
\foreach \x in {10,11,13}
{\draw[fill=white]  (\x,3) -- (\x+1,3) -- (\x+1,4) -- (\x,4) -- cycle;			
\fill[cyan] (\x+0.5,3.5) circle (0.25);
}
\foreach \x in {4,5,6,8}
{\draw[fill=white]  (\x,3) -- (\x+1,3) -- (\x+1,4) -- (\x,4) -- cycle;			
\fill[red] (\x+0.5,3.5) circle (0.25);
}
\foreach \x in {16}
{\draw[fill=white,white]  (\x,3) -- (\x+2,3) -- (\x+2,4) -- (\x,4) -- cycle;
\draw[-] (\x,3) -- (\x,4);
\draw[-] (\x,3) -- (\x+2,3);
\draw[-] (\x,4) -- (\x+2,4);
\draw[-] (\x+1,3) -- (\x+1,4);
\node at (\x+1.5,3.5) {$\cdots$};
}
\foreach \x in {0}
{\draw[fill=white,white]  (\x,3) -- (\x-2,3) -- (\x-2,4) -- (\x,4) -- cycle;
\draw[-] (\x,3) -- (\x,4);
\draw[-] (\x,3) -- (\x-2,3);
\draw[-] (\x,4) -- (\x-2,4);
\draw[-] (\x-1,3) -- (\x-1,4);
\node at (\x-1.5,3.5) {$\cdots$};
}
}
\tikz[scale=0.6]{
\foreach \x in {0,1,2,3,4,5,6,7,8,9,10,11,12,13,14,15}
{\draw[fill=white]  (\x,3) -- (\x+1,3) -- (\x+1,4) -- (\x,4) -- cycle;			
}
\foreach \x in {11,13}
{\draw[fill=white]  (\x,3) -- (\x+1,3) -- (\x+1,4) -- (\x,4) -- cycle;			
\fill[cyan] (\x+0.5,3.5) circle (0.25);
}
\foreach \x in {4,5,6,8,12}
{\draw[fill=white]  (\x,3) -- (\x+1,3) -- (\x+1,4) -- (\x,4) -- cycle;			
\fill[red] (\x+0.5,3.5) circle (0.25);
}
\foreach \x in {16}
{\draw[fill=white,white]  (\x,3) -- (\x+2,3) -- (\x+2,4) -- (\x,4) -- cycle;
\draw[-] (\x,3) -- (\x,4);
\draw[-] (\x,3) -- (\x+2,3);
\draw[-] (\x,4) -- (\x+2,4);
\draw[-] (\x+1,3) -- (\x+1,4);
\node at (\x+1.5,3.5) {$\cdots$};
}
\foreach \x in {0}
{\draw[fill=white,white]  (\x,3) -- (\x-2,3) -- (\x-2,4) -- (\x,4) -- cycle;
\draw[-] (\x,3) -- (\x,4);
\draw[-] (\x,3) -- (\x-2,3);
\draw[-] (\x,4) -- (\x-2,4);
\draw[-] (\x-1,3) -- (\x-1,4);
\node at (\x-1.5,3.5) {$\cdots$};
}
}
\tikz[scale=0.6]{
\foreach \x in {0,1,2,3,4,5,6,7,8,9,10,11,12,13,14,15}
{\draw[fill=white]  (\x,3) -- (\x+1,3) -- (\x+1,4) -- (\x,4) -- cycle;			
}
\foreach \x in {13}
{\draw[fill=white]  (\x,3) -- (\x+1,3) -- (\x+1,4) -- (\x,4) -- cycle;			
\fill[cyan] (\x+0.5,3.5) circle (0.25);
}
\foreach \x in {4,5,6,8,12,14}
{\draw[fill=white]  (\x,3) -- (\x+1,3) -- (\x+1,4) -- (\x,4) -- cycle;			
\fill[red] (\x+0.5,3.5) circle (0.25);
}
\foreach \x in {16}
{\draw[fill=white,white]  (\x,3) -- (\x+2,3) -- (\x+2,4) -- (\x,4) -- cycle;
\draw[-] (\x,3) -- (\x,4);
\draw[-] (\x,3) -- (\x+2,3);
\draw[-] (\x,4) -- (\x+2,4);
\draw[-] (\x+1,3) -- (\x+1,4);
\node at (\x+1.5,3.5) {$\cdots$};
}
\foreach \x in {0}
{\draw[fill=white,white]  (\x,3) -- (\x-2,3) -- (\x-2,4) -- (\x,4) -- cycle;
\draw[-] (\x,3) -- (\x,4);
\draw[-] (\x,3) -- (\x-2,3);
\draw[-] (\x,4) -- (\x-2,4);
\draw[-] (\x-1,3) -- (\x-1,4);
\node at (\x-1.5,3.5) {$\cdots$};
}
}
\tikz[scale=0.6]{
\foreach \x in {0,1,2,3,4,5,6,7,8,9,10,11,12,13,14,15}
{\draw[fill=white]  (\x,3) -- (\x+1,3) -- (\x+1,4) -- (\x,4) -- cycle;			
}
\foreach \x in {}
{\draw[fill=white]  (\x,3) -- (\x+1,3) -- (\x+1,4) -- (\x,4) -- cycle;			
\fill[cyan] (\x+0.5,3.5) circle (0.25);
}
\foreach \x in {4,5,6,8,12,14,15}
{\draw[fill=white]  (\x,3) -- (\x+1,3) -- (\x+1,4) -- (\x,4) -- cycle;			
\fill[red] (\x+0.5,3.5) circle (0.25);
}
\foreach \x in {16}
{\draw[fill=white,white]  (\x,3) -- (\x+2,3) -- (\x+2,4) -- (\x,4) -- cycle;
\draw[-] (\x,3) -- (\x,4);
\draw[-] (\x,3) -- (\x+2,3);
\draw[-] (\x,4) -- (\x+2,4);
\draw[-] (\x+1,3) -- (\x+1,4);
\node at (\x+1.5,3.5) {$\cdots$};
}
\foreach \x in {0}
{\draw[fill=white,white]  (\x,3) -- (\x-2,3) -- (\x-2,4) -- (\x,4) -- cycle;
\draw[-] (\x,3) -- (\x,4);
\draw[-] (\x,3) -- (\x-2,3);
\draw[-] (\x,4) -- (\x-2,4);
\draw[-] (\x-1,3) -- (\x-1,4);
\node at (\x-1.5,3.5) {$\cdots$};
}
}
\caption{A box-ball system time evolution (one time step).}\label{firstbbsexample}
\end{figure}	
\n Each box-ball configuration is coordinatised by counting balls in blocks and spaces between blocks. Tokihiro \cite{bib:tokihiro} showed that dToda can be spatially discretised (ultradiscretization, \ref{sec:ultradis}) in such a way that the spatial discretization (udToda) describes the coordinate evolution of the box-ball system.\\

Although the coordinates were originally developed to count block lengths, coordinate evolution of udToda makes sense even when coordinates are taken to be zero. In fact, making geometric sense of the vanishing of coordinates leads to a remarkable new cellular automaton known as the ghost-box-ball system (GBBS), which the authors first introduced in \cite{bib:era}.

\subsection{The Full Kostant-Toda Lattice} \label{sec:FKT}

The main themes of this paper concern an extension
of the Lax equation (\ref{Lax}) from the tridiagonal Hessenberg phase space represented in (\ref{hessenbergprojection}) to the full phase space of {\it all} lower Hessenberg matrices which will be denoted by 
$\mathcal{H}$.

This dynamical system on $\mathcal{H}$ was first introduced in \cite{bib:efs}, and named therein as the {\it Full Kostant-Toda lattice } (abbreviated in this paper as FKToda), in the context of complete integrability for dynamical systems associated to the classical semi-simple Lie algebras. From the viewpoint of complete integrability it is also related to (but different from) the Toda lattice equations defined on generic symmetric matrices \cite{bib:dlnt}.

The central idea for us is that the various discrete analogues of FKToda that we will consider may be deconstructed in terms of coupled systems of dToda. Hence, it is worth mentioning here that, in continuous time, the {\it tridiagonal} Hessenberg matrices comprise an invariant subspace for FKToda.  A more general and systematic discussion of the invariance of bandlimited Hessenberg matrices is presented in Section \ref{ExtRed}.

\subsection{Summary of Results} \label{section:summary}

\begin{figure}[H]
\centering\renewcommand{\arraystretch}{1.4}
\begin{tabular}{|c|c|c|}
\hline
     Scale & Classical System & Full System\\
     \hline
     Discrete space, continuous time & Toda Lattice & Full Kostant-Toda (FKToda)\\
     Discrete space, discrete time & ``Discrete Toda'' (dToda) & Symes's Map (dFToda)\\
     Ultradiscrete space, discrete time & Box-Ball System (BBS) & Ultradiscrete Full Toda (udFToda)\\
     \hline
\end{tabular}
\caption{The dynamical systems in their three spatio-temporal scales, and in both the classical and full settings.}
\label{casttable}
\end{figure}

This paper provides a description of how the classical systems just described (listed in the middle column of Figure \ref{casttable}) extend to the Full systems listed in the last column. The table also summarizes the notations for these systems that will be used going forward. Note that we refer to the Symes Map as {\it dFToda} rather than what might seem more natural to refer to as ``dFKToda". The reason for this is in accord with the explanation given at the start of Section \ref{sec:dtoda} related to dToda. The Symes map is not a time discretization of FKToda; rather, it is the discretization of a flow, \eqref{luloglaxen}, that commutes with FKToda and so we employ a different denotation.

This work establishes precise conditions for existence and uniqueness of solutions to the Full systems and explains the underlying group theoretic framework for understanding their properties and potential generalizations.
The first column describes the scales, or domains, usually associated with these models. Another scale that might have been considered here is that of continuous space and continuous time. One model for that in the middle column would be the KdV equation, for which the classical Toda lattice is often regarded as an integrable spatial discretization. A natural analogue for this in the third column would be the KP equation. There is a large literature related to a discretization
of KP referred to as dKP \cite{bib:hirota81} and \cite{bib:knw} and references therein. However these are usually presented in a formal infinite dimensional setting. Our approach stems from the Kostant-Kirillov Poisson bracket on the dual of a Lie algebra \cite{bib:efs}, corresponding to the Hessenberg matrices in the cases that we focus on here. It is in this setting that it becomes natural to consider matrix factorizations and potential applications such as to representation theory, random matrix theory, orthogonal polynomials \cite{bib:ew}. However, explorations of connections to dKP, along the lines related to \cite{bib:sik} as described below and in Section \ref{sec:connectionsliterature}, is something to be considered in the future.\\


\n As we have just mentioned, as well as seen above in Section \ref{groupversoflax} and later in Section \ref{fulltodasolnfactorsconj}, matrix factorization plays a fundamental central role in the dynamical systems we are considering. Specifically, we consider the lower-upper factorizations. There is a deeper factorization structure dating back to the work of Loewner and Whitney (Section \ref{sec:lw}) applied to \textit{totally positive} Hessenberg matrices (\ref{resphase}). This is rooted in the study of weighted path matrices for planar networks \cite{bib:fz}, which have a relation to recent developments in first and last passage percolation \cite{bib:ganguly}. Such structures also appear in a related context in \cite{bib:er}. This deeper factorization produces a unique coordinatization of the Toda phase spaces in terms of so-called \textit{Lusztig} coordinates (Section \ref{sec:lusztigfactors}). These coordinates and their associated factorizations, which may be generally referred to as {\it Lusztig factorizations},  are a recurring theme throughout this paper. \\

\n We discuss Lusztig factorization at the three spatio-temporal scales described in Figure \ref{casttable}. First, in the discrete time, discrete space scale, Theorem \ref{remark:substepsdtoda} provides a decomposition of the full discrete Toda map (dFToda) as a sequence of coupled dToda maps. We further establish through Theorem \ref{thm:wp} and Corollary \ref{cor:wp} the well-posedness of the dFToda map under certain positivity conditions, and explicitly describe the evolution on the Lusztig coordinates in terms of $\tau$-functions.

\n As we move on to the ultradiscrete setting, we introduce a full extension of udToda (or BBS) in Definition \ref{definitionofudftodaevol}, in a way that most naturally mirrors the full extensions of classic Toda to full Kostant-Toda, and dToda to dFToda. In some sense, this should be considered the minimal extension that serves this purpose. As an additional benefit of this particular system, we establish a method of fully capturing the Robinson-Schensted-Knuth (RSK) correspondence (a fundamental algorithmic bijection with applications in combinatorics and representation theory \cite{bib:aigner}) in the udFToda dynamics, right down to determining both the insertion and recording tableaux. However, we defer full details of this to an upcoming paper, opting to demonstrate the main ideas through an example which alludes to the general picture. This can be found in Remark \ref{rem:rskasfulludtoda}.\\[3pt]

\n Lastly, in the discrete space, continuous time setting, we establish in Section \ref{sec:return} a Lusztig coordinate description of the full Kostant-Toda lattice, culminating in both matrix and coordinate descriptions of dynamics in terms of Lusztig coordinates in Theorem \ref{thm:matrixandcoordrepoffktoda}. 

 \bigskip

Discussions of bidiagonal representations of the Full Kostant Toda lattice and its discretizations have appeared very recently elsewhere in the literature. In \cite{bib:sik} another Toda-type system, the so-called discrete hungry Toda lattice (dhToda), is presented along with a prescription for a continuum limit that yields FKToda in  bidiagonal form. The approach in \cite{bib:sik} differs from ours in that it goes from dhToda to FKToda while ours proceeds directly, in the opposite way,  from continuous to discrete,
by embedding an integrable discrete system, due to Symes,  in FKToda. 
\medskip

This paper also yields an alternative perspective on and extension of integrable systems aspects of the seminal work of O'Connell and collaborators that relates the classical Toda lattice
to random matrix theory and associated stochastic processes as well as geometric versions of RSK. Extensions of these kinds of connections to the Full Toda systems is an interesting possible future direction for exploration. In Appendix \ref{appendixb} we show through Theorem \ref{thm:oconnellsetup} and Corollary \ref{cor:oconnellsetup} how one can obtain, simply from Theorem \ref{thm:matrixandcoordrepoffktoda}, O'Connell's differential equations in \cite{bib:o}. 
\medskip 

These points of comparison will be expanded on in Section \ref{sec:connectionsliterature}.

\subsection{Outline}

In Section \ref{sec:intro} and Section \ref{sectionDTodaBBS} we provide the historical and technical background, respectively, needed for this paper. For an expert in the classical Toda integrable systems, much of this background could be skipped or skimmed.
The development of our first results takes place in Section \ref{section:dfktodarecdtoda} where dFToda is explicitly described as a coupled system of dToda lattices. This section also provides precise conditions for existence and well-posedness of the dFToda flows. Section \ref{sec:ud} directly reads off, from those previous results, the detailed structure of udFToda in terms of tropicalized Lusztig parameters and relates this, by example, to the RSK algorithm.  In Section \ref{geometry} we change gears somewhat and present the prior results in a fully Lie theoretic framework. This enables us, in Section  \ref{sec:LusDynam}, to intrinsically describe the dynamics on Lusztig parameters in terms of tau functions associated to the fundamental representations of semisimple  Lie algebras. Finally in Section \ref{sec:return} we show how the discrete system analysis we have carried out informs a deeper analysis of the continuous dynamics of FKToda.  Section \ref{sec:conclusions} presents some concluding remarks describing connections to the current literature and directions for future exploration based on what has been done here. Appendix \ref{appA} 
goes into some of the more technical aspects of the relation between tau functions and factorization. Appendix \ref{appendixb} details the derivation of O'Connell's ODEs from our results in Section \ref{sec:return}. 
\section{Background}\label{sectionDTodaBBS}

In this section we present the essential technical background needed to precisely state and prove our results. This includes the factorization method for solving both the continuous and discrete Toda lattices, the coordinate representations for discrete and ultra-discrete Toda, and Lusztig's method for unique bi-diagonal factorization of unipotent matrices.

\subsection{Solution by Factorization} \label{features}

As observed in Section \ref{sec:FKT}, the classical tri-diagonal system is an invariant subsystem of FKToda. The dynamics of FKToda is specified by the Lax equation (\ref{Lax}) {\it extended} to the full phase space of lower Hessenberg matrices $\mathcal{H} $; i.e., to all matrices of the form  
\begin{eqnarray*}
X \in  \left(\begin{array}{ccccc}
* & 1 & & &\\
*& * & 1 & &\\
\vdots & \ddots & \ddots & \ddots &\\
\vdots &  & \ddots & \ddots & 1\\
* & \dots & \dots & * & *
\end{array} \right).
\end{eqnarray*}

Having the form of a Lax equation points the way to constructing explicit solutions of FKToda; effectively, this Lax equation is the infinitesimal version of the statement that solutions advance by conjugating initial data with respect to a group element coming from the lower unipotent projection of $e^{sX}$. This was made precise independently by Adler, Kostant and Symes, in the late 1970's (\cite{bib:adler}, \cite{bib:kostant}, \cite{bib:symes}). In the version applied here, this is based on the so-called $LU$ factorization (equivalent to Gaussian elimination): given an invertible matrix $g \in GL(n, \mathbb{R})$ one seeks a factorization of the form 
$g = L R$ where $L$ is lower unipotent and $R$ is invertible, upper triangular.  When such a factorization exists (and generically it does), it is unique and $L$ will be denoted by $\Pi_-(g)$ and $R$ by $\Pi_+(g)$. 
In this extended setting the following realizes the construction of explicit solutions.

\begin{thm}      \cite{bib:adler, bib:kostant, bib:symes}       \label{factorisationthmbkg}
{\textbf{(The Factorization Theorem)}} To solve the FKToda system, \vspace{0.15cm}
\begin{equation} \label{Lax2}
\dfrac{d}{ds}X=[X,\pi_-(X)],~~~~X(0)=X_0 \in \mathcal{H},
\end{equation}
where now
\begin{eqnarray*}
\pi_- (X) \in  \left(\begin{array}{ccccc}
0 &  & & &\\
*& 0 & & &\\
\vdots & \ddots & \ddots & &\\
\vdots &  & \ddots & 0 &\\
* & \dots & \dots & * & 0
\end{array} \right),
\end{eqnarray*}

factor $e^{s X_0 }=\Pi_-( e^{s X_0})\Pi_+(e^{s X_0})$, if possible (locally it is). Then, the solution is given by\vspace{0.15cm}
\begin{equation}\label{fulltodasolnfactorsconj}
X(s)=\Pi_-^{-1}(e^{s X_0})X_0\Pi_-(e^{s X_0}).
\end{equation}
\end{thm}

It is a straightforward application of the product rule using (\ref{Lax2}) to see that 
\begin{equation} \label{Isospectral}
\frac{d}{ds} \tr X^k = \tr \frac{d}{ds} X^k = \tr \left[ X^k , \pi_- (X) \right] = 0.
\end{equation}
This implies the so-called {\it isospectrality} of the Toda lattice: the eigenvalues of $X$ remain invariant under the Toda flow.
Moreover, with respect to an underlying symplectic structure generalizing the standard one for (\ref{hamiltonian}) \cite{bib:efs}, these constants of motion are in involution in the sense of Arnold-Liouville \cite{bib:arnold}. 

\subsection{Discretization}
Symes's form of dToda described in Section \ref{sec:dtoda} has a natural extension to a discrete analogue of FKToda, again based on the $LU$ factorization algorithm.

Symes's dynamics is inductively defined  as a two-step discrete evolution on Hessenberg matrices. If at (discrete) time $t$ one has a matrix $X(t)$, one obtains $X(t+1)$ as follows:

\begin{enumerate}
\item Perform Gaussian elimination to factor $X(t)=L(t)R(t)$, with $L(t)$ lower unipotent and $R(t)$ upper triangular.
\item Permute the factors to define $X(t+1) = R(t)L(t)$.
\end{enumerate}

By construction, one has
\begin{equation} \label{symeseqn}
X(t+1)=R(t)L(t) = (L(t)^{-1}X(t))L(t) = L(t)^{-1}X(t)L(t).
\end{equation}

\n Thus, this discrete evolution is given by conjugating a matrix by its lower unipotent factor. Since the spectrum of a matrix is invariant under conjugation, it follows that the eigenvalues are constants of motion for this discrete evolution; i.e., this discrete flow is isospectral. This is completely analogous to what was seen in Theorem \ref{factorisationthmbkg}.

\n This flow is what we shall henceforth call discrete time full Kostant-Toda (dFToda).

\begin{prop} \cite{bib:symes, bib:dlt, bib:watkins} \label{prop:dft}
To solve the discrete-time Full Toda lattice with initial condition $X(0)=X_0$, factor $e^{t\log X_0}=X_0^t=\Pi_-(X_0^t)\Pi_+(X_0^t)$, if possible. Then, the solution is given by\vspace{0.4cm}
\begin{equation}\label{todasolnfactorsconj}
X(t)=\Pi_-^{-1}(X_0^t)X_0\Pi_-(X_0^t), 
\end{equation}
for all $t\in \mathbb{N} \cup\{0\}$. (See Section \ref{ExtRed} for discussion of extending this flow to all $t \in \mathbb{Z}$.)
\end{prop}

It is of course not always the case that a given matrix has an LU factorization in which the coefficients of the factors do not become singular. It is possible to continue the dynamics (both continuous and discrete) through these singularities \cite{bib:efs}; however, we will not need to deal with that in this paper since it will be seen in 
Corollary \ref{cor:wp} that under appropriate {\it positivity} conditions on the initial data $X_0$, the discrete flows exist for all forward time without any singularities arising.

\subsection{Coordinate Representations}

Up to this point it has been both sufficient and useful to represent the Toda equations we consider in the compact form of a Lax equation defined on a matrix phase space.  But going forward we will need to work with explicit coordinate representations of these differential and difference equations. The classical Toda lattice discussed in Section \ref{history} first appeared in coordinate form as Hamilton's equations for the Hamiltonian (\ref{hamiltonian}) which are
given by

\begin{align}
\dot{q}_j&=p_j, ~~~~~~~~~~~~~~~~~~~~\,~~~~~~ j=1,\ldots,n,\label{eqoneofhamiltod}\\\notag\\
\dot{p}_j&=\left\{\begin{array}{cl}
-e^{q_1-q_2} & \text{if }j=1\\
e^{q_{j-1}-q_j}-e^{q_{j}-q_{j+1}} & \text{if }1<j<n\\
e^{q_{n-1}-q_n} & \text{if }j=n.
\end{array}\right.\label{eqtwoofhamiltod}
\end{align}

\n In this representation, boundary conditions of $q_0=-\infty$ and $q_{n+1}=\infty$ have been imposed, which result in $e^{q_0-q_1}=e^{q_n-q_{n+1}}=0$.

In the remainder of this subsection we describe the corresponding coordinate forms of the discrete and ultra-discrete Toda systems.

\subsubsection{dToda}

In the tridiagonal case, it is traditional \cite{bib:h} to use $I_1,\ldots,I_n$ to denote the diagonal entries of the upper bidiagonal matrix $R(t)$ and $V_1,\ldots,V_{n-1}$ to denote the subdiagonal entries of the lower bidiagonal matrix $L(t)$:

\begin{equation}\label{ltrtdtoda}
L(t)=\arraycolsep=3.1pt\def\arraystretch{1.5}\left[\begin{array}{cccc}
1\\
V_1^t&1\\&\ddots&\ddots\\
&&V_{n-1}^t&1
\end{array}\right],~~~~\text{and}~~~~R(t)=\arraycolsep=4.4pt\def\arraystretch{1.3}\left[\begin{array}{cccc}
I_1^t & 1\\
&I_2^t&\ddots\\
&&\ddots&1\\
&&&I_n^t
\end{array}\right].
\end{equation}

\n The dynamic evolution $L(t+1)R(t+1)=R(t)L(t)$ can then be written out explicitly, and shown \cite{bib:tokihiro} to be equivalent to the following system:

\begin{equation} \label{explicitDT}
\left\{\renewcommand{\arraystretch}{1.6}\setlength{\tabcolsep}{16pt}
\begin{array}{lcl}
V_0^t=V_n^t=0,\\
I_i^{t+1}=V_i^t+\dfrac{I_i^t\cdots I_{1}^t}{I_{i-1}^{t+1}\cdots I_{1}^{t+1}}&&i=1,\ldots,n,\\
V_i^{t+1}I_i^{t+1}=I_{i+1}^tV_i^t &~~~~~~~& i=1,\ldots,n-1.
\end{array}
\right.
\end{equation}

\subsubsection{Ultra-discretization and Box-Ball Systems}\label{sec:ultradis}

The space-time discretization of the Toda lattice we begin with is derived from dToda by a semiclassical type of limit, known as {\it Maslov dequantization} or {\it tropicalization} \cite{bib:lmrs,bib:era}. The resulting ultradiscrete system is a cellular automaton, first  introduced in 1990 by Takahashi and Satsuma \cite{bib:ts} who referred to it as the soliton box-ball system (BBS) for reasons that will become evident below.\\

\n To perform the ultradiscretization, one pushes forwards the semiring structure on $(\R_{\geq 0},+,\times)$ to $\R\cup\{\infty\}$ by the family of bijections $D_\hbar$, for $\hbar>0$, given by 

\begin{equation}
D_\hbar(x)=\left\{\begin{array}{ccl}
-\hbar\ln x &~~~& \text{if }x\neq 0\\
-\infty && \text{if }x= 0
\end{array}\right.,
\end{equation}
and taking the limit as $\hbar\to 0^+$. The essential result of this is that it replaces the operations $+$ and $\times$ on $\R_{\geq 0}$ by operations $\min$ and $+$, respectively, on $\R\cup\{\infty\}$.\\

\n For dToda, given in the form of Equations \ref{explicitDT}, this process amounts to making the change of variables
\begin{equation}
I_i^t=e^{-\frac{1}{\hbar}Q_i^t(\hbar)},~~~~
V_i^t=e^{-\frac{1}{\hbar}W_i^t(\hbar)},\label{tropvarsforbbstoda}
\end{equation}
and taking the zero $\hbar$-limit. The result of this semiclassical limit yields the equations of Theorem \ref{thmbbscoords2}, describing the (coordinate) evolution of the box-ball system.\\

\n Recall that the (basic) box-ball system consists of a one-dimensional infinite array of boxes with a finite number of the boxes filled with balls, and no more than one ball in each box (see, for example, Figure \ref{firstbbsexfordef}). 
  One refers to a full consecutive sequence of balls (having empty boxes on each side) as a {\it block}. One may think of these blocks as being coherent solitary masses.

\begin{figure}[H]
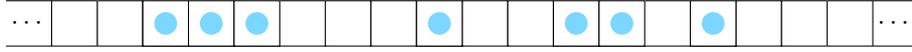

\centering
\tikz[scale=0.6]{
\foreach \x in {0,1,2,3,4,5,6,7,8,9,10,11,12,13,14,15}
{\draw[fill=white]  (\x,3) -- (\x+1,3) -- (\x+1,4) -- (\x,4) -- cycle;			
}
\foreach \x in {1,2,3,7,10,11,13}
{\draw[fill=white]  (\x,3) -- (\x+1,3) -- (\x+1,4) -- (\x,4) -- cycle;			
\fill[cyan] (\x+0.5,3.5) circle (0.25);
}
\foreach \x in {}
{\draw[fill=white]  (\x,3) -- (\x+1,3) -- (\x+1,4) -- (\x,4) -- cycle;			
\fill[red] (\x+0.5,3.5) circle (0.25);
}
\foreach \x in {16}
{\draw[fill=white,white]  (\x,3) -- (\x+2,3) -- (\x+2,4) -- (\x,4) -- cycle;
\draw[-] (\x,3) -- (\x,4);
\draw[-] (\x,3) -- (\x+2,3);
\draw[-] (\x,4) -- (\x+2,4);
\draw[-] (\x+1,3) -- (\x+1,4);
\node at (\x+1.5,3.5) {$\cdots$};
}
\foreach \x in {0}
{\draw[fill=white,white]  (\x,3) -- (\x-2,3) -- (\x-2,4) -- (\x,4) -- cycle;
\draw[-] (\x,3) -- (\x,4);
\draw[-] (\x,3) -- (\x-2,3);
\draw[-] (\x,4) -- (\x-2,4);
\draw[-] (\x-1,3) -- (\x-1,4);
\node at (\x-1.5,3.5) {$\cdots$};
}
}
\caption{A Box-Ball State}\label{firstbbsexfordef}
\end{figure}

Suppose at time $t$, one has $N$ blocks. Let $Q_1^t$, $Q_2^t$, $\ldots$, $Q_N^t$ denote the lengths of these blocks, taken from left to right. Let $W_1^t$, $W_2^t$, $\ldots$, $W_{N-1}^t$ denote the lengths of the sets of empty boxes between the $N$ sets of filled boxes, again taken from left to right. Lastly, let $W_0^t$ and $W_N^t$ be formally defined to be $\infty$, reflecting the fact that the empty boxes continue infinitely in both directions.

\begin{thm} \cite{bib:tokihiro}\label{thmbbscoords2}\index{Box-Ball Coordinate Dynamics}
The coordinates $(W_0^t,Q_1^t,W_1^t,\ldots,Q_N^t,W_N^t)$ evolve under the box ball dynamics according to
\begin{align}
W_0^{t+1}&=W_N^{t+1}=\infty,\\
W_i^{t+1}&=Q_{i+1}^t+W_i^t-Q_i^{t+1},~~~~~~~~~~~~~~~~~~~i=1,\ldots,N-1,\label{Witplusoneeqnbbs}\\
Q_i^{t+1}&=\min\left(W_i^{t},\sum_{j=1}^i Q_j^t-\sum_{j=1}^{i-1}
Q_j^{t+1}\right),~~~~~i=1,\ldots,N.\label{Qitplusoneeqnbbs}
\end{align}
\end{thm}

\begin{ex}
Starting with the initial state in Figure \ref{firstbbsexfordef} the next iteration is shown in Figure \ref{thmbbscoords}:
\begin{figure}[H]
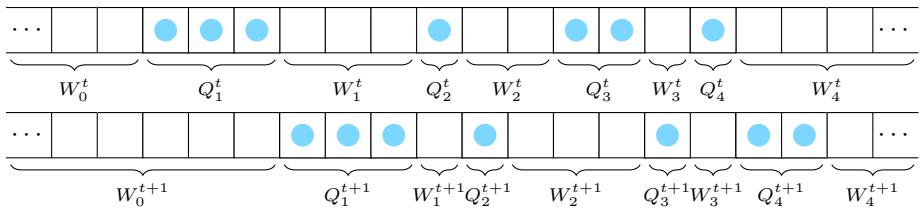

\centering
\tikz[scale=0.6]{
\foreach \x in {0,1,2,3,4,5,6,7,8,9,10,11,12,13,14,15}
{\draw[fill=white]  (\x,3) -- (\x+1,3) -- (\x+1,4) -- (\x,4) -- cycle;			
}
\foreach \x in {1,2,3,7,10,11,13}
{\draw[fill=white]  (\x,3) -- (\x+1,3) -- (\x+1,4) -- (\x,4) -- cycle;			
\fill[cyan] (\x+0.5,3.5) circle (0.25);
}
\foreach \x in {}
{\draw[fill=white]  (\x,3) -- (\x+1,3) -- (\x+1,4) -- (\x,4) -- cycle;			
\fill[red] (\x+0.5,3.5) circle (0.25);
}
\foreach \x in {16}
{\draw[fill=white,white]  (\x,3) -- (\x+2,3) -- (\x+2,4) -- (\x,4) -- cycle;
\draw[-] (\x,3) -- (\x,4);
\draw[-] (\x,3) -- (\x+2,3);
\draw[-] (\x,4) -- (\x+2,4);
\draw[-] (\x+1,3) -- (\x+1,4);
\node at (\x+1.5,3.5) {$\cdots$};
}
\foreach \x in {0}
{\draw[fill=white,white]  (\x,3) -- (\x-2,3) -- (\x-2,4) -- (\x,4) -- cycle;
\draw[-] (\x,3) -- (\x,4);
\draw[-] (\x,3) -- (\x-2,3);
\draw[-] (\x,4) -- (\x-2,4);
\draw[-] (\x-1,3) -- (\x-1,4);
\node at (\x-1.5,3.5) {$\cdots$};
}
\draw [decorate,decoration={brace,amplitude=4pt}] (0.9,2.85) -- (-1.9,2.85) node [black,midway,yshift=-0.4cm] {\scriptsize{$W_0^t$}};
\draw [decorate,decoration={brace,amplitude=4pt}] (3.9,2.85) -- (1.1,2.85) node [black,midway,yshift=-0.4cm] {\scriptsize{$Q_1^t$}};
\draw [decorate,decoration={brace,amplitude=4pt}] (6.9,2.85) -- (4.1,2.85) node [black,midway,yshift=-0.4cm] {\scriptsize{$W_1^t$}};
\draw [decorate,decoration={brace,amplitude=4pt}] (7.9,2.85) -- (7.1,2.85) node [black,midway,yshift=-0.4cm] {\scriptsize{$Q_2^t$}};
\draw [decorate,decoration={brace,amplitude=4pt}] (9.9,2.85) -- (8.1,2.85) node [black,midway,yshift=-0.4cm] {\scriptsize{$W_2^t$}};
\draw [decorate,decoration={brace,amplitude=4pt}] (11.9,2.85) -- (10.1,2.85) node [black,midway,yshift=-0.4cm] {\scriptsize{$Q_3^t$}};
\draw [decorate,decoration={brace,amplitude=4pt}] (12.9,2.85) -- (12.1,2.85) node [black,midway,yshift=-0.4cm] {\scriptsize{$W_3^t$}};
\draw [decorate,decoration={brace,amplitude=4pt}] (13.9,2.85) -- (13.1,2.85) node [black,midway,yshift=-0.4cm] {\scriptsize{$Q_4^t$}};
\draw [decorate,decoration={brace,amplitude=4pt}] (17.9,2.85) -- (14.1,2.85) node [black,midway,yshift=-0.4cm] {\scriptsize{$W_4^t$}};
}
\tikz[scale=0.6]{
\foreach \x in {0,1,2,3,4,5,6,7,8,9,10,11,12,13,14,15}
{\draw[fill=white]  (\x,3) -- (\x+1,3) -- (\x+1,4) -- (\x,4) -- cycle;			
}
\foreach \x in {4,5,6,8,12,14,15}
{\draw[fill=white]  (\x,3) -- (\x+1,3) -- (\x+1,4) -- (\x,4) -- cycle;			
\fill[cyan] (\x+0.5,3.5) circle (0.25);
}
\foreach \x in {}
{\draw[fill=white]  (\x,3) -- (\x+1,3) -- (\x+1,4) -- (\x,4) -- cycle;			
\fill[red] (\x+0.5,3.5) circle (0.25);
}
\foreach \x in {16}
{\draw[fill=white,white]  (\x,3) -- (\x+2,3) -- (\x+2,4) -- (\x,4) -- cycle;
\draw[-] (\x,3) -- (\x,4);
\draw[-] (\x,3) -- (\x+2,3);
\draw[-] (\x,4) -- (\x+2,4);
\draw[-] (\x+1,3) -- (\x+1,4);
\node at (\x+1.5,3.5) {$\cdots$};
}
\foreach \x in {0}
{\draw[fill=white,white]  (\x,3) -- (\x-2,3) -- (\x-2,4) -- (\x,4) -- cycle;
\draw[-] (\x,3) -- (\x,4);
\draw[-] (\x,3) -- (\x-2,3);
\draw[-] (\x,4) -- (\x-2,4);
\draw[-] (\x-1,3) -- (\x-1,4);
\node at (\x-1.5,3.5) {$\cdots$};
}
\draw [decorate,decoration={brace,amplitude=4pt}] (3.9,2.85) -- (-1.9,2.85) node [black,midway,yshift=-0.4cm] {\scriptsize{$W_0^{t+1}$}};
\draw [decorate,decoration={brace,amplitude=4pt}] (6.9,2.85) -- (4.1,2.85) node [black,midway,yshift=-0.4cm] {\scriptsize{$Q_1^{t+1}$}};
\draw [decorate,decoration={brace,amplitude=4pt}] (7.9,2.85) -- (7.1,2.85) node [black,midway,yshift=-0.4cm] {\scriptsize{$W_1^{t+1}$}};
\draw [decorate,decoration={brace,amplitude=4pt}] (8.9,2.85) -- (8.1,2.85) node [black,midway,yshift=-0.4cm] {\scriptsize{$~Q_2^{t+1}$}};
\draw [decorate,decoration={brace,amplitude=4pt}] (11.9,2.85) -- (9.1,2.85) node [black,midway,yshift=-0.4cm] {\scriptsize{$W_2^{t+1}$}};
\draw [decorate,decoration={brace,amplitude=4pt}] (12.9,2.85) -- (12.1,2.85) node [black,midway,yshift=-0.4cm] {\scriptsize{$Q_3^{t+1}$}};
\draw [decorate,decoration={brace,amplitude=4pt}] (13.9,2.85) -- (13.1,2.85) node [black,midway,yshift=-0.4cm] {\scriptsize{$~W_3^{t+1}$}};
\draw [decorate,decoration={brace,amplitude=4pt}] (15.9,2.85) -- (14.1,2.85) node [black,midway,yshift=-0.4cm] {\scriptsize{$Q_4^{t+1}$}};
\draw [decorate,decoration={brace,amplitude=4pt}] (17.9,2.85) -- (16.1,2.85) node [black,midway,yshift=-0.4cm] {\scriptsize{$W_4^{t+1}$}};
}
\caption{The box-ball coordinates on a box-ball system and its time evolution. \label{thmbbscoords}}
\end{figure}
\end{ex}

\subsection{Lusztig Factorization}\label{sec:lusztigfactors}

A matrix  $L \in N_-$ is defined to be {\it totally positive} with respect to $N_-$ if every minor that does not identically vanish on $N_-$ has a positive value. We refer the reader to  \cite{bib:fz} for more background on this. 

We introduce here an elegant parametrization of the totally positive matrices within $N_-$ (denoted $N^{>0}_-$), in terms of negative simple roots of the associated Lie algebra, that is due to Lusztig \cite{bib:lu}. In explicit terms this is given by a factorization of the form 

\begin{eqnarray} \label{LusFac}
L &=& (1 + \alpha_1 f_{h_1}) \cdots (1 + \alpha_M f_{h_M})
\end{eqnarray}
where $L \in N^{>0}_-$, $M = {n \choose 2}$, $h_j \in \{1, \dots, n \}$, $\alpha_j \in \mathbb{R}_{>0}$, 1 denotes the identity matrix and $f_i$ is the elementary lower matrix with 1 in the 
$(i+1, i)$ entry and zero elsewhere. For later use, we set
$$
\ell_i(\alpha) \doteq 1 + \alpha f_i = \left[\begin{array}{ccccccc}
1 \\
& 1 \\
& & \ddots  \\
&&& 1 \\
&&&\alpha&1 \\
&&&&& \ddots \\
&&&&&& 1
\end{array}\right].
$$
Let $w_0$ denote the longest permutation in the symmetric group, $\frak{S}_n$. Let $s_h$ denote the consecutive transposition in $\frak{S}_n$ that interchanges $h$ and $h+1$. Then $s_h$ are generators for $\frak{S}_n$ and so $w_0$ may be written as a word in these generators,

The minimal length of such a word for $w_0$ is $M$ and so such an expression of the form 
$$
w_0 = s_{h_1} \cdots s_{h_M}
$$
will be referred to as a reduced word decomposition of $w_0$. Any other permutation will have representations of shorter length than $M$ and there is only one permutation having minimal length $M$; hence, its reference as the {\it longest} permutation. For further background on these matters we refer the reader to \cite{bib:bfz}.
The permutation matrix corresponding to $w_0$ is

$$\widehat{w}_0 =\renewcommand{\arraystretch}{0.9} \left[\begin{array}{ccccc}
&&&&1\\
&&&1\\
&&1\\
&\iddots\\
1
\end{array}\right],$$
Now let ${\bf h} = (h_1, \dots h_M)$ denote a reduced word decomposition, $w_0 = s_{h_1} \cdots s_{h_M}$. For this article we will fix and always use the following particular reduced word decomposition
\begin{eqnarray}\label{longestwordrepn}
{\bf h^0} &=& (1,2,...,n-1,1,2,...,n-2,..., 3,2,1,2,1).
\end{eqnarray}
One then has
\begin{thm}\label{theorem:bfz} \cite{bib:bfz}
Through the correspondence (\ref{LusFac}), ${\bf h^{0}}$ gives a bijection $\mathbb{R}^M_{>0} \to N^{>0}_-$ between the set of $M$-tuples of positive real numbers, $\left(\alpha_1, \dots, \alpha_M\right)$, and the variety of totally positive lower-unipotent matrices.
\end{thm}

\subsection{Loewner-Whitney Theorem}
\label{sec:lw}
Extending what was described in the previous subsection, one says that a square matrix is {\it totally non-negative} (TNN) if all of its minors are non-negative.
We will make use of the following important characterization of TNN  matrices.
\begin{thm} \cite{bib:fz}
Any invertible TNN matrix is a product of elementary  tridiagonal matrices with non-negative matrix entries.
\end{thm}
By {\it elementary tridiagonal} here one means a matrix that differs from the identity $\mathbb{I}$ in a single entry located either on the main diagonal or immediately above or below it.  
This has the following useful refinement
\begin{cor} \label{cor:lw}\cite{bib:fz} (Theorem 14) 
Any TNN matrix can be written as a product of the form $LDU$ where $L, D, U$ are respectively lower unipotent, diagonal, upper unipotent and TNN in their respective senses.
\end{cor}
In particular, for this paper we will often be interested in the case of matrices having the form stated in the above corollary but where $DU$ is upper bidiagonal with all superdiagonal entries being 1 and all diagonal entries being positive. We will refer to TNN matrices, $LDU$, with $DU$ having this restricted form as {\it TNN Hessenberg matrices}. 

\section{dFToda as Recursive dToda}\label{section:dfktodarecdtoda}

We will now apply Lusztig 
factorization to analyze the dynamics of dFToda. This presumes that the total positivity conditions of Theorem  \ref{theorem:bfz} are preserved under this dynamics. That is true but we defer the proof of this to Section \ref{sec:wp}. 

One may regard Theorem \ref{theorem:bfz} as a unique factorization theorem for elements of $N^{>0}_-$ which may be presented in the form of a bidiagonal factorization
\begin{eqnarray*}
L &=& T_1(u)T_2(u)\cdots T_{n-1}(u)\,\,\, \mbox{where for}\\
u &=& (u_i^j)_{1\leq i\leq j\leq n-1}, \\ 
T_i(u) &=& \ell_1(u_1^i)\ell_2(u_2^{i+1})\cdots \ell_{n-i}(u_{n-i}^{n-1}).
\end{eqnarray*}
$T_i(u)$ has the block bidiagonal form
$$\left[\begin{array}{ccccc|c}
1&&&&\\u_1^i & 1&&&\\&u_2^{i+1}&\ddots&&\\ &&\ddots &1&\\&&&u_{n-i}^{n-1}&1&\\\hline
&&&&&\mathbb{I}_{i-1}
\end{array}\right]$$
where $\mathbb{I}_{i-1}$ is the $(i-1)\x (i-1)$ identity matrix.\\

\n Starting with dFToda,  as in Equation \eqref{symeseqn}, one takes time-dependent Lusztig parameters for $L(t)$ via:
\begin{eqnarray} \label{Tfact}
L(t)=T_1(u(t))T_2(u(t))\cdots T_{n-1}(u(t))
\end{eqnarray}
according to the factorization in Theorem \ref{theorem:bfz}. Then, after a time-step of dFToda, one has 
$$L(t+1)R(t+1)=R(t)L(t)$$
with $L(t+1)$ endowed with its own Lusztig factorization:
$$L(t+1)=T_1(u(t+1))T_2(u(t+1))\cdots T_{n-1}(u(t+1)).$$

\begin{lem}\label{lemma:blockpreservationsymes}
For each $i\in \{1,\ldots,n-1\}$, $\alpha\in \R_{>0}^{i}$ and $d\in \R_{>0}^n$, if $R=(\text{diag}(d)+\epsilon)$ and  $T=\ell_1(\alpha_1)\ell_2(\alpha_2)\cdots \ell_i(\alpha_i)$ are such that $RT$ has an LU-decomposition, then the lower part of this LU-decomposition has the same structure as $T$ and the upper part has the same structure as $R$.
\end{lem}
\begin{proof}
One has
$$T=\left[\begin{array}{ccccc|c}
1&&&&\\\alpha_1 & 1&&&\\&\alpha_2&\ddots&&\\ &&\ddots &1&\\&&&\alpha_i&1&\\\hline
&&&&&\mathbb{I}_{n-i-1}
\end{array}\right],~~~\text{and}~~~R=\left[\begin{array}{cccc|cccc}
d_1 & 1 & &&&\\
&d_2 & \ddots &  & &&\\
&&\ddots & 1 &  & &\\
&&&d_{i+1} & 1 &  &\\
\hline
&&&&d_{i+2} & 1\\
&&&&&d_{i+3} & \ddots\\
&&&&&&\ddots & 1\\
&&&&&& & d_n
\end{array}\right].$$
The bottom-left block of $RT$ is the $(n-i-1)\times (i+1)$ zero matrix, and the bottom-right block is the same as the bottom right block of $R$. Since the product $RT$ is evidently a tridiagonal matrix with ones on its superdiagonal, its LU decomposition, if it exists, is a product of a lower bidiagonal unipotent matrix and an upper bidiagonal matrix with ones on the superdiagonal. The zero in the $(i+2,i+1)$ entry of $RT$ and the fact that the diagonal entries of the upper factor of $RT$ must be nonzero (because the determinant does not vanish) forces the lower factor of $RT$  to decouple into a block diagonal matrix of two lower bidiagonal unipotent matrices. This then forces the lower-right block of the lower part of $RT$ to be the identity matrix.
\end{proof}

\begin{defn} \label{def:dft}
Define a sequence of upper bidiagonal matrices $(R^{(i)}(t))_{i=0}^{n-1}$ associated to the pair $R(t)$ and $L(t)$ via
\begin{enumerate}[(i)]
\item $R^{(0)}(t)=R(t)$
\item $T_i(u(t+1))R^{(i)}(t)=R^{(i-1)}(t)T_i(u(t))$ for $i=1,\ldots,n-1$.
\end{enumerate}
For convenience, let $H^{(i)}(t)$ denote the tridiagonal Hessenberg matrix of interest in (ii), \textit{i.e.},
\begin{equation}\label{defnhessenbergtruncs}
H^{(i)}(t):=R^{(i-1)}(t)T_i(t).\end{equation}
\end{defn}

\begin{thm}\label{remark:substepsdtoda}
Defining $L(t)=T_1(u(t))T_2(u(t))\cdots T_{n-1}(u(t))$, $R^{(0)}(t)=R(t)$ and $R^{(n-1)}(t)=R(t+1)$, Definition \ref{def:dft} gives the single time-step of dFToda:
$$L(t+1)R(t+1)=R(t)L(t),$$
with each application of Definition \ref{def:dft} (ii) reducing to a dToda time step of decreasing dimension.
\end{thm}

\begin{proof}
By Lemma \ref{lemma:blockpreservationsymes}, one has a factorization $T_i(u(t+1))$ in part (ii) with $u(t+1)$ replaced with some other tuple $v$ since the block structure is preserved by the (tridiagonal) dToda map. This definition would then give $$R(t)L(t)=T_1(v)T_2(v)\cdots T_{n-1}(v)R^{(n-1)}(t).$$
Thus, $T_1(v)T_2(v)\cdots T_{n-1}(v)$ and $R^{(n-1)}(t)$ are the lower and upper parts of the $LU$ decomposition of $R(t)L(t)$. Hence, $L(t+1)=T_1(v)T_2(v)\cdots T_{n-1}(v)$ and $R(t+1)=R^{(n-1)}(t)$. However, by the uniqueness of Lusztig factorization (Theorem \ref{theorem:bfz}), one must then have $v=u(t+1)$, which is why (ii) is presented in this form. This means that the recursive dToda map applications on pairs $R^{(i-1)}(t)$ and $T_i(u(t))$ combine precisely to give the dFToda evolution on the Lusztig parameters.\\[3pt]
Now we show that Definition \ref{def:dft} (ii) reduces to a dToda on the top-left $(n-i+1)\x (n-i+1)$ block. One sees this by defining $(I_j^{t,i})_{j=1}^n$ to be the diagonal entries of $R^{(i)}(t)$ for $i=0,\ldots,n-1$. Thus, 
$$I_j^{t,0}=I_j^t,~~~~I_j^{t,n-1}=I_j^{t+1}$$
in the usual notation of dToda, and one then has 
$$T_i(u(t+1))=\left[\begin{array}{ccccc|c}
1&&&&\\u_1^i(t+1) & 1&&&\\&u_2^{i+1}(t+1)&\ddots&&\\ &&\ddots &1&\\&&&u_{n-i}^{n-1}(t+1)&1&\\\hline
&&&&&\mathbb{I}_{i-1}
\end{array}\right]$$
and
$$R^{(i)}(t)=
\left[\begin{array}{cccc|cccc}
I_1^{t,i} & 1&&&&&&\\
&I_2^{t,i} & \ddots&&&&&\\
&&\ddots & 1&&\\
&& & I_{n-i+1}^{t,i}&1&\\\hline
&&& & I_{n-i+2}^{t,i-1}&1&\\
&&&& & I_{n-i+3}^{t,i-1}&\ddots&\\
&&&&& & \ddots&1\\
&&&&& & &I_{n}^{t,i-1}\\
\end{array}\right],$$
observing that the lower-right block of $R^{(i)}(t)$ is equal to the lower-right block of $R^{(i-1)}(t)$ by Lemma \ref{lemma:blockpreservationsymes}. Furthermore, because of this block structure, the top-left blocks are determined precisely by one step of the $(n-i+1)\x (n-i+1)$ dToda map:
\begin{equation}\label{eqn:dfktodasubstep}
    \left(R^{(i)}(t)\right)_{[n-i+1],[n-i+1]}\left(T_i(u(t+1))\right)_{[n-i+1],[n-i+1]}
=
\left(T_i(u(t))\right)_{[n-i+1],[n-i+1]}
\left(R^{(i-1)}(t)\right)_{[n-i+1],[n-i+1]}
\end{equation}where, for $S_1,S_2\subset [n]$, the notation $(M)_{S_1,S_2}$ picks out the submatrix of $M$ with row indices in $S_1$ and column indices in $S_2$. \textit{i.e.}
$$
\left[\begin{array}{cccc}
I_1^{t,i-1} & 1\\
& I_2^{t,i-1} & \ddots\\
& & \ddots & 1\\
& &  & I_{n-i+1}^{t,i-1}\\
\end{array}\right]
\left[\begin{array}{ccccc}
1\\
u_1^i(t) & \ddots\\
&\ddots & 1\\
&&u_{n-i}^{n-1}(t) & 1\\
\end{array}\right]$$
$$=$$
$$
\left[\begin{array}{ccccc}
1\\
u_1^i(t+1) & \ddots\\
&\ddots & 1\\
&&u_{n-i}^{n-1}(t+1) & 1\\
\end{array}\right]
\left[\begin{array}{cccc}
I_1^{t,i} & 1\\
& I_2^{t,i} & \ddots\\
& & \ddots & 1\\
& &  & I_{n-i+1}^{t,i}\\
\end{array}\right]
$$
which shows that each of the sub-steps of the algorithm in Definition \ref{def:dft} is indeed a dToda map, each time of lower dimension than the last.
\end{proof}

\subsection{Well-posedness of dFToda} \label{sec:wp}

The prior discussion in this section assumes that it is always possible to carry out the required factorizations.We now show that under some reasonable positivity assumptions these factorizations do always exist. 

\begin{thm} \label{thm:wp}
Assume that the principal minors of $X(t)$ are all positive. Then $X(t)$ necessarily has an $LU$ factorization, $X(t) = L(t) R(t)$. Assume further that $L(t)$ is totally positive with factorization (\ref{Tfact}). Then the principal minors of $X(t+1)$ are also positive, so that the factorization $X(t+1) = L(t+1) R(t+1) $ is defined. Furthermore, $L(t+1)$ is totally positive with a unique factorization of the form (\ref{Tfact}) whose coefficients are explicitly given by 
\begin{eqnarray} \label{tauf}
u_j^{i+j-1}(t+1) &=& u_j^{i+j-1}(t) \frac{\tau_{j+1}^{i-1}(t) \tau_{j-1}^i(t)}{\tau_{j}^{i-1}(t) \tau_{j}^i(t)} \,\,\, \text{with} \\ \label{taus}
\tau_{j}^i(t) & = &  \tau_j(H^{(i)}(t))=~~ \tau_j\left( R(t) T_1(u(t)) T_2(u(t)) \cdots T_i(u(t)) \right)\label{taufundefnithhess}
\end{eqnarray}
where $\tau_j(M)$ denotes the $j^{th}$ principal minor of the square matrix $M$. (For the sake of brevity, unless it is unclear from context, the time dependence of the $\tau$ functions may not be explicitly stated.)
\end{thm}

\begin{proof}
The idea of the proof will be to decompose a time step of the dFToda dynamics into a coupled sequence of dToda time steps and then inductively apply  (\ref{explicitDT}) at each of the dToda sub-steps. For the base step of the induction consider, from Definition \ref{def:dft}, the first stage
\begin{eqnarray*}
T_1(u(t+1))R^{(1)}(t) &=& R^{(0)}(t)T_1(u(t))\\
&=& R(t)T_1(u(t)).
\end{eqnarray*}
By the assumption that the principal submatrices of $X(t)$ have positive determinants (denoted by $\tau^0_j(t)$ for the $j^{th}$ minor -- notation consistent with (\ref{taus}) -- with $\tau^0_0(t) \equiv 1$ ) it follows directly \cite{bib:strang} that 
\begin{eqnarray} \label{R-repn}
R(t) &=& \left[\begin{array}{cccc}
\tau^0_1/\tau^0_0 & 1\\
&\tau^0_2/\tau^0_1  & \ddots\\
&&\ddots & 1\\
&&&\tau^0_n/\tau^0_{n-1} 
\end{array}\right]
\end{eqnarray}
is well defined and comparing to (\ref{ltrtdtoda}) we set 
\begin{eqnarray}
I^t_j = \tau^0_j/\tau^0_{j-1}
\end{eqnarray}
and 
$$
T_1(t) = \arraycolsep=3.1pt\def\arraystretch{1.5}\left[\begin{array}{cccc}
1\\
V_1^t&1\\&\ddots&\ddots\\
&&V_{n-1}^t&1
\end{array}\right],
$$
so that
\begin{eqnarray}
V^t_j = u^j_j.
\end{eqnarray}
The assumptions of the theorem imply that 
$$
I^t_j > 0,\,\,\, j = 1, \dots, n \qquad V^t_j > 0,\,\,\,  j = 1, \dots, n-1 .
$$
One may then conclude that for 

\begin{equation}
T_1(u(t+1))=\arraycolsep=3.1pt\def\arraystretch{1.5}\left[\begin{array}{cccc}
1\\
V_1^{t+1}&1\\&\ddots&\ddots\\
&&V_{n-1}^{t+1}&1
\end{array}\right],~~~~\text{and}~~~~R^{(1)}(t)=\arraycolsep=4.4pt\def\arraystretch{1.3}\left[\begin{array}{cccc}
I_1^{t+1} & 1\\
&I_2^{t+1}&\ddots\\
&&\ddots&1\\
&&&I_n^{t+1}
\end{array}\right],
\end{equation}
due to the subtraction free form of \eqref{explicitDT}, one has that

\begin{eqnarray*}
I^{t+1}_1 &=& V^t_1 + I^{t}_1 > 0\\
I^{t+1}_2 &=& V^t_2 + \frac{I^{t}_2 I^{t}_1}{I^{t+1}_1} > 0\\
& \vdots & \\
I^{t+1}_{n-1} &=& V^t_{n-1} + \frac{I^{t}_{n-1} \cdots I^{t}_2}{I^{t+1}_{n-2} \cdots I^{t+1}_1} > 0\\
I^{t+1}_{n} &=& \frac{I^{t}_{n} \cdots I^{t}_2}{I^{t+1}_{n-1} \cdots I^{t+1}_1} > 0\\
V^{t+1}_i &=& \frac{I^{t}_{i+1}}{I^{t+1}_i} V^{t}_i > 0\\
V^{t+1}_n &=& 0.
\end{eqnarray*}
In summary, one has

$$
I^{t+1}_j > 0,\,\,\, j = 1, \dots, n \qquad V^{t+1}_j > 0,\,\,\,  j = 1, \dots, n-1 .
$$

\n Hence, the principal submatrices of $T_1(u(t+1))R^{(1)}(t)$ all have positive determinant and $T_1(u(t+1))$ is a product of elementary positive Lusztig factors in standard form. 

We are now in a position to claim the induction step; namely, that if the principal minors of $R^{(i-1)}(t)$ are all positive, and if $T_i(u(t))$ is a product of elementary positive Lusztig factors in standard form, then the same holds for $R^{(i)}(t)$ and $T_i(u(t+1))$, respectively.  The argument for this proceeds exactly as for the base case just discussed except that, as a consequence of Lemma \ref{lemma:blockpreservationsymes}, at the $i^{th}$ stage (\ref{explicitDT}) is replaced by

\begin{equation*} 
\left\{\renewcommand{\arraystretch}{1.6}\setlength{\tabcolsep}{16pt}
\begin{array}{lcl}
V_0^t= V^{t}_{n-i+1}= \cdots = V_n^t=0,\\
I_j^{t+1}=V_j^t+\dfrac{I_j^t\cdots I_{1}^t}{I_{j-1}^{t+1}\cdots I_{1}^{t+1}}&&j=1,\ldots,n,\\
V_j^{t+1}I_j^{t+1}=I_{j+1}^tV_j^t &~~~~~~~& j=1,\ldots,n-i.
\end{array}
\right.
\end{equation*}
Consequently at the $i^{th}$ stage one has

$$
I^{t+1}_j > 0,\,\,\, j = 1, \dots, n \qquad V^{t+1}_j > 0,\,\,\,  j = 1, \dots, n-i .
$$

 At the end of this inductive process one has determined an $L(t+1)$ with standard Lusztig factorization which must therefore be totally positive by Lusztig's theorem (Theorem \ref{theorem:bfz}) as well as
$R(t+1) = R^{(n-1)}(t)$ with positive diagonal entries. Hence, $X(t+1)$ has principal minors with positive determinant.

The derivation of formulae (\ref{tauf}) and (\ref{taus}) along with the explicit representation

\begin{eqnarray*}
R^{(i)}(t) &=& \left[\begin{array}{cccc}
\tau^{i}_1/\tau^{i}_0 & 1\\
&\tau^{i}_2/\tau^{i}_1  & \ddots\\
&&\ddots & 1\\
&&&\tau^{i}_n/\tau^{i}_{n-1} 
\end{array}\right]
\end{eqnarray*}
is deferred to Appendix \ref{appA}. 
\end{proof}

\begin{cor} \label{cor:wp}
Let $\mathcal{B}$ denote the subvariety of $\mathcal{H}$ comprised  of upper bidiagonal matrices of the form 
\begin{eqnarray*}
\left[\begin{array}{cccc}
* & 1\\
& *  & \ddots\\
&&\ddots & 1\\
&&& * 
\end{array}\right]
\end{eqnarray*}
and let $\mathcal{B}^{>0}$ denote the submanifold in which all diagonal entries are positive. The dFToda map is defined everywhere on the restricted phase space \begin{eqnarray} \label{resphase}
(N_-^{>0} \times \mathcal{B}^{>0}~) \subset \mathcal{H}
\end{eqnarray}
and this space is preserved under all forward iterates $ (t \in \mathbb{N} \cup\{0\}) $ of dFToda. 
\end{cor}
\n In the remainder of this paper we denote $(N_-^{>0} \times \mathcal{B}^{>0})$ by 
$\mathcal{H}^{>0}$ and refer to it as the subvariety of totally positive Hessenberg matrices. Indeed, (\ref{resphase}) is totally positive with respect to the space of Hessenberg matrices in that every minor that does not identically vanish on $\mathcal{H}$ must be positive. Moreover, this space is dense in the space of TNN Hessenberg matrices as defined in Section \ref{sec:lw}. Going forward we will denote this latter space by $\mathcal{H}^{\geq0}$. 

\subsection{Extensions and Reductions} \label{ExtRed}

We conclude this section by describing an extension of the dFToda dynamics and generalizations of its phase space. 

\begin{prop}
The dFToda dynamics can be defined for all time $t \in \mathbb{Z}$. The extension to negative time, $t \in \{ -1, -2, -3, \dots \}$ is given by UL factorization: if $X(t) = R(t)L(t)$ then define
\begin{eqnarray} \nonumber
X(t - 1) &=& L(t)R(t)\\ \label{reverse}
&=& L(t) X(t) L(t)^{-1}.
\end{eqnarray}
This is well-defined for all negative time.
\end{prop}
\begin{proof}
This UL factorization clearly constitutes a reverse flow consistent with that of the Symes evolution. This flow exists for all $t$ by the same arguments applied in Theorem \ref{thm:wp} and Corollary \ref{cor:wp}. This follows because the equations for the backwards time flow (\ref{reverse}) have an entirely similar subtraction-free form:
\begin{equation} \label{backwardDT}
\left\{\renewcommand{\arraystretch}{1.6}\setlength{\tabcolsep}{16pt}
\begin{array}{lcl}
V_0^t=V_n^t=0,\\
I_i^{t-1}=V_i^t+\dfrac{I_n^tI_{n-1}^t\cdots I_{i}^t}{I_{n}^{t-1}I_{n-1}^{t-1}\cdots I_{i+1}^{t-1}}&&i=1,\ldots,n,\\
V_i^{t-1}=\dfrac{I_{i}^tV_i^t}{I_{i+1}^{t-1}} &~~~~~~~& i=1,\ldots,n-1.
\end{array}
\right.
\end{equation}
\end{proof}


\begin{defn}
For $d\in\{1,2,\ldots,n\}$, let $\mathcal{H}^d$ be the subset of $\mathcal{H}$ for which entries $h_{ij}=0$ whenever $i-j>d$, \textit{i.e.}, this is the set of lower Hessenberg matrices that are zero below the first $d$ sub-diagonals. ($\mathcal{H}^1$ is the subspace we have referred to as tridiagonal Hessenberg matrices.)  Let $N_-^d$ denote elements of $N_-$ for which $n_{ij}=0$ whenever $i-j>d$ and $N_-^{d, > 0}$ the set of totally positive elements in $N_-^d$, \textit{i.e.}, those elements whose minors that do not identically vanish on $N_-^d$ have a positive value.
\end{defn}

\begin{prop}
The dFToda map is defined everywhere on the subspace $(N_-^{d,>0} \times \mathcal{B}^{>0}) \subset \mathcal{H}^d$ and this space is preserved under all iterations, for $t \in \mathbb{Z}$, of the map. For $d<n$ we may refer to this as a {\rm band-limited} sub-system of dFToda. 
\end{prop}
 \begin{proof} If an initial condition lies in $(N_-^{d,>0} \times \mathcal{B}^{>0})$ then (\ref{tauf}) still holds and it follows from this that iterates remain in this band-limited subspace. 
(Note that from (\ref{taus}), by induction the tau factors appearing in (\ref{tauf}) are non-vanishing.) With this subspace invariance in place, the other arguments that were applied for $d = n$ continue to hold, but terminate 
at the $d^{th}$ stage.
 \end{proof}

\section{Ultradiscrete Full Toda: tropicalizing Lusztig Parameters} \label{sec:ud}
We now begin the process of ultradiscretising the recursive dToda description of dFToda given in Section \ref{section:dfktodarecdtoda}.
\begin{defn}
Define variables $(Q_j^{t,i}(\epsilon))_{j=1}^{n}$ and $W_j^i(t)(\epsilon)$ via
$$I_j^{t,i}=e^{-\frac{1}{\epsilon}Q_j^{t,i}(\epsilon)},~~~~~~~j=1,\ldots,n,~~i=0,\ldots,n-1,$$
and
$$u_j^i(t)=e^{-\frac{1}{\epsilon}W_j^i(t)(\epsilon)},~~~~~~~1\leq j\leq i\leq n-1.$$
\end{defn}

\n For each $1\leq i\leq n-1$, Equation (\ref{eqn:dfktodasubstep}), 
$$    \left(R^{(i)}(t)\right)_{[n-i+1],[n-i+1]}\left(T_i(u(t+1))\right)_{[n-i+1],[n-i+1]}
=
\left(T_i(u(t))\right)_{[n-i+1],[n-i+1]}
\left(R^{(i-1)}(t)\right)_{[n-i+1],[n-i+1]}$$
gives, via ultradiscretization (see Section \ref{sec:ultradis}), a map from a tuple 
$$(W_1^i(t),W_2^{i+1}(t),\ldots,W_{n-i}^{n-1}(t),\ldots Q_1^{t,i-1},Q_2^{t,i-1},\ldots,Q_{n-i+1}^{t,i-1})$$ 
to 
$$(W_1^i(t+1),W_2^{i+1}(t+1),\ldots,W_{n-i}^{n-1}(t+1),\ldots Q_1^{t,i},Q_2^{t,i},\ldots,Q_{n-i+1}^{t,i}),$$
which, by Theorem \ref{remark:substepsdtoda}, is the tropicalization of a dToda evolution, hence it is interpreted as an $(n-i+1)$-soliton box-ball evolution. The precise correspondence here is that for each step $i$:

\begin{itemize}
    \item Prior to the BBS evolution, the $j$-th block of balls consists of $Q_j^{t,i-1}$ balls, and the gap between the $j$-th and $(j+1)$-st block of balls is of length $W_j^{i+j-1}(t)$.
    \item After the BBS evolution, the $j$-th block of balls consists of $Q_j^{t,i}$ balls, and the gap between the $j$-th and $(j+1)$-st block of balls is of length $W_j^{i+j-1}(t+1)$. 
\end{itemize}

\n Thus, one sees that the overall dFToda  evolution
$$
\left(
(I_j^{t,0})_{j=1}^n,~(u_i^j(t))_{1\leq i\leq j \leq n-1}
\right)
~~\mapsto~~
\left(
(I_j^{t,n-1})_{j=1}^n,~(u_i^j(t+1))_{1\leq i\leq j \leq n-1}
\right)
$$
produces an overall ultradiscretised evolution (which we will call udFToda):
$$
\left(
(Q_j^{t,0})_{j=1}^n,~(W_i^j(t))_{1\leq i\leq j \leq n-1}
\right)
~~\mapsto~~
\left(
(Q_j^{t,n-1})_{j=1}^n,~(W_i^j(t+1))_{1\leq i\leq j \leq n-1}
\right)
$$
whose BBS coordinate interpretation is presented as follows.

\begin{defn}\label{definitionofudftodaevol}
The udFToda system is an evolution on pairs $(\mathbf{Q}(t),\mathbf{W}(t))$ where  $$\mathbf{Q}(t)=(Q_j^t)_{j=1}^n,~~~~~\mathbf{W}(t)=(W_i^j(t))_{1\leq i\leq j\leq n-1},$$
with $Q_j^t\in \N$ and $W_i^j(t)\in \N$. The evolution is given in $n-1$ steps:
\begin{itemize}
    \item The first step is to implement the BBS evolution on $$(\infty,Q_1^t,W_1^1(t),Q_2^t,W_2^2(t),\ldots,W_{n-1}^{n-1}(t),Q_{n}^t,\infty)$$
    and define $W_j^j(t+1)$ to be the length of the $j$-th gap and $Q_{n}^{t+1}$ to be the length of the $(n-1)$-st block of balls after this application of the BBS evolution.
    \item For the $i$-th step (with $i>1$), take the first $n-i+1$ blocks of balls (\textit{i.e.}, all but the last) and separate them with the following sequence of gaps
    $$W_1^i(t),~W_2^{i+1}(t),~\ldots,~W_{n-i}^{n-1}(t).$$
    Perform the BBS evolution on this box-ball configuration and define $W_j^{i+j-1}(t+1)$ to be the length of the $j$-th gap and $Q_{n-i+1}^{t+1}$ to be the length of the last block of balls after this application of the BBS evolution.
\end{itemize}
After completing steps $1$ through $(n-1)$, all of the time $t+1$ variables 
$$\mathbf{Q}(t+1)=(Q_j^{t+1})_{j=1}^n,~~~~~\mathbf{W}(t+1)=(W_i^j(t+1))_{1\leq i\leq j\leq n-1}$$
have been defined.
\end{defn}

\begin{ex}\label{udftodaexample}
Let us compute a single time step of the udFToda flow for $n=4$:
$$\mathbf{Q}(t)=(2,1,3,1),~~~\mathbf{W}(t)=\begin{array}{ccccc}
&&3&&\\&1&&1&\\2&&1&&2
\end{array}.$$
For convenience, we write $\mathbf{W}(t)$ as compactly as a triangular array whose bottom (first) row is $(W_1^1(t),W_2^2(t),W_3^3(t))$, whose second row is $(W_1^2(t),W_2^3(t))$, and whose top (third) row is $(W_1^3(t))$.\\
To determine the first row of $\mathbf{W}(t+1)$ and last entry of $\mathbf{Q}(t+1)$, one takes $\mathbf{Q}(t)$ as the block sizes, with gaps given by the bottom row of $\mathbf{W}(t)$, performing a BBS evolution to obtain:
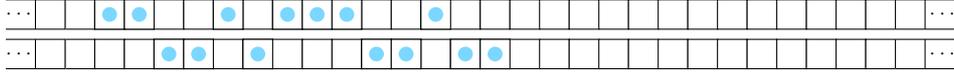
\begin{figure}[H]
\centering
\tikz[scale=0.39]{
\foreach \x in {0,1,2,3,4,5,6,7,8,9,10,11,12,13,14,15,16,17,18,19,20,21,22,23,24,25,26,27}
{\draw[fill=white]  (\x,3) -- (\x+1,3) -- (\x+1,4) -- (\x,4) -- cycle;			
}
\foreach \x in {1,2,5,7,8,9,12}
{\draw[fill=white]  (\x,3) -- (\x+1,3) -- (\x+1,4) -- (\x,4) -- cycle;			
\fill[cyan] (\x+0.5,3.5) circle (0.25);
}
\foreach \x in {}
{\draw[fill=white]  (\x,3) -- (\x+1,3) -- (\x+1,4) -- (\x,4) -- cycle;			
\fill[red] (\x+0.5,3.5) circle (0.25);
}
\foreach \x in {28}
{\draw[fill=white,white]  (\x,3) -- (\x+2,3) -- (\x+2,4) -- (\x,4) -- cycle;
\draw[-] (\x,3) -- (\x,4);
\draw[-] (\x,3) -- (\x+2,3);
\draw[-] (\x,4) -- (\x+2,4);
\draw[-] (\x+1,3) -- (\x+1,4);
\node at (\x+1.5,3.5) {~\scriptsize{$\cdots$}};
}
\foreach \x in {0}
{\draw[fill=white,white]  (\x,3) -- (\x-2,3) -- (\x-2,4) -- (\x,4) -- cycle;
\draw[-] (\x,3) -- (\x,4);
\draw[-] (\x,3) -- (\x-2,3);
\draw[-] (\x,4) -- (\x-2,4);
\draw[-] (\x-1,3) -- (\x-1,4);
\node at (\x-1.5,3.5) {\scriptsize{$\cdots$}};
}
}
\\[3pt]
\tikz[scale=0.39]{
\foreach \x in {0,1,2,3,4,5,6,7,8,9,10,11,12,13,14,15,16,17,18,19,20,21,22,23,24,25,26,27}
{\draw[fill=white]  (\x,3) -- (\x+1,3) -- (\x+1,4) -- (\x,4) -- cycle;			
}
\foreach \x in {3,4,6,10,11,13,14}
{\draw[fill=white]  (\x,3) -- (\x+1,3) -- (\x+1,4) -- (\x,4) -- cycle;			
\fill[cyan] (\x+0.5,3.5) circle (0.25);
}
\foreach \x in {}
{\draw[fill=white]  (\x,3) -- (\x+1,3) -- (\x+1,4) -- (\x,4) -- cycle;			
\fill[red] (\x+0.5,3.5) circle (0.25);
}
\foreach \x in {28}
{\draw[fill=white,white]  (\x,3) -- (\x+2,3) -- (\x+2,4) -- (\x,4) -- cycle;
\draw[-] (\x,3) -- (\x,4);
\draw[-] (\x,3) -- (\x+2,3);
\draw[-] (\x,4) -- (\x+2,4);
\draw[-] (\x+1,3) -- (\x+1,4);
\node at (\x+1.5,3.5) {~\scriptsize{$\cdots$}};
}
\foreach \x in {0}
{\draw[fill=white,white]  (\x,3) -- (\x-2,3) -- (\x-2,4) -- (\x,4) -- cycle;
\draw[-] (\x,3) -- (\x,4);
\draw[-] (\x,3) -- (\x-2,3);
\draw[-] (\x,4) -- (\x-2,4);
\draw[-] (\x-1,3) -- (\x-1,4);
\node at (\x-1.5,3.5) {\scriptsize{$\cdots$}};
}
}
\end{figure}
\n From this, we read off the following:
$$\mathbf{Q}(t+1)=(?,?,?,2),~~~\mathbf{W}(t+1)=\begin{array}{ccccc}
&&?&&\\&?&&?&\\1&&3&&1
\end{array}.$$
Next, to determine the second row of $\mathbf{W}(t+1)$ and the penultimate entry of $\mathbf{Q}(t+1)$, one uses all but the last blocks of the previously obtained BBS (i.e., block lengths of $2$, $1$ and $2$) and separates them with gaps given by the second row of $\mathbf{W}(t)$ to obtain:
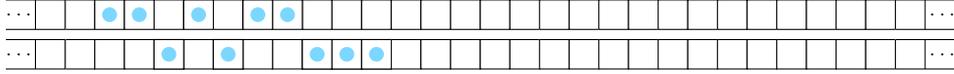
\begin{figure}[H]
\centering
\tikz[scale=0.39]{
\foreach \x in {0,1,2,3,4,5,6,7,8,9,10,11,12,13,14,15,16,17,18,19,20,21,22,23,24,25,26,27}
{\draw[fill=white]  (\x,3) -- (\x+1,3) -- (\x+1,4) -- (\x,4) -- cycle;			
}
\foreach \x in {1,2,4,6,7}
{\draw[fill=white]  (\x,3) -- (\x+1,3) -- (\x+1,4) -- (\x,4) -- cycle;			
\fill[cyan] (\x+0.5,3.5) circle (0.25);
}
\foreach \x in {}
{\draw[fill=white]  (\x,3) -- (\x+1,3) -- (\x+1,4) -- (\x,4) -- cycle;			
\fill[red] (\x+0.5,3.5) circle (0.25);
}
\foreach \x in {28}
{\draw[fill=white,white]  (\x,3) -- (\x+2,3) -- (\x+2,4) -- (\x,4) -- cycle;
\draw[-] (\x,3) -- (\x,4);
\draw[-] (\x,3) -- (\x+2,3);
\draw[-] (\x,4) -- (\x+2,4);
\draw[-] (\x+1,3) -- (\x+1,4);
\node at (\x+1.5,3.5) {~\scriptsize{$\cdots$}};
}
\foreach \x in {0}
{\draw[fill=white,white]  (\x,3) -- (\x-2,3) -- (\x-2,4) -- (\x,4) -- cycle;
\draw[-] (\x,3) -- (\x,4);
\draw[-] (\x,3) -- (\x-2,3);
\draw[-] (\x,4) -- (\x-2,4);
\draw[-] (\x-1,3) -- (\x-1,4);
\node at (\x-1.5,3.5) {\scriptsize{$\cdots$}};
}
}
\\[3pt]
\tikz[scale=0.39]{
\foreach \x in {0,1,2,3,4,5,6,7,8,9,10,11,12,13,14,15,16,17,18,19,20,21,22,23,24,25,26,27}
{\draw[fill=white]  (\x,3) -- (\x+1,3) -- (\x+1,4) -- (\x,4) -- cycle;			
}
\foreach \x in {3,5,8,9,10}
{\draw[fill=white]  (\x,3) -- (\x+1,3) -- (\x+1,4) -- (\x,4) -- cycle;			
\fill[cyan] (\x+0.5,3.5) circle (0.25);
}
\foreach \x in {}
{\draw[fill=white]  (\x,3) -- (\x+1,3) -- (\x+1,4) -- (\x,4) -- cycle;			
\fill[red] (\x+0.5,3.5) circle (0.25);
}
\foreach \x in {28}
{\draw[fill=white,white]  (\x,3) -- (\x+2,3) -- (\x+2,4) -- (\x,4) -- cycle;
\draw[-] (\x,3) -- (\x,4);
\draw[-] (\x,3) -- (\x+2,3);
\draw[-] (\x,4) -- (\x+2,4);
\draw[-] (\x+1,3) -- (\x+1,4);
\node at (\x+1.5,3.5) {~\scriptsize{$\cdots$}};
}
\foreach \x in {0}
{\draw[fill=white,white]  (\x,3) -- (\x-2,3) -- (\x-2,4) -- (\x,4) -- cycle;
\draw[-] (\x,3) -- (\x,4);
\draw[-] (\x,3) -- (\x-2,3);
\draw[-] (\x,4) -- (\x-2,4);
\draw[-] (\x-1,3) -- (\x-1,4);
\node at (\x-1.5,3.5) {\scriptsize{$\cdots$}};
}
}
\end{figure}
\n From this, we read off the following:
$$\mathbf{Q}(t+1)=(?,?,3,2),~~~\mathbf{W}(t+1)=\begin{array}{ccccc}
&&?&&\\&1&&2&\\1&&3&&1
\end{array}.$$
Finally, we can determine the third (and final) row of $\mathbf{W}(t+1)$ as well as the second (and first) entries of $\mathbf{Q}(t+1)$ by taking all but the last blocks of the previously obtained BBS (i.e., block lengths of $1$ and $1$) and separates them with gaps given by the third row of $\mathbf{W}(t)$ to obtain:
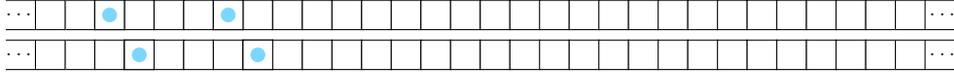
\begin{figure}[H]
\centering
\tikz[scale=0.39]{
\foreach \x in {0,1,2,3,4,5,6,7,8,9,10,11,12,13,14,15,16,17,18,19,20,21,22,23,24,25,26,27}
{\draw[fill=white]  (\x,3) -- (\x+1,3) -- (\x+1,4) -- (\x,4) -- cycle;			
}
\foreach \x in {1,5}
{\draw[fill=white]  (\x,3) -- (\x+1,3) -- (\x+1,4) -- (\x,4) -- cycle;			
\fill[cyan] (\x+0.5,3.5) circle (0.25);
}
\foreach \x in {}
{\draw[fill=white]  (\x,3) -- (\x+1,3) -- (\x+1,4) -- (\x,4) -- cycle;			
\fill[red] (\x+0.5,3.5) circle (0.25);
}
\foreach \x in {28}
{\draw[fill=white,white]  (\x,3) -- (\x+2,3) -- (\x+2,4) -- (\x,4) -- cycle;
\draw[-] (\x,3) -- (\x,4);
\draw[-] (\x,3) -- (\x+2,3);
\draw[-] (\x,4) -- (\x+2,4);
\draw[-] (\x+1,3) -- (\x+1,4);
\node at (\x+1.5,3.5) {~\scriptsize{$\cdots$}};
}
\foreach \x in {0}
{\draw[fill=white,white]  (\x,3) -- (\x-2,3) -- (\x-2,4) -- (\x,4) -- cycle;
\draw[-] (\x,3) -- (\x,4);
\draw[-] (\x,3) -- (\x-2,3);
\draw[-] (\x,4) -- (\x-2,4);
\draw[-] (\x-1,3) -- (\x-1,4);
\node at (\x-1.5,3.5) {\scriptsize{$\cdots$}};
}
}
\\[3pt]
\tikz[scale=0.39]{
\foreach \x in {0,1,2,3,4,5,6,7,8,9,10,11,12,13,14,15,16,17,18,19,20,21,22,23,24,25,26,27}
{\draw[fill=white]  (\x,3) -- (\x+1,3) -- (\x+1,4) -- (\x,4) -- cycle;			
}
\foreach \x in {2,6}
{\draw[fill=white]  (\x,3) -- (\x+1,3) -- (\x+1,4) -- (\x,4) -- cycle;			
\fill[cyan] (\x+0.5,3.5) circle (0.25);
}
\foreach \x in {}
{\draw[fill=white]  (\x,3) -- (\x+1,3) -- (\x+1,4) -- (\x,4) -- cycle;			
\fill[red] (\x+0.5,3.5) circle (0.25);
}
\foreach \x in {28}
{\draw[fill=white,white]  (\x,3) -- (\x+2,3) -- (\x+2,4) -- (\x,4) -- cycle;
\draw[-] (\x,3) -- (\x,4);
\draw[-] (\x,3) -- (\x+2,3);
\draw[-] (\x,4) -- (\x+2,4);
\draw[-] (\x+1,3) -- (\x+1,4);
\node at (\x+1.5,3.5) {~\scriptsize{$\cdots$}};
}
\foreach \x in {0}
{\draw[fill=white,white]  (\x,3) -- (\x-2,3) -- (\x-2,4) -- (\x,4) -- cycle;
\draw[-] (\x,3) -- (\x,4);
\draw[-] (\x,3) -- (\x-2,3);
\draw[-] (\x,4) -- (\x-2,4);
\draw[-] (\x-1,3) -- (\x-1,4);
\node at (\x-1.5,3.5) {\scriptsize{$\cdots$}};
}
}
\end{figure}
\n From this final calculation, we can now write down result of applying the udFToda evolution to $\mathbf{Q}(t)$ and $\mathbf{W}(t)$:
$$\mathbf{Q}(t+1)=(1,1,3,2),~~~\mathbf{W}(t+1)=\begin{array}{ccccc}
&&3&&\\&1&&2&\\1&&3&&1
\end{array}.$$
\end{ex}

\subsection{Diagrammatic Representation of udFToda}
To aid in visualising the udFToda process, we employ a diagrammatic representation of a BBS (udToda) evolution. We will represent a BBS evolution
$$(Q_1^t,W_1^t,Q_2^t,\ldots,W_{n-1}^t,Q_n^t)\mapsto
(Q_1^{t+1},W_1^{t+1},Q_2^{t+1},\ldots,W_{n-1}^{t+1},Q_n^{t+1})$$
(where we drop the $\infty$'s on either end) in the following diagram:

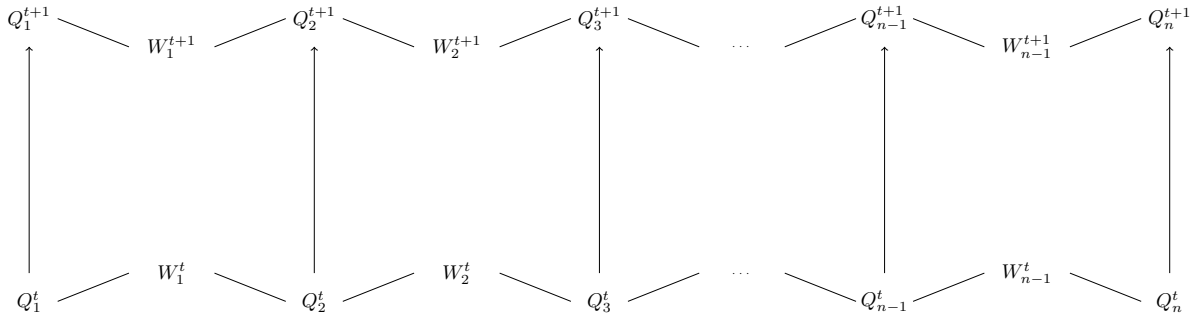
\begin{figure}[H]
\centering
\scalebox{0.75}{
\tikz[scale=5]{
\foreach \s in {0.10}{
\foreach \x in {1,...,5}{
\foreach \n in {6}{
\ifnum \x<4
\node at (\x,\n-2) {$Q_{\x}^{t}$};
\fi
}
}
\foreach \x in {1,...,5}{
\foreach \n in {5}{
\ifnum \x<4
\node at (\x,\n) {$Q_{\x}^{t+1}$};
\fi
\draw[->](\x,\n-1+\s)--(\x,\n-\s);
\ifnum \x<5
\draw[-](\x+\s,\n) -- (\x+0.5-\s*1.5,\n-\s);
\draw[-](\x+1-\s,\n) -- (\x+0.5+\s*1.5,\n-\s);
\ifnum \x<3
\node at (\x+0.5,\n-\s) {{$W_{\x}^{t+1}$}};
\fi
\draw[-](\x+\s,\n-1) -- (\x+0.5-\s*1.5,\n+\s-1);
\draw[-](\x+1-\s,\n-1) -- (\x+0.5+\s*1.5,\n+\s-1);
\ifnum \x<3
\node at (\x+0.5,\n+\s-1) {{$W_{\x}^{t}$}};
\fi
\fi}
}
\node at (3+0.5,5-\s) {\tiny{$\cdots$}};
\node at (3+0.5,5+\s-1) {\tiny{$\cdots$}};
\node at (4,5) {$Q_{n-1}^{t+1}$};
\node at (4,4) {$Q_{n-1}^{t}$};
\node at (4+0.5,5-\s) {{$W_{n-1}^{t+1}$}};
\node at (4+0.5,5+\s-1) {{$W_{n-1}^{t}$}};
\node at (5,5) {$Q_{n}^{t+1}$};
\node at (5,4) {$Q_{n}^{t}$};
}
}
}
\end{figure}

\begin{ex}
Consider the following BBS evolution:
\begin{figure}[H]
\centering
\tikz[scale=0.44]{
\foreach \s in {-1}{
\node (1) at (-4,3.5) {$t$: };
\node (2) at (-4,3.5+\s) {$t+1$: };
\foreach \x in {0,1,2,3,4,5,6,7,8,9,10,11,12,13,14,15}
{
\draw[fill=white]  (\x,3+\s) -- (\x+1,3+\s) -- (\x+1,4+\s) -- (\x,4+\s) -- cycle;			
}
\foreach \x in {4,5,6,8,12,14,15}
{\draw[fill=white]  (\x,3+\s) -- (\x+1,3+\s) -- (\x+1,4+\s) -- (\x,4+\s) -- cycle;			
\fill[cyan] (\x+0.5,3.5+\s) circle (0.25);
}
\foreach \x in {}
{\draw[fill=white]  (\x,3+\s) -- (\x+1,3+\s) -- (\x+1,4+\s) -- (\x,4+\s) -- cycle;			
\fill[red] (\x+0.5,3.5+\s) circle (0.25);
}
\foreach \x in {16}
{\draw[fill=white,white]  (\x,3+\s) -- (\x+2,3+\s) -- (\x+2,4+\s) -- (\x,4+\s) -- cycle;
\draw[-] (\x,3+\s) -- (\x,4+\s);
\draw[-] (\x,3+\s) -- (\x+2,3+\s);
\draw[-] (\x,4+\s) -- (\x+2,4+\s);
\draw[-] (\x+1,3+\s) -- (\x+1,4+\s);
\node at (\x+1.5,3.5+\s) {$\cdots$};
}
\foreach \x in {0}
{\draw[fill=white,white]  (\x,3) -- (\x-2,3) -- (\x-2,4) -- (\x,4) -- cycle;
\draw[-] (\x,3+\s) -- (\x,4+\s);
\draw[-] (\x,3+\s) -- (\x-2,3+\s);
\draw[-] (\x,4+\s) -- (\x-2,4+\s);
\draw[-] (\x-1,3+\s) -- (\x-1,4+\s);
\node at (\x-1.5,3.5+\s) {$\cdots$};
}
}
\foreach \x in {0,1,2,3,4,5,6,7,8,9,10,11,12,13,14,15}
{\draw[fill=white]  (\x,3) -- (\x+1,3) -- (\x+1,4) -- (\x,4) -- cycle;			
}
\foreach \x in {1,2,3,7,10,11,13}
{\draw[fill=white]  (\x,3) -- (\x+1,3) -- (\x+1,4) -- (\x,4) -- cycle;			
\fill[cyan] (\x+0.5,3.5) circle (0.25);
}
\foreach \x in {}
{\draw[fill=white]  (\x,3) -- (\x+1,3) -- (\x+1,4) -- (\x,4) -- cycle;			
\fill[red] (\x+0.5,3.5) circle (0.25);
}
\foreach \x in {16}
{\draw[fill=white,white]  (\x,3) -- (\x+2,3) -- (\x+2,4) -- (\x,4) -- cycle;
\draw[-] (\x,3) -- (\x,4);
\draw[-] (\x,3) -- (\x+2,3);
\draw[-] (\x,4) -- (\x+2,4);
\draw[-] (\x+1,3) -- (\x+1,4);
\node at (\x+1.5,3.5) {$\cdots$};
}
\foreach \x in {0}
{\draw[fill=white,white]  (\x,3) -- (\x-2,3) -- (\x-2,4) -- (\x,4) -- cycle;
\draw[-] (\x,3) -- (\x,4);
\draw[-] (\x,3) -- (\x-2,3);
\draw[-] (\x,4) -- (\x-2,4);
\draw[-] (\x-1,3) -- (\x-1,4);
\node at (\x-1.5,3.5) {$\cdots$};
}
}
\end{figure}
\n which is represented diagrammatically by
\begin{figure}[H]
\centering
\tikz[scale=0.54]{
\foreach \s in {12}{
\node (A) at (-3,3+\s-\s*0.2) {3};
\node (B) at (3,3+\s-\s*0.2) {1};
\node (C) at (9,3+\s-\s*0.2) {1};
\node (D) at (15,3+\s-\s*0.2) {2};
\node (E) at (0,3+\s-\s*0.2-\s*0.1) {1};
\node (F) at (6,3+\s-\s*0.2-\s*0.1) {3};
\node (G) at (12,3+\s-\s*0.2-\s*0.1) {1};
\node (H) at (0,3+\s*0.4) {3};
\node (I) at (6,3+\s*0.4) {2};
\node (J) at (12,3+\s*0.4) {1};
\node (K) at (-3,3+\s*0.3) {3};
\node (L) at (3,3+\s*0.3) {1};
\node (M) at (9,3+\s*0.3) {2};
\node (N) at (15,3+\s*0.3) {1};
\draw[->] (K) -- (A);
\draw[->] (L) -- (B);
\draw[->] (M) -- (C);
\draw[->] (N) -- (D);
\draw[-] (A) -- (E);
\draw[-] (B) -- (E);
\draw[-] (B) -- (F);
\draw[-] (C) -- (F);
\draw[-] (C) -- (G);
\draw[-] (D) -- (G);
\draw[-] (K) -- (H);
\draw[-] (L) -- (H);
\draw[-] (L) -- (I);
\draw[-] (M) -- (I);
\draw[-] (M) -- (J);
\draw[-] (N) -- (J);
}
}
\end{figure}
\end{ex}

\n Using these basic BBS building blocks, we can now stack them into a representation of the full udFToda evolution. To illustrate the point, the following is the resulting diagram for the udFToda evolution when $n=5$:

\begin{figure}[H]
\centering
\scalebox{1}{
\tikz[scale=2.8]{
\foreach \s in {0.14}{
\foreach \x in {1,...,2}{
\foreach \n [evaluate=\n as \neval using int(\n+\x-1)] in {4}{
\node at (\x,\n) {$Q^{(\n)}_{\x}$};
\draw[->](\x,\n-1+\s)--(\x,\n-\s);
\ifnum \x<2
\draw[-](\x+\s,\n) -- (\x+0.5-\s*1.5,\n-\s);
\draw[-](\x+1-\s,\n) -- (\x+0.5+\s*1.5,\n-\s);
\node at (\x+0.5,\n-\s) {\tiny{$W_{\x}^{\neval}(t+1)$}};
\draw[-](\x+\s,\n-1) -- (\x+0.5-\s*1.5,\n+\s-1);
\draw[-](\x+1-\s,\n-1) -- (\x+0.5+\s*1.5,\n+\s-1);
\node at (\x+0.5,\n+\s-1) {\tiny{$W_{\x}^{\neval}(t)$}};
\fi}
}
\foreach \x in {1,...,3}{
\foreach \n [evaluate=\n as \neval using int(\n+\x-1)] in {3}{
\node at (\x,\n) {$Q^{(\n)}_{\x}$};
\draw[->](\x,\n-1+\s)--(\x,\n-\s);
\ifnum \x<3
\draw[-](\x+\s,\n) -- (\x+0.5-\s*1.5,\n-\s);
\draw[-](\x+1-\s,\n) -- (\x+0.5+\s*1.5,\n-\s);
\node at (\x+0.5,\n-\s) {\tiny{$W_{\x}^{\neval}(t+1)$}};
\draw[-](\x+\s,\n-1) -- (\x+0.5-\s*1.5,\n+\s-1);
\draw[-](\x+1-\s,\n-1) -- (\x+0.5+\s*1.5,\n+\s-1);
\node at (\x+0.5,\n+\s-1) {\tiny{$W_{\x}^{\neval}(t)$}};
\fi}
}
\foreach \x in {1,...,4}{
\foreach \n [evaluate=\n as \neval using int(\n+\x-1)] in {2}{
\node at (\x,\n) {$Q^{(\n)}_{\x}$};
\draw[->](\x,\n-1+\s)--(\x,\n-\s);
\ifnum \x<4
\draw[-](\x+\s,\n) -- (\x+0.5-\s*1.5,\n-\s);
\draw[-](\x+1-\s,\n) -- (\x+0.5+\s*1.5,\n-\s);
\node at (\x+0.5,\n-\s) {\tiny{$W_{\x}^{\neval}(t+1)$}};
\draw[-](\x+\s,\n-1) -- (\x+0.5-\s*1.5,\n+\s-1);
\draw[-](\x+1-\s,\n-1) -- (\x+0.5+\s*1.5,\n+\s-1);
\node at (\x+0.5,\n+\s-1) {\tiny{$W_{\x}^{\neval}(t)$}};
\fi}
}
\foreach \x in {1,...,5}{
\foreach \n [evaluate=\n as \neval using int(\n+\x-1)] in {1}{
\node at (\x,\n) {$Q^{(\n)}_{\x}$};
\draw[->](\x,\n-1+\s)--(\x,\n-\s);
\ifnum \x<5
\draw[-](\x+\s,\n) -- (\x+0.5-\s*1.5,\n-\s);
\draw[-](\x+1-\s,\n) -- (\x+0.5+\s*1.5,\n-\s);
\node at (\x+0.5,\n-\s) {\tiny{$W_{\x}^{\neval}(t+1)$}};
\draw[-](\x+\s,\n-1) -- (\x+0.5-\s*1.5,\n+\s-1);
\draw[-](\x+1-\s,\n-1) -- (\x+0.5+\s*1.5,\n+\s-1);
\node at (\x+0.5,\n+\s-1) {\tiny{$W_{\x}^{\neval}(t)$}};
\fi}
}
\foreach \x in {1,...,5}{
\foreach \n in {0}{
\node at (\x,\n) {$Q^{(\n)}_{\x}$};
}
}
\foreach \x in {1,...,5}{
\node at (\x,0-\s) {\rotatebox{90}{=}};
\node at (\x,-2*\s) {$Q_{\x}^t$};
}
\foreach \x in {2,...,5}{
\node at (\x,7-\x+\s-1) {\rotatebox{90}{=}};
\node at (\x,7-\x+2*\s-1) {$~~~Q_{\x}^{t+1}$};
}
\node at (1,5+\s-1) {\rotatebox{90}{=}};
\node at (1,5+2*\s-1) {$~~~Q_{1}^{t+1}$};
}
\foreach \x in {1,...,4}{
\node at (0.4,\x-0.5) {Step \x ~~~$\rightarrow$ };
}
}
}
\end{figure}

\n Observe that one does indeed have four steps, each of which determines the rows of $\mathbf{W}(t+1)$ row-by-row (from the bottom, up) and the entries of $\mathbf{Q}(t+1)$ from right-to-left.\\[4pt]


\begin{ex} The following is the all-in-one diagrammatic representation of the single time step evolution in Example \ref{udftodaexample}:
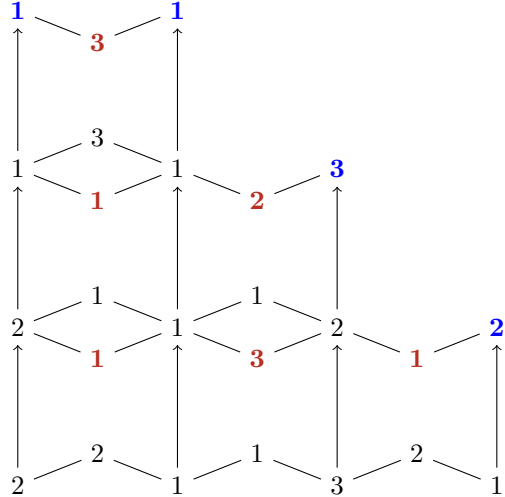
\begin{figure}[H]
\centering
\tikz[scale=0.7]{
\node (A1) at (0,0) {2};
\node (B1) at (3,0) {1};
\node (C1) at (6,0) {3};
\node (D1) at (9,0) {1};
\node (A2) at (0,3) {2};
\node (B2) at (3,3) {1};
\node (C2) at (6,3) {2};
\node (D2) at (9,3) {\textbf{\color{blue}{2}}};
\node (A3) at (0,6) {1};
\node (B3) at (3,6) {1};
\node (C3) at (6,6) {\textbf{\color{blue}{3}}};
\node (A4) at (0,9) {\textbf{\color{blue}{1}}};
\node (B4) at (3,9) {\textbf{\color{blue}{1}}};
\node (E1) at (1.5,0.6) {2};
\node (F1) at (4.5,0.6) {1};
\node (G1) at (7.5,0.6) {2};
\node (E2) at (1.5,2.4) {\textbf{\color{red}{1}}};
\node (F2) at (4.5,2.4) {\textbf{\color{red}{3}}};
\node (G2) at (7.5,2.4) {\textbf{\color{red}{1}}};
\node (E3) at (1.5,3.6) {1};
\node (F3) at (4.5,3.6) {1};
\node (E4) at (1.5,5.4) {\textbf{\color{red}{1}}};
\node (F4) at (4.5,5.4) {\textbf{\color{red}{2}}};
\node (E5) at (1.5,6.6) {3};
\node (E6) at (1.5,8.4) {\textbf{\color{red}{3}}};
\draw [-] (A1) -- (E1);
\draw [-] (B1) -- (E1);
\draw [-] (B1) -- (F1);
\draw [-] (C1) -- (F1);
\draw [-] (C1) -- (G1);
\draw [-] (D1) -- (G1);
\draw [-] (A2) -- (E2);
\draw [-] (B2) -- (E2);
\draw [-] (B2) -- (F2);
\draw [-] (C2) -- (F2);
\draw [-] (C2) -- (G2);
\draw [-] (D2) -- (G2);
\draw [-] (A2) -- (E3);
\draw [-] (B2) -- (E3);
\draw [-] (B2) -- (F3);
\draw [-] (C2) -- (F3);
\draw [-] (A3) -- (E4);
\draw [-] (B3) -- (E4);
\draw [-] (B3) -- (F4);
\draw [-] (C3) -- (F4);
\draw [-] (A3) -- (E5);
\draw [-] (B3) -- (E5);
\draw [-] (A4) -- (E6);
\draw [-] (B4) -- (E6);
\draw [->] (A1) -- (A2);
\draw [->] (A2) -- (A3);
\draw [->] (A3) -- (A4);
\draw [->] (B1) -- (B2);
\draw [->] (B2) -- (B3);
\draw [->] (B3) -- (B4);
\draw [->] (C1) -- (C2);
\draw [->] (C2) -- (C3);
\draw [->] (D1) -- (D2);
}
\end{figure}
\n In the above, for ease of comprehension, we use bold blue numbers to represent the entries of $\mathbf{Q}(t+1)$ and bold red numbers for the entries of the triangular array $\mathbf{W}(t+1)$. The bold numbers, collectively, provide all of the time $t+1$ information.
\end{ex}

\subsection{RSK}
In \cite{bib:era}, we demonstrated how the classical box-ball system can be extended to a cellular automaton that captures the udToda dynamics when coordinates are allowed to degenerate to zero. A particular application of this was realised through the encoding of the basic version of a combinatorial algorithm known as \textit{Schensted insertion}. In that paper, the version realised was of word insertion into a word. The full version of Schensted insertion, used in proving the famous Robinson-Schensted-Knuth correspondence, involves word insertion into a so-called \textit{semistandard Young tableau}. It is known that Schensted insertion is an iterated process of coupled Schensted word insertions \cite{bib:ny}. At the level of coordinates, therefore, one should expect the full Schensted insertion process to be captured by udFToda, with zeroes allowed in coordinates. The following example illustrates this.

\begin{ex}\label{schenstedudftoda} The process of Schensted inserting \cite{bib:aigner} the word $1112334$ (or $1^32^13^24^1$) into the tableau 
$$\begin{ytableau}
1&1&2&3&3\\
2&2&4\\3&3\\4&4
\end{ytableau}$$
results in the following tableau:

$$\begin{ytableau}
1&1&1&1&1&2&3&3&4\\2&2&2&3&3\\3&3&4\\4&4
\end{ytableau}.$$
We initialise udFToda with $\mathbf{Q}(t)=(0,3,1,2,1)$\footnote{The initial 0 in the first entry of $\mathbf{Q}(t)$, which propagates for all time $t$, is a necessary inclusion for the encoding of an insertion word \cite{bib:era}. Also, as an implicit extension of \cite[Equation 4.35]{bib:era} is the interpretation of $\mathbf{Q}(t+1)$ as an accumulation of row growths.} and $\mathbf{W}(t)=\begin{array}{ccccccc}
&&&2&&&\\
&&2&&0&&\\
&2&&0&&1&\\
2&&1&&2&&0
\end{array}$, so that $\mathbf{Q}(t)$ encodes the insertion word and $\mathbf{W}(t)$ encodes the initial tableau. The encoding is given by counting 1's, 2's, 3's and 4's. For the tableau, one counts in the rows. After performing the udFToda time evolution, as described in Definition \ref{definitionofudftodaevol}, one obtains $\mathbf{Q}(t+1)=(0,0,1,2,4)$ and $\mathbf{W}(t+1)=\begin{array}{ccccccc}
&&&2&&&\\
&&2&&1&&\\
&3&&2&&0&\\
5&&1&&2&&1
\end{array}$. The triangle $\mathbf{W}(t+1)$ captures the result of Schensted inserting the insertion word into the initial tableau. For reference, we include the udFToda diagram for the time evolution $(\mathbf{Q}(t),\mathbf{W}(t))\mapsto (\mathbf{Q}(t+1),\mathbf{W}(t+1))$.

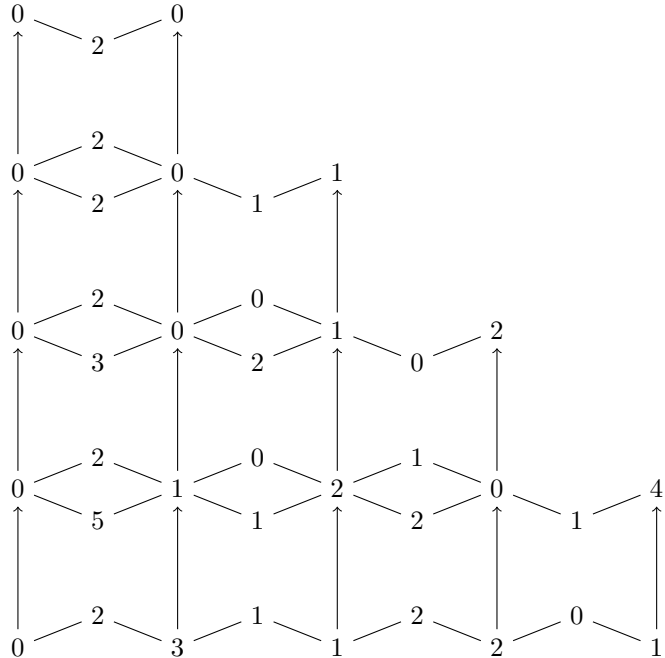
\begin{figure}[H]
\centering
\tikz[scale=0.7]{
\node (H1) at (-3,0) {0};
\node (H2) at (-3,3) {0};
\node (H3) at (-3,6) {0};
\node (H4) at (-3,9) {0};
\node (H5) at (-3,12) {0};
\node (H6) at (0,12) {0};
\node (A1) at (0,0) {3};
\node (B1) at (3,0) {1};
\node (C1) at (6,0) {2};
\node (D1) at (9,0) {1};
\node (A2) at (0,3) {1};
\node (B2) at (3,3) {2};
\node (C2) at (6,3) {0};
\node (D2) at (9,3) {4};
\node (A3) at (0,6) {0};
\node (B3) at (3,6) {1};
\node (C3) at (6,6) {2};
\node (A4) at (0,9) {0};
\node (B4) at (3,9) {1};
\node (J1) at (-1.5,0.6) {2};
\node (E1) at (1.5,0.6) {1};
\node (F1) at (4.5,0.6) {2};
\node (G1) at (7.5,0.6) {0};
\node (J2) at (-1.5,2.4) {5};
\node (E2) at (1.5,2.4) {1};
\node (F2) at (4.5,2.4) {2};
\node (G2) at (7.5,2.4) {1};
\node (J3) at (-1.5,3.6) {2};
\node (E3) at (1.5,3.6) {0};
\node (F3) at (4.5,3.6) {1};
\node (J4) at (-1.5,5.4) {3};
\node (E4) at (1.5,5.4) {2};
\node (F4) at (4.5,5.4) {0};
\node (J5) at (-1.5,6.6) {2};
\node (E5) at (1.5,6.6) {0};
\node (J6) at (-1.5,8.4) {2};
\node (E6) at (1.5,8.4) {1};
\node (J7) at (-1.5,9.6) {2};
\node (E7) at (-1.5,11.4) {2};
\draw [-] (A1) -- (E1);
\draw [-] (B1) -- (E1);
\draw [-] (B1) -- (F1);
\draw [-] (C1) -- (F1);
\draw [-] (C1) -- (G1);
\draw [-] (D1) -- (G1);
\draw [-] (A2) -- (E2);
\draw [-] (B2) -- (E2);
\draw [-] (B2) -- (F2);
\draw [-] (C2) -- (F2);
\draw [-] (C2) -- (G2);
\draw [-] (D2) -- (G2);
\draw [-] (A2) -- (E3);
\draw [-] (B2) -- (E3);
\draw [-] (B2) -- (F3);
\draw [-] (C2) -- (F3);
\draw [-] (A3) -- (E4);
\draw [-] (B3) -- (E4);
\draw [-] (B3) -- (F4);
\draw [-] (C3) -- (F4);
\draw [-] (A3) -- (E5);
\draw [-] (B3) -- (E5);
\draw [-] (A4) -- (E6);
\draw [-] (B4) -- (E6);
\draw [->] (A1) -- (A2);
\draw [->] (A2) -- (A3);
\draw [->] (A3) -- (A4);
\draw [->] (B1) -- (B2);
\draw [->] (B2) -- (B3);
\draw [->] (B3) -- (B4);
\draw [->] (C1) -- (C2);
\draw [->] (C2) -- (C3);
\draw [->] (D1) -- (D2);
\draw [->] (H1) -- (H2);
\draw [->] (H2) -- (H3);
\draw [->] (H3) -- (H4);
\draw [->] (H4) -- (H5);
\draw [->] (A4) -- (H6);
\draw [-] (H1) -- (J1);
\draw [-] (A1) -- (J1);
\draw [-] (H2) -- (J2);
\draw [-] (A2) -- (J2);
\draw [-] (H2) -- (J3);
\draw [-] (A2) -- (J3);
\draw [-] (H3) -- (J5);
\draw [-] (A3) -- (J5);
\draw [-] (H3) -- (J4);
\draw [-] (A3) -- (J4);
\draw [-] (H4) -- (J6);
\draw [-] (A4) -- (J6);
\draw [-] (H4) -- (J7);
\draw [-] (A4) -- (J7);
\draw [-] (H5) -- (E7);
\draw [-] (H6) -- (E7);
}
\end{figure}
\end{ex}

\begin{rem}\label{rem:rskasfulludtoda}
We observe further that the vector $\mathbf{Q}(t+1)$ captures precisely the amount by which each row of the tableau has grown. In Example \ref{schenstedudftoda}, one sees that $\mathbf{Q}(t+1)=(0,0,1,2,4)$, when read from right to left, is understood as a growth of 4 boxes in the first row, a growth of 2 boxes in the second row, and a growth of 1 box in the third row. The fourth row does not grow. We state this now, referring the reader to Equation 4.53 of \cite{bib:era} for the key component required to prove this observation inductively. The referenced equation says that a word insertion into a row of a tableau grows the row by the last coordinate of $\mathbf{Q}(t+1)$. The recursive applications for lower dimensions, corresponding to the subsequent rows, yields the full result.\\[-8pt] 

\n The significance of this observation is that a recursive application of udFToda, driven inhomogeneously by a sequence of words, captures the full RSK correspondence (the insertion and recording tableaux, for those familiar with this combinatorial bijection) and will be explored further by the authors in an upcoming paper. 
\end{rem}

\section{The Geometry of the Full Kostant Toda Lattice and its Discretizations} \label{geometry}
 
 As mentioned at the end of Section \ref{history}, the matrix reformulation of the Toda lattice in Flaschka's variables leads to natural Lie-theoretic interpretations and extensions. While we continue to remain in the phase space setting of Hessenberg matrices, we introduce, in this section, a modicum of concepts/terminology from the theory of Lie groups and their representations that will facilitate a natural geometric re-interpretation of our results. 

\subsection{Kostant's Theorem, a Flag Manifold Embedding and Linearization} \label{sec:kostant}
 
 \n We consider the Lie algebra decomposition of $n \times n$ matrices
 
 \begin{eqnarray} \label{frak}
 \frak{g} = \frak{gl}(n, \mathbb{R}) & = & \frak{n}_- \oplus \frak{b}_+  
 \end{eqnarray}
 where $\frak{n}_-$ is the lower triangular nilpotent sub-algebra (nilradical) and $\frak{b}_+$ is a maximal solvable sub-algebra, referred to as a {\it Borel sub-algebra}:
 \begin{eqnarray*}
 \frak{n}_- =  \left(\begin{array}{ccccc}
0 &  & & &\\
*& 0 &  & &\\
\vdots & \ddots & \ddots & \ddots &\\
\vdots &  & \ddots & \ddots & \\
* & \dots & \dots & * & 0
\end{array} \right), &&
\frak{b}_+ = \left(\begin{array}{ccccc}
* & * &\dots &\dots &*\\
& * & * & &\vdots\\
& & \ddots & \ddots &\vdots\\
& & & \ddots & *\\
& & & & *
\end{array} \right).
\end{eqnarray*}
We will also use $\frak{b}_-$ to denote the transpose of $\frak{b}_+$. Employing the {\it principal nilpotent} element, 
 \begin{eqnarray*}
\epsilon &=& \left(\begin{array}{ccccc}
0 & 1 & & &\\
& 0 & 1 & &\\
& & \ddots & \ddots &\\
& & & \ddots & 1\\
& & & & 0
\end{array} \right)
\end{eqnarray*}
one defines an extended Toda phase space
\begin{eqnarray*}
\epsilon + \frak{b}_- &=& \left(\begin{array}{ccccc}
* & 1 & & &\\
*& * & 1 & &\\
\vdots & \ddots & \ddots & \ddots &\\
\vdots &  & \ddots & \ddots & 1\\
* & \dots & \dots & * & *
\end{array} \right), \\
\end{eqnarray*}
 (which is the space of all lower Hessenberg matrices, $\mathcal{H}$, introduced earlier) on which the Toda Lax equation (\ref{Lax2}) as well as the discrete time Toda equation (\ref{symeseqn}) and their extensions are defined.  It will be useful to introduce the following algebra and group projections,
\begin{eqnarray*}
 \pi_- : \frak{g} \to  \frak{n}_-, \qquad
  \Pi_- : G \to  N_- \\
 \pi_+ : \frak{g} \to  \frak{b}_+, \qquad
 \Pi_+ : G \to  B_+
\end{eqnarray*}
where $G = GL(n, \mathbb{R})$, $N_-$, the lower unipotent matrices, is the exponential group of the algebra $\frak{n}_-$ and $B_+$, the invertible upper triangular matrices, is the exponential group of the algebra $\frak{b}_+$. $\Pi_\pm$ are defined on the open dense subset of $G$ where there is an $LU$ factorization.

 These factorizations provide the basis for effective linearization in both continuous and discrete time Toda. This is facilitated by a key theorem due to Kostant that provides a natural embedding of the Toda dynamics into a flag manifold that plays a role analogous to action-angle variables in classical integrable systems theory.\\
 \smallskip
 
 Before getting to the statement of Kostant's theorem we introduce the following distinguished bidiagonal element of $\mathcal{H}$,
 
\begin{equation} \label{epslam}
\epsilon_\Lambda = \left[\begin{array}{cccc}
\lambda_1 & 1\\
&\lambda_2 & \ddots\\
&&\ddots & 1\\
&&&\lambda_n
\end{array}\right],
\end{equation}
where $ \lambda_1 > \cdots > \lambda_n $ are the eigenvalues of $X_0$ which, for simplicity of exposition, we assume to be distinct. By isospectrality, these are the eigenvalues of $X(t)$ for all $t$. It is immediate from (\ref{Lax2}) that 
$\epsilon_\Lambda$ is a fixed point of the Toda flow.

\begin{thm} \cite{bib:kostant} \label{KThm} For each $X \in \epsilon + \frak{b}_-$ there exists a {\bf unique} lower unipotent $L \in N_-$ such that 
$$
X = L^{-1} \epsilon_\Lambda L.
$$
\end{thm}

Set 
$$
\mathcal{H}_\Lambda = \{ X \in \epsilon + \frak{b}_- : \sigma(X) = \lambda_1 > \cdots > \lambda_n\}, 
$$
the isospectral manifold. It is immediate from the uniqueness statement in Theorem \ref{KThm} that the following defines an embedding into a compact, homogeneous space, known as a 
{\it flag manifold}. 
\begin{eqnarray} \label{eq:principal}
\kappa_\Lambda : \mathcal{H}_\Lambda &\to & G/B_+\\ \nonumber
X &\mapsto& L \mod B_+
\end{eqnarray}
where $X = L^{-1} \epsilon_\Lambda L.$ 

This mapping simultaneously linearizes and completes the Toda flows 
\cite{bib:efh, bib:efs}.  We will refer to $\kappa_\Lambda$ as the {\it principal embedding} because of its relation to the principal nilpotent element $\epsilon$. 

\n If $X_0 = L_0^{-1} \epsilon_\Lambda L_0$ then 
\begin{eqnarray} \nonumber
X(s) &=& \Pi_-^{-1}(e^{s X_0})X_0 \Pi_-(e^{s X_0})\\ \nonumber
     &=& \Pi_-^{-1}(e^{s X_0})L_0^{-1}\epsilon_\Lambda L_0 \Pi_-(e^{s X_0})\\ \nonumber
\kappa_\Lambda(X(s)) &=& L_0 \Pi_-(e^{s X_0}) \mod B_+\\ \nonumber 
&=& L_0 (e^{s X_0}) \mod B_+\\ \label{linprinc}
&=& e^{\epsilon_\Lambda s} L_0 \mod B_+.
\end{eqnarray} 
Thus, one sees that under the embedding $\kappa_\Lambda$, the Toda flow maps to a linear semigroup action by $e^{\epsilon_\Lambda s}$ through the image of the initial value,
$\kappa_\Lambda(X_0) = L_0$ which exists for all time.

The continuous Toda flow commutes with the discrete time Toda evolution. In fact, there is a hierarchy of continuous flows commuting with dFToda and with one another given by Lax equations of the form 
\begin{eqnarray} \label{hierarchy}
\frac{d}{ds_m} X = \left[ X, \pi_- X^m\right].
\end{eqnarray}
On $\mathcal{H}_\Lambda$ the first $n-1$ of these flows are locally independent and their respective images under the principal embedding has the form

\begin{eqnarray} \label{hiflow}
e^{\epsilon^m_\Lambda s_m} L_0 \mod B_+
\end{eqnarray}
and together they generate a torus action (see Section \ref{sec:toremb}) that spans the closure of $\kappa_\Lambda(\mathcal{H}_\Lambda)$ in $G/B_+$ \cite{bib:efs}.
\medskip{}

\begin{rem} \label{rem:Lie} The language we have used in this section to describe the Toda phase space and dynamics (Borel and nilradical sub-algebras, principal nilpotent elements) have immediate extensions to the setting of general real semisimple Lie algebras to define {\it full} versions \cite{bib:efs} of the so-called {\it generalized Toda lattices} \cite{bib:kostant2}. Kostant also introduced an analogue of $\epsilon_\Lambda$ in Lemma 3.52 of \cite{bib:kostant2}.
This element is distinguished by the fact that it lies in the intersection 
$\mathcal{H}_\Lambda  \cap \frak{b_+}$, and in that regard it is essentially unique. (There are $n!$ elements in this intersection corresponding to permutations of the distinct eigenvalues.)
\end{rem}

\subsection{Representation Theory and tau-functions} \label{sec:repn}

Let $G$ denote the group $GL(\ell +1)$ and let $(\rho_n, V_n)$ denote the $n^{th}$ fundamental representation of $G$. $\rho_1$ is the 
{\it birth representation} with $V_1 = \mathbb{C}^{\ell + 1}$ defined by 
$$
\rho_1(g) v = g v
$$
for $v \in V_1$. This induces a representation on the exterior algebra of $V_1$ that defines the remaining fundamental representations, respectively, on 
$V_n = \bigwedge^n \mathbb{C}^{\ell + 1}$ given by 
$$
\rho_n(g) v_1 \wedge \dots \wedge v_n = g v_1 \wedge \dots \wedge g v_n.
$$
With respect to the standard basis $e_1, \dots, e_{\ell + 1}$ of $\mathbb{C}^{\ell + 1}$
one defines a Hermitian inner product on $V_n$ by 
$$
\langle e_{i_1} \wedge \dots \wedge e_{i_n}, e_{j_1} \wedge \dots \wedge e_{j_n}\rangle = \delta_{i_1, j_1} \cdots \delta_{i_n, j_n}.
$$
Set $v^{(n)} = e_1 \wedge \dots \wedge e_n$ and $v_{(n)} = e_{\ell - n + 2} \wedge \dots \wedge e_{\ell + 1}$. These are, respectively, the highest and lowest weight vectors, with respect to lexicographic order, of the representation $\rho_n$. \cite{bib:fuha}

We can now define the $n^{th}$ {\it $\tau$-function} to be
$$
\tau_n(s) = \langle \exp(s X_0) v^{(n)}, v^{(n)}\rangle
$$
which is just the $n^{th}$ principal minor of $\exp{s X_0}$.


\subsection{Homogeneous Space Representations of dFToda} \label{sec:homog}

 Invoking Proposition \ref{prop:dft}, the principal embedding defined for Full Toda may be dynamically extended to its discretization as

\begin{eqnarray} \nonumber
\kappa_\Lambda: \mathcal{H}^{>0}_\Lambda & \to & G/B_+\\
X(t) &\mapsto& \epsilon^t_\Lambda L_0 \mod B_+  \label{kappaSymes}
\end{eqnarray}
for $t \in \mathbb{Z}$. If $X(0)$ is totally positive then, by Corollary \ref{cor:wp}, $X(t)$ remains totally positive and so defined for all time. Thus, $\kappa_\Lambda(X(t))$ is defined and by (\ref{linprinc}) and Proposition \ref{prop:dft} it takes the value stated in (\ref{kappaSymes}). Moreover, if $L_0$ is totally positive, then $[\epsilon^t_\Lambda L_0]_-$ is totally positive as well (see Remark \ref{tphiflow}).

However, it is important to note that the total positivity of $X$ does not necessarily imply that $L = \kappa_\Lambda(X)$ is totally positive. Nonetheless, the Symes factorization algorithm may be directly related to the principal embedding as follows. Begin by applying Kostant's theorem to a totally positive initial condition $X(0)$ followed by its LU factorization, and then by the Symes commutation action and continue to iterate this
\begin{eqnarray*}
X(0) &=& L_0^{-1} \epsilon_\Lambda L_0\\
&=& L(0) R(0)\\
X(1) &=& L(0)^{-1} X(0) L(0)\\
&=& L(1) R(1)\\
X(2) &=& L(1)^{-1} X(1) L(1)\\
&=& L(2) R(2) \cdots.
\end{eqnarray*}
By Corollary \ref{cor:wp}, the sequence $L(0), L(1), L(2), \dots $ is well-defined (for all discrete time) and totally positive. This sequence defines an orbit under another flag manifold embedding
\begin{eqnarray} \label{Symesemb}
S_\Lambda: \mathcal{H}^{>0}_\Lambda &\to& (G/B_+)^{>0} \\
X &\to & \,\,\,\, L 
\end{eqnarray}
where $X = LR$. This will be referred to as the {\it Symes embedding}. 

\begin{rem} \label{TPFlag}
For a complete introduction to the totally positive flag manifold $(G/B_+)^{>0}$ and relations to FKToda we refer the reader to \cite{bib:kw}. For our purposes here we will only need to consider $N_-^{>0}$ which is an open dense cell in $(G/B_+)^{>0}$. By Corollary \ref{cor:wp}, dFToda preserves $N_-^{>0}$; so we slightly abuse nomenclature when we speak of the Lusztig parametrization as providing a parametrization of $(G/B_+)^{>0}$.
\end{rem}

\begin{rem} \label{tphiflow}
We note that since $e^{\epsilon^m_\Lambda s_m}$ is manifestly totally positive and invertible, if $L_0$ is also totally positive then by the Loewner-Whitney theorem the product in (\ref{hiflow}) is TNN and so 
by Corollary \ref{cor:lw} 
the corresponding flow induced through $\kappa_\Lambda$ on the flag manifold
preserves $(G/B_+)^{>0}$. The same holds for the discrete flow (\ref{kappaSymes}) as long as the eigenvalues are all positive. 
\end{rem}
\bigskip

To relate the Symes orbit to 
(\ref{kappaSymes}) define the sequence $L_t \in N^{>0}_-$ coming from that principal embedding by
\begin{eqnarray*}
\epsilon^t_\Lambda L_0 &=& L_t R_t \,\,\,\, \mbox{or} \\
L_t &=& \left[ \epsilon^t_\Lambda L_0 \right]_- \left( \doteq \epsilon^t_\Lambda L_0 \mod B_+ \right).
\end{eqnarray*}
 Then
\begin{eqnarray*}
L(0) &=& \left[ L_0^{-1} \epsilon_\Lambda L_0 \right]_-\\
&=& L_0^{-1} \left[ \epsilon_\Lambda L_0 \right]_-\\
&=& L_0^{-1} L_1 \\
L(1) &=& \left[ L(0)^{-1} X(0) L(0) \right]_- = \left[ L(0)^{-1} X(0) L(0) R(0) \right]_-\\
&=& \left[ L(0)^{-1} X(0)^2\right]_-\\
&=& \left[L(0)^{-1} L_0^{-1} \epsilon^2_\Lambda L_0 \right]_-\\
&=& L(0)^{-1} L_0^{-1} \left[\epsilon^2_\Lambda L_0 \right]_-\\
&=& L_1^{-1} L_0 L_0^{-1} L_2\\
&=& L_1^{-1} L_2\\
&\vdots&\\
L(t) &=& L_t^{-1} L_{t+1}.
\end{eqnarray*}
Thus the Symes orbit in $G/B_+$ is seen to be a discrete, matrix ``logarithmic derivative" of the linear semigroup flow (\ref{kappaSymes}).

\subsection{Torus Embedding} \label{sec:toremb}

This section describes  a modification of the principal embedding that represents the collective hierarchy of Toda flows, (\ref{hierarchy}), as a canonical torus action on a related flag manifold. (This action is the analogue of angle variables in classical integrable systems theory.) 
We make use of the following basic lemma (see \cite{bib:er}, Lemma 4.4 for a derivation):

\begin{lem}\label{uppertodiagefhexpl}
If $\lambda_1,\ldots,\lambda_n$ are distinct, then one has $\epsilon_\Lambda=U D_\Lambda U^{-1}$, where $U=(u_{ij})$ is the upper triangular matrix given by
$$u_{ij}=\prod_{k=1}^{i-1}(\lambda_j-\lambda_k),~~~1\leq i\leq j\leq n$$
and $D_\Lambda=\epsilon_\Lambda-\epsilon=\text{diag}(\lambda_1,\ldots,\lambda_n)$.
\end{lem}

\n Based on this one is led to consider a variation on the principal embedding,
\begin{eqnarray*}
\mbox{tor}_\Lambda : \mathcal{H}_\Lambda &\to & G/B_+\\
X &\mapsto& U^{-1}L \mod B_+
\end{eqnarray*}
from $X$ to its matrix of left eigenvectors 
$\left([U^{-1}L] X = D_\Lambda [U^{-1}L]\right)$. As before one can track the dFToda dynamics through this embedding and find that
\begin{eqnarray*}
\mbox{tor}_\Lambda(X(t)) &=&  D^t_\Lambda U^{-1}L \mod B_+. 
\end{eqnarray*}
Also, as before, this can be extended to the commuting hierarchy of Toda flows (\ref{hierarchy}),
\begin{eqnarray*}
\mbox{tor}_\Lambda(X(t_1, \dots, t_n)) &=& 
\left({\sum_{m=1}^{n-1}  D^{m t_m}_\Lambda}\right)U^{-1}L \mod B_+, 
\end{eqnarray*}
$t_j \in \mathbb{Z}$.
This amounts to the discretization of the (split) torus action,
\begin{eqnarray} \label{torusact}
(\mathbb{R}^*)^n \times G/B_+ &\to & G/B_+\\ \nonumber
(r_1, \dots, r_n) \times U^{-1} L &\to & \mbox{diag} (r_1, \cdots, r_n) U^{-1} L,
\end{eqnarray}
explaining why $\mbox{tor}_\Lambda$ is referred to as the \textit{torus embedding}. Since $\mbox{tor}_\Lambda(\mathcal{H}_\Lambda)$ is a Zariski open subset of $G/B_+$, the closures of the orbits of 
(\ref{torusact}) stratify $G/B_+$. These {\it torus strata} have interesting combinatorial applications 
described in \cite{bib:ggms}. Their relevance for the Full Toda Lattice is considered in more detail in 
\cite{bib:efs, bib:kw, bib:ks}. Here we will simply briefly elaborate on the feature that dFToda is the stroboscope of the completely integrable system described in \cite{bib:efs}. 

The flag manifold has its own natural embedding into a large projective space whose coordinates are comprised of the minors of $\ell$-tuples of column vectors of a representative matrix $g \in G$. For $g$ in the image of $\mbox{tor}_\Lambda$, these are minors within $\ell$-tuples of column vectors of $U^{-1} L$ (equivalently, $\ell$-tuples of eigenvectors of $X$). As projective coordinates these minors are referred to as {Pl\"ucker coordinates} and, for those of interest to us, we may index them as follows. For an $(n-k)$-set $J \subset \{ 1, \dots, n\}$, $J_0 = \{ 1, \dots, n-k\}$, $\pi_J$ will denote the $k \times k$ minor of the first $k$ columns of $U^{-1}L$ with row indices in 
$\{ 1, \dots, n\} \backslash J$ and $\pi^*_J$ is the $(n-k) \times (n-k)$ minor of the first $n-k$ columns of $U^{-1}L$ and rows from $J$. A key observation here is that under the torus action (\ref{torusact}) the product $\pi_J \pi^*_J$  transforms just by the factor $\prod_{i=1}^n r_i$, independent of $J$. Hence the  rational function $\pi_J \pi^*_J/\pi_{J_0} \pi^*_{J_0}$ is invariant under the torus action which is consistent with the restriction of the dynamics to the torus strata. Moreover,  $\pi_J \pi^*_J/\pi_{J_0} \pi^*_{J_0} \circ \mbox{tor}_\Lambda$ is a rational function in the coordinates of $X$ which is invariant under full Toda and dFToda.

\section{Tau functions and the Discrete Dynamics of Lusztig parameters}
\label{sec:LusDynam}

One can now ask if or how the torus orbit stratification under the torus embedding passes over to the Symes embedding. One can at least start by revisiting the relation between the principal and Symes embedding discussed in Section \ref{sec:homog}. This can be summarized as

\begin{eqnarray*}
\kappa_\Lambda: \mathcal{H}_\Lambda &\to& G/B_+\\
|| && \,\,\,\,\, \downarrow \mathcal{I}\\
S_\Lambda: \mathcal{H}_\Lambda &\to& G/B_+
\end{eqnarray*}
where
$$
\mathcal{I}(L_0) = L(0) = L_0^{-1} [\epsilon_\Lambda L_0]_-.
$$
The iterated dFToda dynamics is well-defined, making this diagram commute, if one starts with an initial matrix $ X(0) \in \mathcal{H}^{>0}_\Lambda $. One can now examine how the dynamics of dFToda distributes itself across the recursive       
dToda structure described in Section \ref{section:dfktodarecdtoda}. Recall that the recursive dToda structure is given in terms of a sequence of time-dependent tridiagonal Hessenberg matrices, $H^{(i)}(t)$
defined in (\ref{defnhessenbergtruncs}).
The evolution is then given by a coupled dToda evolution,

\begin{align*}
    H^{(i)}(t)&=T_i(t+1)R^{(i)}(t)=R^{(i-1)}(t)T_i(t)\\
    H^{(i)}(t+1)&=R^{(i-1)}(t+1)T_i(t+1)\\
    &=R^{(i-1)}(t+1)H^{(i)}(t)(R^{(i)}(t))^{-1}\\
    &=R^{(i-1)}(t+1)R^{(i-1)}(t)H^{(i)}(t-1)(R^{(i)}(t-1))^{-1}(R^{(i)}(t))^{-1}\\
    & \vdots \\
    &= \prod^{1}_{s=t+1}R^{(i-1)}(s) H^{(i)}(0) \prod^{t}_{s=0}(R^{(i)}(s))^{-1}\\
    &= \prod^{0}_{s=t+1}R^{(i-1)}(s) T_i(0) \prod^{t}_{s=0}(R^{(i)}(s))^{-1}.
\end{align*}
From this it is possible to inductively define the dynamics on the $R^{(i)}$ directly through the $\tau$-function; i.e., independently of the dynamics on the $T_i$ or factorization.

\begin{eqnarray*}
\tau^i_j(t+1) &=&  \Big\langle \prod^{0}_{s=t+1}R^{(i-1)}(s) T_i(0) \prod^{t}_{s=0}(R^{(i)}(s))^{-1} v^{(j)}, v^{(j)}\Big\rangle\\
&=&  \prod^{t}_{s=0} (\tau^i_j(s))^{-1}\Big\langle \prod^{0}_{s=t+1}R^{(i-1)}(s) T_i(0)  v^{(j)}, v^{(j)}\Big\rangle.
\end{eqnarray*}

One can also describe this dynamics in terms of coordinates. As mentioned after Definition \ref{TPFlag},  $N^{>0}$ is an open dense cell in $\left( G/B_+\right)^{>0}$, invariant under dFToda
and so the Lusztig parameters provide coordinates on this cell in terms of which the dynamics under the Symes embedding, $S_\Lambda$, may be described. This is already apparent from
(\ref{tauf}) which we will now expand on. 
\medskip

Let $\mathcal{T}$ denote the triangular array of Lusztig parameters which we display, in Figure \ref{lusztignetwork}, as edge weights on a directed graph (edges are oriented in the downward and right directions and unlabelled edges will be taken to have weight 1). 

\begin{rem} Figure \ref{lusztignetwork} provides a diagrammatic representation for the coordinatization of $N_-^{>0}$ by Lusztig parameters. The $a_{ij}$ entry of the corresponding lower unipotent matrix is given as the sum of multiplicative path weights from {\rm source} i on the left to {\rm sink} j on the right. Similarly the $D_{ii}$ give the diagonal entries of a general lower triangular matrix. This diagrammatic approach may also be used to calculate minors of lower triangular matrices \cite{bib:fz}.
\end{rem}
Now we describe the dFToda dynamics in terms of the evolution of these Lusztig parameters. This can be presented in terms of the tridiagonal decomposition into the $H^{(i)}$ as described in Definition \ref{def:dft}. In these terms, here is the sequence of moves that gives a time-step in the evolution of the Lusztig parrameters:
\begin{itemize}
    \item Stage 0: One begins with the initial Hessenberg matrix $X(0) = L(0)R(0)$ where by Theorem \ref{theorem:bfz} one has the unique factorization 
    $L(0) = T_1(u(0))T_2(u(0))\cdots T_{n-1}(u(0))$ and $R(0) = R^{(0)}(t=0)$ is determined by the principal minors, $\tau^0_j(0)$, of $X(0)$ as in (\ref{R-repn}). 
    \item Stage 1: One has 
    $H^{(i)}(t):=R^{(i-1)}(t)T_i(t)$ where the Lusztig parameters $\vec{u}(i,t)$ of $T_i(t)$ are those appearing in the $i^{th}$ ``column" of Figure \ref{lusztignetwork}. As in Stage 0, the principal minors, $\tau^i_j(t)$, of $H^{(i)}(t)$ determine $R^{(i)}(t)$. 
    \item Stage 2: Then, by Definition \ref{def:dft}(ii), $T_i(t+1) = H^{(i)}(t) (R^{(i)}(t))^{-1}$ which updates the Lusztig parameters in the $i^{th}$ column. Explicitly, in the $i^{th}$ column, the update is given by 
    (\ref{tauf})
    $$u_j^{i+j-1}(t+1) = u_j^{i+j-1}(t) \frac{\tau_{j+1}^{i-1}(t) \tau_{j-1}^i(t)}{\tau_{j}^{i-1}(t) \tau_{j}^i(t)}.$$
    \item Stage 3: Finally, applying (\ref{taus}), one has
    $\tau_{j}^{i-1}(t+1)  =  \tau_j\left( R(t+1) T_1(u(t+1)) T_2(u(t+1)) \cdots T_{i-1}(u(t+1)) \right)$ which determines $R^{(i-1)}(t+1)$.
     $R(t+1)$ is determined {\em{ab initio}} for all $t$ since $\tau^0_j(t)$ is determined by the minors of $X(0)$ and $\epsilon_\Lambda$ \cite{bib:er}. Thus we have inductively determined $$H^{(i)}(t+1):=R^{(i-1)}(t+1)T_i(t+1)$$
    completing the step from $t$ to $t+1$ and setting the stage for the subsequent update.
\end{itemize}

\begin{figure}[H]
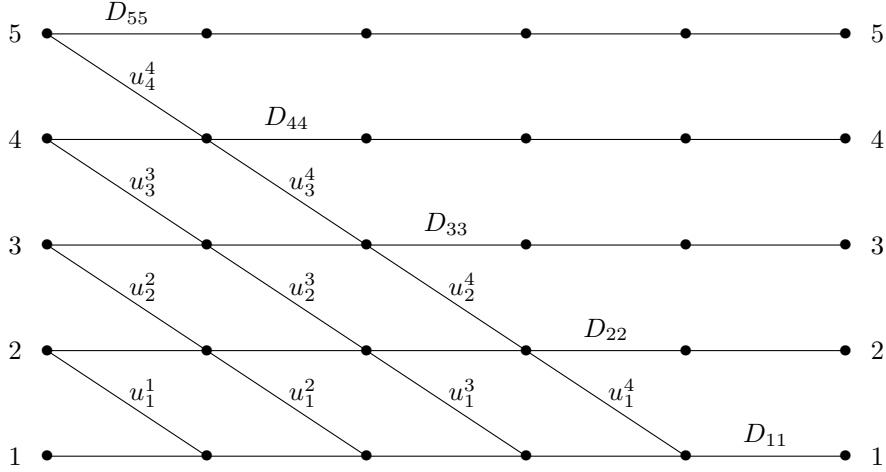

\centering
\tikz[scale=1.4]{
\foreach \y in {1,2,3,4,5}{
\draw[-] (0,\y) -- (7.5,\y);
\node at (-0.3,\y) {\y};
\node at (7.8,\y) {\y};
\foreach \x in {0,1,2,3,4,5}{
\node at (1.5*\x,\y) {$\bullet$};
}
}
\draw[-] (0,2) -- (1.5,1);
\draw[-] (0,3) -- (3,1);
\draw[-] (0,4) -- (4.5,1);
\draw[-] (0,5) -- (6,1);
\node at (0.75,5.2) {$D_{55}$};
\node at (0.9,4.6) {$u_4^4$};
\node at (0.9,3.6) {$u_3^3$};
\node at (0.9,2.6) {$u_2^2$};
\node at (0.9,1.6) {$u_1^1$};
\node at (2.25,4.2) {$D_{44}$};
\node at (2.4,3.6) {$u_3^4$};
\node at (2.4,2.6) {$u_2^3$};
\node at (2.4,1.6) {$u_1^2$};
\node at (3.75,3.2) {$D_{33}$};
\node at (3.9,2.6) {$u_2^4$};
\node at (3.9,1.6) {$u_1^3$};
\node at (5.25,2.2) {$D_{22}$};
\node at (5.4,1.6) {$u_1^4$};
\node at (6.75,1.2) {$D_{11}$};
}
\caption{A network diagram for Lusztig parameters}\label{lusztignetwork}
\end{figure}

\section{Return to Full Kostant Toda and Continuous Dynamics} \label{sec:return}

Having completed our analysis of the discrete and ultradiscrete Toda systems it is natural to ask if that analysis provides any new insights into the continuous FKToda system,
(\ref{Lax2}), from which these were derived. That is indeed the case and that is what will be described in this section. In particular we derive a direct representation of the FKToda flow in terms of Lusztig parameters which in some sense is the analogue of the Lusztig parameter dynamics for dFToda described in Section \ref{sec:LusDynam}. The reader should recall the group and algebra projections, $\Pi_\pm$ and $\pi_\pm$ respectively, defined in
Section \ref{sec:kostant} which will be used throughout this section. 
\medskip

Recall from Remark \ref{tphiflow} that if $L_0$ is totally positive then

$$
L_s \doteq \Pi_-(e^{s \epsilon_\Lambda} L_{0})
$$ remains totally positive for all (continuous) $s$. This is the FKToda flow represented under the principal embedding, (\ref{eq:principal}) and (\ref{linprinc}), as a flow on $(G/B_+)^{>0}$. This implicitly defines the FKToda dynamics as a dynamics on Lusztig parameters. We will now make this parameter dynamics explicit.

We work with the equations
\begin{eqnarray} \label{OC1}
e^{s\epsilon_\Lambda} L_{0} &=& L_s R_s\\ \label{OC2}
\epsilon_\Lambda e^{s\epsilon_\Lambda} L_{0} &=& \frac{d}{ds}\left(L_s R_s\right).
\end{eqnarray}

\n Noting that
\begin{eqnarray} \label{OC3}
L_s &=& T_1(s)\cdots T_{n-1}(s)
\end{eqnarray}
is the unique Lusztig 
factorization of $L_s$, and that the left-hand side of (\ref{OC2}) is $\epsilon_\Lambda L_sR_s$, and substituting the Lusztig factorization in for $L_s$, we obtain the following equation:
$$\epsilon_\Lambda T_1(s)\cdots T_{n-1}(s) R_s = \dfrac{d}{ds}(T_1(s)\cdots T_{n-1}(s)R_s).$$
Applying the product rule, multiplying on the left by the inverse of $L_s=T_1(s)\cdots T_{n-1}(s)$ and on the right by $R_s^{-1}$ yields
\begin{eqnarray} \label{eqn:dKLe}
T^{-1}_{n-1}\cdots T^{-1}_1 \epsilon_\Lambda T_1\cdots T_{n-1}  &=& 
\sum_{j=1}^{n-1}\left(T^{-1}_{n-1}\cdots T^{-1}_{j+1} \left(T^{-1}_{j} \frac{d}{ds}T_j\right) T_{j+1}\cdots T_{n-1}\right) + \frac{d}{ds}R.
\end{eqnarray}
Now assume there are upper bidiagonal matrices $R^{(i)}, i = 0,\dots , n-1 $, with $R^{(0)} = \epsilon_\Lambda$ satisfying
\begin{eqnarray} \label{eqn:dLe}
\frac{d}{ds}T_j &=& R^{(j-1)} T_j 
- T_j R^{(j)}\\ \label{eqn:findLe}
\frac{d}{ds}R &=& R^{(n-1)}. 
\end{eqnarray}
The equations (\ref{eqn:dLe}) are equivalent to
\begin{eqnarray} \label{eqn:split}
T^{-1}_{j}\frac{d}{ds}T_j = T^{-1}_{j}R^{(j-1)} T_j 
-  R^{(j)}.
\end{eqnarray}
Substituting these latter, inductively on $j$, into (\ref{eqn:dKLe}), reduces
(\ref{eqn:dKLe}) to
(\ref{eqn:findLe}). (Note that $R(s)$ is not necessarily bidiagonal Hessenberg.)

It follows from (\ref{eqn:split}) that the $R^{(j)}, j>0,$ must, inductively, satisfy

\begin{eqnarray} \label{Rjandjminusone}
R^{(j)} &=& \pi_+ \left( T^{-1}_{j}R^{(j-1)} T_j \right)\\
\label{Rj}
&=& \pi_+ \left( T^{-1}_{j} \cdots T_1^{-1} \epsilon_\Lambda T_1 \cdots T_j \right)
\end{eqnarray}
where the $T_j(s)$ are given by (\ref{OC3}).
To then verify that 
(\ref{eqn:findLe}) holds for $R^{(j)}$ defined by \eqref{Rj}, first observe that
\begin{eqnarray} \nonumber
R^{(n-1)} &=& \pi_+ \left( T^{-1}_{n-1} \cdots T_1^{-1} \epsilon_\Lambda T_1 \cdots T_{n-1} \right)\\ \nonumber
&=&\pi_+ \left( L_s^{-1} \epsilon_\Lambda L_s\right)\\
\label{plusproj}
&=& \pi_+ \left( X(s) \right).
\end{eqnarray}
where, by Theorem \ref{factorisationthmbkg}, $X(s)$ solves FKToda with initial condition 
$X(0) = L_0^{-1} \epsilon_\Lambda L_0$. 
We recall from that theorem the construction of the factorization solution of FKToda, we set 
$L(s) =  L_{0}^{-1}L_s$ and
differentiate the relation 
\begin{eqnarray*} 
e^{sX} &=& L(s) R(s)\\
X e^{sX} &=& \frac{dL}{ds} R + L \frac{dR}{ds}\\
X L(s) R(s) &=& \frac{dL}{ds} R + L \frac{dR}{ds}\\
L^{-1}(s) X L(s) 
&=& L^{-1}(s) \frac{dL}{ds}  +  \frac{dR}{ds} R^{-1}(s)\\
X(s) &=& L^{-1}(s) \frac{dL}{ds}  +  \frac{dR}{ds} R^{-1}(s).
\end{eqnarray*}
From this it is clear that
$$
\pi_+ \left( X( s)\right) = \frac{dR}{ds} R^{-1}(s).
$$
Comparing this to (\ref{plusproj})
verifies (\ref{eqn:findLe}).
\begin{rem}
It is worth noting that $L(s)$ represents the continuous analogue of 
the relation between the discrete principal and Symes embeddings described in Section \ref{sec:homog}. This suggests a relation between the two in the nature of the latter corresponding to B\"acklund transformations for the former.
\end{rem}
 
We now turn to the derivation based on 
Equations (\ref{eqn:dLe}) and (\ref{eqn:findLe}) of explicit ODEs for the evolution of the Lusztig parameters. Starting with $j=1$,
equation (\ref{eqn:split}) becomes
\begin{eqnarray*}
T^{-1}_{1}\frac{d}{ds}T_1 &=& T^{-1}_{1} \epsilon_\Lambda T_1 
-  R^{(1)}\\
&=& T^{-1}_{1} \epsilon_\Lambda T_1 
- \pi_+(T^{-1}_{1} \epsilon_\Lambda T_1)\\
&=& \pi_-(T^{-1}_{1} \epsilon_\Lambda T_1)\\
\frac{d}{ds}T_1 &=& T_{1}\pi_-(T^{-1}_{1}\epsilon_\Lambda T_1).
\end{eqnarray*}

In general one knows from (\ref{Rj}) and Kostant's theorem that $R^{(i)}$ is upper Hessenberg. Furthermore, from the lower bidiagonal form of $T_j$,  one has that the only non-zero entries of $\frac{d}{ds} T_j$ are those of height $-1$.\\
Hence the equation 
\begin{eqnarray}
\frac{d}{ds}T_j = T_{j}\pi_-(T^{-1}_{j}R^{(j-1)} T_j)\label{TjandRjminusone}
\end{eqnarray}
only yields non-trivial equations 
from the principal sub-diagonals of each side, which is a set of $n-j$ scalar ODEs:
\begin{eqnarray*}
\left(\frac{d}{ds}T_j\right)_{i+1,i} &=& \left(T_{j}\pi_-(T^{-1}_{j}R^{(j-1)} T_j)\right)_{i+1,i}\\
&=& \left(\pi_-(T^{-1}_{j}R^{(j-1)} T_j)_-\right)_{i+1,i}; \,\,\,\, i = 1, \dots, n-j.
\end{eqnarray*}
The second equality follows since $T_j$ is lower unipotent; hence, multiplying by it on the left does not alter the entries on the principal sub-diagonal of  $\pi_-(T^{-1}_{j}R^{(j-1)} T_j)$.

\begin{rem}\label{rem:extlusztig}
In what follows, for convenience, we extend the Lusztig coordinates $(u_i^j)_{1\leq i\leq j\leq n-1}$ to $(u_i^j)_{i,j\in\Z}$ by declaring $u_i^j=0$ for $i$ and $j$ not satisfying $1\leq i\leq j\leq n-1$.
\end{rem}

\begin{thm}\label{thm:matrixandcoordrepoffktoda}
Equations (\ref{Rjandjminusone}) and (\ref{TjandRjminusone}), along with the initialization $R^{(0)}=\epsilon_\Lambda$, are equivalent to the following equations:
\begin{equation}\frac{d}{ds}T_j=[T_j,T_j\epsilon-R^{(j-1)}],~~~~R^{(j)}=\epsilon_\Lambda+\left[\epsilon,\sum_{i=1}^j T_i\right]\label{eq:commutatorsfulltoda}\end{equation}
which are equivalent to the following equations for the Lusztig coordinates:
\begin{eqnarray}
    \dfrac{\dot{u}_i^{j+i-1}}{u_i^{j+i-1}} &=& R_{i+1,i+1}^{(j-1)}-R_{i,i}^{(j-1)}+u_{i-1}^{j+i-2}-u_i^{j+i-1},~~~1\leq j\leq n-1,~~~1\leq i\leq n-j\label{inductiveueqn}\\
    R_{i,i}^{(j)} &=& R_{i,i}^{(j-1)}+u_i^{j+i-1}-u_{i-1}^{j+i-2},~~~1\leq j\leq n-1,~~~1\leq i\leq n.\label{inductiveReqn}
\end{eqnarray}

\end{thm}
\begin{proof}
Since $T_j(s)$ is lower unipotent, one has that $T_{j,-}(s):=T_j(s)-\mathbb{I}_n$ is lower nilpotent, with $(T_{j,-}(s))^n=0$. It therefore follows that the inverse of $T_j(s)$ can be expressed as follows:
$$T_j^{-1}=\mathbb{I}_n-T_{j,-}+T_{j,-}^2-\cdots +(-1)^{n-1}T_{j,-}^{n-1}.$$
We now decompose $R^{(j-1)}$ as $\epsilon+R^{(j-1)}_\Delta$, where $(R^{(j-1)})_\Delta$ is a diagonal matrix. We therefore write $T_j^{-1}R^{(j-1)}T_j$ as
$$\left(\mathbb{I}_n-T_{j,-}+T_{j,-}^2-\cdots +(-1)^{n-1}T_{j,-}^{n-1}\right)(\epsilon +(\epsilon T_{j,-}+R^{(j-1)}_\Delta) + R^{(j-1)}_\Delta T_{j,-})$$
where each single term in the above is a matrix whose entries lie solely along some diagonal of a given height.\\
From the above, one sees that
\begin{enumerate}[(i)]
\item The principal superdiagonal of $T_j^{-1}R^{(j-1)}T_j$ is $\epsilon$, as expected. 
\item The diagonal of $T_j^{-1}R^{(j-1)}T_j$ is $\epsilon T_{j,-}+R^{(j-1)}_\Delta-T_{j,-}\epsilon$. 
\item The principal subdiagonal of $T_j^{-1}R^{(j-1)}T_j$ is $R^{(j-1)}_\Delta T_{j,-}-T_{j,-}(\epsilon T_{j,-}+R^{(j-1)}_\Delta)+T_{j,-}^2\epsilon$.
\item Likewise, the subdiagonals below the principal subdiagonal can be extracted similarly.
\end{enumerate}
We now claim that $T_j[T_j^{-1}R^{(j-1)}T_j]_-$ has nonzero entries only on its principal subdiagonal. What is clear \textit{a priori} is that this matrix is strictly lower triangular since it is the product of a lower triangular matrix, $T_j$, and and strictly lower triangular matrix $[T_j^{-1}R^{(j-1)}T_j]_-$. We now decompose the matrix into its upper and strictly lower projections:
$$T_j^{-1}R^{(j-1)}T_j=\pi_+(T_j^{-1}R^{(j-1)}T_j)+\pi_-(T_j^{-1}R^{(j-1)}T_j).$$
Multiplying both sides on the left by $T_j$ and rearranging yields
\begin{equation}T_j\pi_-(T_j^{-1}R^{(j-1)}T_j)=R^{(j-1)}T_j-T_j\pi_+(T_j^{-1}R^{(j-1)}T_j).\label{eq:tjprojconjdecomp}\end{equation}
Now, we observe that $R^{(j-1)}T_j$, being the product of an upper bidiagonal matrix and a lower bidiagonal matrix, has no nonzero entries below the principal subdiagonal. Likewise, $T_j\pi_+(T_j^{-1}R^{(j-1)}T_j)$, being the product of a lower bidiagonal matrix and an upper bidiagonal matrix, also has no nonzero entries below the principal subdiagonal. Therefore, by Equation (\ref{eq:tjprojconjdecomp}), the strictly lower triangular matrix $T_j\pi_-(T_j^{-1}R^{(j-1)}T_j)$ has no nonzero entries below its principal subdiagonal, and is therefore itself a principal subdiagonal matrix.\\

\n Now we may proceed to unpack Equations (\ref{TjandRjminusone}) and (\ref{Rjandjminusone}). Firstly, since $T_j$ is lower unipotent, the principal subdiagonal of $T_j\pi_+(T_j^{-1}R^{(j-1)}T_j)$ is equal to the principal subdiagonal of $\pi_+(T_j^{-1}R^{(j-1)}T_j)$. Thus, Equation (\ref{TjandRjminusone}) is equivalent to 
\begin{equation}
    \frac{d}{ds}T_j =R^{(j-1)}_\Delta T_{j,-}-T_{j,-}(\epsilon T_{j,-}+R^{(j-1)}_\Delta)+T_{j,-}^2\epsilon.\label{eq:derivTjdiags}
\end{equation}

\n Therefore, 
\begin{equation}\dot{u}_i^{i+j-1}=R_{i+1,i+1}^{(j-1)}u_i^{i+j-1}-(u_i^{i+j-1})^2-u_i^{i+j-1}R_{i,i}^{(j-1)}+u_i^{i+j-1}u_{i-1}^{i+j-2}\label{eq:lusztigder}\end{equation}
where, when $i=1$, the last term $u_0^{j-1}$ is understood to be zero (Remark \ref{rem:extlusztig}). Likewise, also using the convention set forth in Remark \ref{rem:extlusztig}, when $i>n-j$, $u_i^{i+j-1}=0$, so that Equation (\ref{eq:lusztigder}) becomes vacuous when beyond this bound. Dividing Equation (\ref{eq:lusztigder}) by $u_i^{i+j-1}$ and reordering terms yields Equation (\ref{inductiveueqn}) precisely.\\

\n Turning our attention to Equation (\ref{Rjandjminusone}), we recall that $T_j^{-1}R^{(j-1)}T_j$ is Hessenberg, hence $\pi_+(T_j^{-1}R^{(j-1)}T_j)$ is comprised of the superdiagonal and diagonal of $T_j^{-1}R^{(j-1)}T_j$. By the analysis above, namely items (i) and (ii), Equation (\ref{Rjandjminusone}) yields
\begin{equation}
    R^{(j)}_{i,i}=\left(\epsilon T_{j,-}+R^{(j-1)}_\Delta-T_{j,-}\epsilon\right)_{i,i}=u_i^{i+j-1}+R_{i,i}^{(j-1)}-u_{i-1}^{i+j-2}\label{eq:Rsbydiags}
    \end{equation}
which gives Equation (\ref{inductiveReqn}).\\

\n Finally, to see the commutator expressions for the evolutions of $T_j$ and $R^{(j)}$, we note that Equation (\ref{eq:derivTjdiags}) can immediately be written as
$$\frac{d}{ds}T_{j,-}=[T_{j,-},~T_{j,-}\epsilon-R^{(j-1)}_\Delta].$$
Next, we note that we can add the identity inside the derivative, as well as in the first argument of the commutator without affecting either the derivative or commutator, and we may also add and subtract $\epsilon$ in the second argument:
$$\frac{d}{ds}T_j=[T_j,~T_{j,-}\epsilon+\epsilon-\epsilon-R^{(j-1)}_\Delta]\\
=[T_j,~T_j\epsilon-R^{(j-1)}],$$
which is the first equation of (\ref{eq:commutatorsfulltoda}). For the second equation of (\ref{eq:commutatorsfulltoda}), we see that Equation (\ref{eq:Rsbydiags}) can be expressed as
$$R^{(j)}_\Delta=R^{(j-1)}_\Delta +[\epsilon,T_{j,-}].$$
We add $\epsilon$ to both sides, and add the identity matrix in the second argument of the commutator to obtain
$$R^{(j)}=R^{(j-1)} +[\epsilon,T_{j}]=R^{(j-2)} +[\epsilon,T_{j}]+[\epsilon,T_{j-1}]=\cdots =R^{(0)}+\sum_{i=1}^j [\epsilon,T_{i}].$$
The second equation of (\ref{eq:commutatorsfulltoda}) follows from this by recalling that $R^{(0)}=\epsilon_\Lambda$ and by using linearity of the commutator in the second argument.
\end{proof}

We show in Appendix \ref{appendixb}
how Equations (\ref{inductiveueqn}) and (\ref{inductiveReqn}) can be used for an alternative derivation of O'Connell's system of differential equations in (6.4) of \cite{bib:o}.

\section{Concluding Remarks} \label{sec:conclusions}

\subsection{Connections to the Literature}\label{sec:connectionsliterature}

\n The following equations in \cite{bib:sik} describe their dhToda dynamics:
\begin{equation}\label{shinjohungrytodaeqs}
\left\{
\begin{array}{l}
Q_1^{(s,n)}=q_1^{(s,n)}-\mu^{(n)},\\
q_k^{(s,n+1)}+Q_{k+1}^{(s,n)}=q_{k+1}^{(s,n)}+Q_k^{(s+M,n)},~~~k=1,2,\ldots,m-1,\\
q_k^{(s,n+1)}Q_k^{(s,n)}=q_k^{(s,n)}Q_k^{(s+M,n)},~~~k=1,2,\ldots,m,\\
e_{k-1}^{(s,n+1)}+Q_k^{(s+1,n)}=e_k^{(s,n)}+Q_k^{(s,n)},~~~k=1,2,\ldots,m,\\
e_{k}^{(s,n+1)}Q_k^{(s+1,n)}=e_{k}^{(s,n)}Q_{k+1}^{(s+1,n)},~~~k=1,2,\ldots,m-1,\\
e_0^{(s,n)}:=0,~~~e_m^{(s,n)}:=0.
\end{array}
\right.
\end{equation}

\n In a coupled matrix representation, the dynamic evolution is achieved by a combination of iterated lower-upper factorizations, each of which is a regular dToda evolution, followed by an additional factoring algorithm between two upper bidiagonal matrices (one for the $q$-variables and another for the $Q$-variables). Taking the sequence $(\mu^{(n)})_{n\in\N}$ to be the zero sequence forces $Q_k^{(s,n)}=q_1^{(s,n)}$, which can be seen from the first three equations in (\ref{shinjohungrytodaeqs}). 
This is often referred to as the {\it autonomous} dhToda equation and arises in the study of algorithms for finding eigenvalues of certain classes of banded matrices \cite{bib:fiyin}. The resulting system is still a little more general than the dFToda system we present in this paper, but choosing to set some additional $e$-variables equal to zero restricts this system to dFToda.\\

The approach of \cite{bib:sik} is based on extended Hankel determinants that are related to a (discrete) time variation of the measure of orthogonality for families of biorthognal polynomials. We mention that the connection between biorthogonal polynomials and FKToda was already noticed by Ercolani and McLaughlin in \cite{bib:em}. The approach of this paper, based on a direct embedding of Symes's algorithm into full Toda, avoids the need to consider any ancillary determinants and leads to a more direct derivation of the relations between FKToda and dFToda. In addition the approach of \cite{bib:sik} is formal and this makes it difficult to precisely characterize under what conditions the systems considered are well-posed. The derivations in this paper, based on Lusztig parametrization, make it possible to rigorously establish a completely precise characterization of when dFToda is well-posed (Theorem \ref{thm:wp}). Furthermore, the constructions introduced in Section \ref{geometry} reveal how many of the results in this paper may be extended to Full generalized Toda lattices associated to more general real semisimple Lie algebras (see Remark \ref{rem:Lie}).

It will be of interest to consider possible extensions of the results here to the more general dhToda systems described in \cite{bib:sik} which have connections to numerical eigenvalue algorithms \cite{bib:yf}. That work also points out intriguing connections to ecological models known as discrete hungry Lotka-Volterra systems and continuum analogs. Furthermore, in ultradiscrete versions of discrete hungry Toda \cite{bib:tns}, an advanced box-ball system  is presented, involving labelling balls and moving balls according to a label-determined order of priority. In these treatments, the key focus regarding solitonic behaviour is the so-called \textit{soliton scattering rule} which describes how labels redistribute amongst blocks under the ultradiscrete hungry Toda evolution, as well as associated conserved quantities of this particular cellular automaton realization of the dynamics. The system we introduced in Section \ref{sec:ud} is simpler in the sense that it does not require additional structures beyond standard BBS (\textit{i.e.}, not requiring labellings, higher capacities, etc.); nevertheless, the fact that the discrete full Toda lattice may be obtained as a particular special case of the discrete hungry Toda molecule, suggests that the general ultradiscrete systems may lead to further interesting extensions to what we have done in this paper, such as Lie-theoretic interpretations.

~\\
\n We have noted that O'Connell derived ODEs for the induced dynamics of the Toda flow on Lusztig parameters. We have compared this to the FKToda flow on parameters that we derived in Section \ref{sec:return}. From this we are able to directly recover O'Connell's equations (see Appendix \ref{appendixb}). What we do is in many ways a more direct and elementary derivation of these ODEs; however, O'Connell is able to extend these to stochastic equations for reflected Brownian motions on the line, with many deep and fascinating applications to processes that break open new territory. It would be interesting to explore similar applications for our approach. In particular, it may be promising to interpret dFToda as B\"acklund transformations for FKToda in a stochastic setting.

\subsection{Future Directions }

The classical Toda lattice has been central to developments and applications in many areas of mathematics ranging from representation theory to quantum mechanics \cite{bib:howe}, from Hamiltonian dynamics/ symplectic geometry to statistical mechanics \cite{bib:spohn} and from combinatorics to random geometry \cite{bib:ew}. We believe the developments described in this paper set the stage for new developments and applications in related veins. Here we just mention a couple of those directions (by no means exhaustive) that it would be natural to pursue in the near future. 

\subsubsection{Integrability and Lie Theoretic Generalizations}
One such direction concerns the complete integrability of FKToda and its discretizations done in an extended fully Lie theoretic setting.  While the eigenvalues of $X_0$ provide just enough constants of motion (see (\ref{Isospectral})) to establish the complete integrability of the
 tridiagonal Toda lattice, they do not suffice for the Full Kostant Toda  lattice. Finding additional invariants relies on reformulating the 
  FKToda equations as a Hamiltonian system with respect to a geometrically defined Poisson bracket that meshes naturally with the geometry of the flag manifolds discussed in Sections \ref{geometry} and \ref{sec:LusDynam}. This Poisson structure is based on the (co-adjoint) orbit method of Kirillov-Kostant
(K-K). The terminology stems from the fact that the Lie bracket on the Lie algebra $\frak{h}$ of a general Lie group $H$ induces a natural Poisson bracket on the dual space $\frak{h}^*$. The orbits on this dual (called co-adjoint orbits) induced by the conjugation action of $H$ on $\frak{h}$ are the symplectic leaves (submanifolds of $\frak{h}^*$ on which the Poisson bracket restricts to be non-degenerate).

Using the G-invariant, non-degenerate inner product, $ \langle X, Y\rangle = \tr(XY)$, on $\frak{g}$, the dual space  $\frak{b}_+^*$ may be identified with $\mathcal{H}$. For the K-K Poisson bracket on $\mathcal{H}$ this may be concretely expressed, for functions $f, g$ on  $\frak{g}$, as 
\begin{eqnarray} \label{KKPB}
\{\tilde{f}, \tilde{g} \}(X) &=& \langle X, \left[ \pi_+ \nabla f(X),  \pi_+ \nabla g(X)\right]\rangle
\end{eqnarray}
where $\tilde{f} = f|_{\mathcal{H}}$ and $\nabla$ denotes the gradient with respect to the inner product on $\frak{g}$. The tridiagonal elements of $\mathcal{H}$, with $b_i > 0$, comprise a coadjoint orbit of $B_+$ in $\frak{b}_+^*$ and the restriction of
(\ref{KKPB}) to these matrices does indeed yield a bracket equivalent to the standard one for the Toda lattice equations of (\ref{hamiltonian}). Moreover, the Hamiltonian for FKToda  is the same as that for the tridiagonal case in Flaschka's coordinates: $\tr X^2$. So despite its initially somewhat complicated appearance, the $\frak{b}_+^*$ K-K orbit structure is the correct generalization of Hamiltonian structure for FKToda. 
 
 The additional constants of motion needed to completely integrate FKToda are constructed by considering nested sequences 
 of parabolic subalgebras $\frak{p}_{k}^*$ of $\frak{g}$ (a parabolic subalgebra is one that contains $\frak{b}_+$) with corresponding Lie group $P_k$:
 \begin{eqnarray} \nonumber 
\frak{b}_+ &=& \frak{p}_{m+1} \subset \cdots \subset \frak{p}_{k} \subset \cdots \subset \frak{p}_{0} = \frak{g} \\ \label{parabolic}
\frak{b}_+^* &=& \frak{p}_{m+1}^* \leftarrow \cdots \leftarrow \frak{p}_{k}^* \leftarrow \cdots \leftarrow \frak{p}_{0}^* = \frak{g}^*
\end{eqnarray}
 in which the maps on the second line are the projections induced by duality from the inclusions on the first line. Those projections are {\it Poisson maps},  meaning that the pullback of the K-K Poisson bracket of two functions on 
 $\frak{p}_{k-1}^*$  equals the K-K Poisson bracket of their pullbacks on $\frak{p}_{k}^*$ . 
 
One can show, based on the Factorization Theorem, that functions on  $\frak{p}_{k-1}^*$ invariant under the coadjoint action of
$P_{k-1}$, Poisson commute with {\it all} functions on $\frak{p}_{k-1}^*$. Since the projections in (\ref{parabolic}) are all Poisson maps one can collectively pull back all these sequentially invariant functions to 
$\frak{g}^*$ to form a family of involutive (commuting) functions there. Then restricting this family to $\frak{b}_+^*$ will, by duality, give an involutive family of functions on $\mathcal{H}$ that are constants of motion for FKToda. This is called the {\it method of Thimm} \cite{bib:thimm}. Of course it could be that this family is empty; however, that is not the case. In \cite{bib:efs} it was shown that for a particular chain of parabolic subalgebras this method yields exactly enough independent invariants of this type to demonstrate complete integrability for the Hessenberg phase space. They further showed that this analysis can also be successfully applied to generalized full Toda systems associated to Lie algebras $B_n, C_n$ and $D_n$.   Subsequently Gekhtman and Shapiro, in an elegant reformulation \cite{bib:gs,bib:r}, showed how this further applies to all simple Lie algebras. 

Our re-formulation of FKToda and dFToda in Section \ref{geometry} now makes it natural to extend these considerations of Poisson-Lie structures and integrability to the dFToda setting of Hamiltonian/symplectic maps. Notions of total positivity and Lusztig parameters fit naturally into our framework  but appear to be little studied. Extensions of the above described ideas to the ultra-discrete/box-ball systems also poses some fascinating avenues for study. One possible link there is to recent developments concerning crystals for affine Lie algebras (see \cite{bib:scrimshaw} and references therein). Links between factorization and crystal bases for quantum groups were discovered by Lusztig (see \cite{bib:bfz} for an exposition); but, relating this to symplectic geometry or Hamiltonian dynamics remains largely unexplored as far as we know.

\subsubsection{Representation Theory of Solvable Lie groups and Geometric Quantization}

Connections between completely integrable Hamiltonian systems and quantum mechanics go back to the earliest days of the latter theory, the so-called {\it old quantum theory} of Bohr and Sommerfeld with subsequent refinements by Einstein-Brillouin-Keller (EBK). The basic connection here is that completely integrable systems typically have big symmetry groups related to their constants of motion \`a la Noether's theorem. The first step of quantization is to take the Hamiltonian that defines the dynamical system, written in canonical variables,  and replace the $p_i$ by $i \partial/\partial q_i$ and regard functions of the $q_i$ as multiplication operators. This then effectively replaces the nonlinear Hamiltonian by a {\it linear} differential operator also referred to as the Hamiltonian.   In the case of the tridiagonal Toda lattice the Hamiltonian (\ref{hamiltonian}) is replaced by the operator
 \begin{eqnarray} \label{quantumT}
\widehat{H} = - \frac12 \sum_{i =1}^n \partial^2/\partial q_i^2 + \sum_{j=1}^{n-1}e^{q_j-q_{j+1}}.
\end{eqnarray}
The other constants of motion, $H_i $ in involution with $H$ also have such linear operator representations $\widehat{H}_i$.
Then the EBK formalism determines which values of the motion constants are admissible as spectra (``energy" levels) for this 
commuting family of Hamiltonians.  This associates to those levels, regarded as eigenvalues, common eigenfunctions which are
the {\it states} of  the quantum system and should comprise a basis for its associated Hilbert space. 

On the other hand, the symmetry group associated to the classical integrability acts as a group of symmetries on the collective operators $\{\widehat{H}\}  \cup \{ \widehat{H}_i\}$ and so in turn acts as a group representation on the associated Hilbert space. 
The decomposition of that representation into a sum of irreducible representations then describes the composition of general states in terms of fundamental eigenstates which is a fundamental question in quantum theory. Due to the connection with group theory (about to be discussed) one expects the eigenfunctions to have explicit expressions, such as in terms of spherical functions. 

Turning this around, one can also start with a Lie group and try to build for it integrable systems that realize irreducible representations by the above prescription. These integrable systems would then become the fundamental building blocks for representations of the group. This is precisely what the program, referred to as the Kirillov-Kostant (K-K) orbit method (the quantum analogue of the Poisson orbit method described in the previous subsection), is designed to do. It is the most important component of Kostant's  more general program of {\it Geometric Quantization} which attempts to place the quantization prescription described above on a firm rigorous footing for solvable groups. The K-K orbit method may be implemented for the Lie group $B_+$, the group that underlies the work in this paper. The co-adjoint orbits of $B_+$ on $\frak{b}_+^*$ are symplectic leaves on which local canonical coordinates of  $(p_i, q_i)$ may be constructed as the basis for geometric quantization. For the tridiagonal symplectic leaves one has the quantization of the tridiagonal Toda Lattice described above.
The representation theory in this case has been completely worked out by Kostant \cite{bib:kostant}. The irreducible representations are infinite dimensional (possible since $B_+$ is non-compact) and are indexed by the joint continuous spectrum of (\ref{quantumT}) and the other $\widehat{H}_i$ that Poisson commute with it. The eigenfunction basis is expressed in terms of generalizations of the classical Whittaker functions. This class of examples is one of the principal models and success stories for the K-K orbit method.

The focus in relation to this paper will be to extend the above type of analysis to other symplectic leaves ($B_+$ orbits)  such as generic maximal dimensional orbits of the Full Kostant Toda lattice or the $m$-banded invariant sets for $3 < m < n$ as well as to their discretizations. This is wide open territory with many challenges but the results described in this paper make it possible to begin to undertake those challenges. 

\section*{Statements and Declarations}
\textbf{Conflict of Interest} Authors declare no conflict of interests for this article.

\appendix

\section{Tau-Function Description for Lower-Upper Hessenberg Factorization} \label{appA}

\begin{lem}\label{LemmaSemiInfiniteFactorisationTau}
Let 
$$H=\left[\begin{array}{cccc}
g_1 & 1 \\
h_1 & g_2 & \ddots\\
&\ddots & \ddots & 1\\
&&h_{n-1} & g_n
\end{array}
\right],$$
with $h_i\neq 0$ for all $i$, then either there exist unique sequences $\mb{y}=(y_1,y_2,\ldots,y_{n-1})$ and $\mb{b}=(b_1,b_2,\ldots,b_n)$ such that
$$
\left[\begin{array}{cccc}
1 \\
y_1 & 1 \\
&\ddots & \ddots\\
&&y_{n-1}& 1
\end{array}\right]
\left[\begin{array}{cccc}
b_1 & 1\\
& b_2 & \ddots\\
&&\ddots & 1\\
&&&b_n
\end{array}\right]=H$$
or no solution exists. Furthermore, if a solution exists, it is given by\vspace{0.2cm}
\begin{align}
b_i&=\frac{\tau_i(H)}{\tau_{i-1}(H)}~~~~~i=1,2,\ldots,n\ldots\label{equationsforbs}\\
y_i&=\frac{h_i\tau_{i-1}(H)}{\tau_i(H)}~~~i=1,2,\ldots,n-1,\label{equationsforys}
\end{align}~\\[-12pt]
where $\tau_k(H)$ is defined to be the principal $k^\text{th}$ minor determinant of $H$ if $1\leq k\leq n$, and $\tau_0(H)$ is defined to be $1$.
\end{lem}
\begin{proof}
We start by multiplying to obtain\vspace{0.2cm}
\begin{equation}
\left[\begin{array}{cccc}
g_1 & 1 \\
h_1 & g_2 & \ddots\\
&\ddots & \ddots & 1\\
&&h_{n-1} & g_n
\end{array}
\right]
=
\left[\begin{array}{cccc}
b_1 & 1 \\
b_1y_1 & b_2+y_1 & \ddots\\
&\ddots & \ddots & 1\\
&&b_{n-1}y_{n-1} & b_n+y_{n-1}
\end{array}
\right]
\end{equation}~\\[-12pt]
Comparing coefficients yields the following system of equations:
\begin{align}
b_1&=g_1\\
b_iy_i&=h_i~~~~~\,~i=1,2,\ldots,n-1\label{syseqtwo}\\
b_{i+1}+y_i&=g_{i+1}~~~~i=1,2,\ldots,n-1 \label{syseqthree}
\end{align}
Since $h_i\neq 0$ for all $i$, it is clear that one can solve uniquely for the $(b_1,y_1,b_2,y_2,\ldots)$ in that order, provided no $b_i$ vanishes, assuring Equation (\ref{syseqtwo}) can be solved for $y_i$. \\[5pt]
Completion of the proof boils down to proving Equation (\ref{equationsforbs}), since Equation (\ref{equationsforys}) follows immediately from this and Equation (\ref{syseqtwo}). We proceed by induction on $i$:\\[5pt]
For $i=1$, we have \vspace{0.25cm}
\begin{equation}
b_1=g_1=\tau_1(H)=\frac{\tau_1(H)}{\tau_0(H)}
\end{equation}~\\[-14pt]
so that the base case holds trivially.\\[5pt]
Assume $b_i=\frac{\tau_i(H)}{\tau_{i-1}(H)}$ for some $i<n$, then, by Equations (\ref{syseqtwo}) and (\ref{syseqthree}), one has
\begin{equation}
b_{i+1}=g_{i+1}-\frac{h_i}{b_i} = \frac{g_{i+1}b_i-h_i}{b_i}.
\end{equation}
\n Using the induction hypothesis, one then obtains\vspace{0.2cm}
\begin{equation}
b_{i+1}=\frac{g_{i+1}\tau_i(H)-h_i\tau_{i-1}(H)}{\tau_i(H)}=\frac{\tau_{i+1}(H)}{\tau_i(H)}
\end{equation}~\\[-12pt]
where the latter equality follows from a cofactor expansion of the upper left $(i+1) \times (i+1)$ submatrix, along its bottom row.
\end{proof}

\n We now use Lemma \ref{LemmaSemiInfiniteFactorisationTau} to prove Equations (\ref{tauf}) and (\ref{taus}). We recall the definition of the dFToda evolution process:
\begin{enumerate}
    \item $R(t)=:R^{(0)}(t)$
    \item $T_i(u(t+1))R^{(i)}(t)=R^{(i-1)}(t)T_i(u(t))$, where $L(t)=T_1(u(t))T_2(u(t))\cdots T_{n-1}(u(t))$.
\end{enumerate}
We also recall the following definition (Equation (\ref{defnhessenbergtruncs})):
\begin{equation}
H^{(i)}(t):=R^{(i-1)}(t)T_i(t),\end{equation}
as well as the associated $\tau$-functions (Equation (\ref{taufundefnithhess})):
$$\tau_j^i(t):=\tau_j(H^{(i)}(t))$$
which is the principal $j$-th minor of $H^{(i)}(t)$ for $1\leq j\leq n$, and is defined to be 1 when $j=0$.
\begin{defn}
For each $1\leq i\leq n-1$, define $L^{(i)}(t)=L_1(t)L_2(t)\cdots L_i(t)$.
\end{defn}
By Lemma \ref{LemmaSemiInfiniteFactorisationTau}, one finds $T_i(u(t+1))$ and $R^{(i)}(t)$ as the following:
\begin{equation}\label{eqn:Tiplusone}
T_i(u(t+1))=\left[\begin{array}{ccccc}
1\\ 
H^{(i)}_{2,1}(t) \frac{\tau_0^i(t)}{\tau_1^i(t)} & 1\\
&H^{(i)}_{3,2}(t) \frac{\tau_1^i(t)}{\tau_2^i(t)} & \ddots\\
&&\ddots & 1\\
&&&H^{(i)}_{n,n-1}(t) \frac{\tau_{n-2}^i(t)}{\tau_{n-1}^i(t)} &1
\end{array}\right]
\end{equation}
and 
$$R^{(i)}(t)=\left[\begin{array}{cccc}
\frac{\tau_1^i(t)}{\tau_0^i(t)} & 1\\
&\frac{\tau_2^i(t)}{\tau_1^i(t)}&\ddots\\
&&\ddots & 1\\
&&&\frac{\tau_n^i(t)}{\tau_{n-1}^i(t)}
\end{array}\right].$$
Therefore, determining equations for the entries of $T_i(u(t+1))$ and $R^{(i)}(t)$ amounts to determining $H^{(i)}_{j+1,j}(t)$ for $j=1,\ldots,n-1$ and $\tau_j^i(t)=\tau_j(R^{(i-1)}(t)T_i(u(t)))$ for $j=1,\ldots,n$.\\
We now make the following key observation:
\begin{align*}
    \tau_k(H^{(i)}(t))&=\tau_k(R^{(i-1)}(t)T_i(u(t)))\\
    &=\tau_k(T_1(u(t+1))T_2(u(t+1))\cdots T_{i-1}(u(t+1))R^{(i-1)}(t)T_i(u(t)))\\
    &=\tau_k(R^{(0)}(t)T_1(u(t))T_2(u(t))\cdots T_{i-1}(u(t))T_i(u(t)))\\
    &=\tau_k(R(t)L^{(i)}(t)),
\end{align*}
where the second equality comes from the the invariance of principal minors under left-multiplication by lower unipotent matrices, and the third equality comes from reversing the definition of the dFToda iterative construction of the $R^{(i)}$'s.\\

\n The significance of this observation lies in the fact that $\tau_k(H^{(i)}(t))$ is known at the outset, \textit{i.e.}, before having to compute $R^{(i)}(t)$ or $T_i(u(t+1))$ for any $i>0$.\\

\n We now turn our attention to $H^{(i)}_{j+1,j}(t)$. Multiplying out $H^{(i)}=R^{(i-1)}(t)T_i(u(t))$, one finds 
$$H^{(i)}(t)=\left[\begin{array}{cccccccc}
* & 1\\
u_1^i(t)\frac{\tau_2^{i-1}(t)}{\tau_1^{i-1}(t)}& * & 1\\
& u_2^{i+1}(t)\frac{\tau_3^{i-1}(t)}{\tau_2^{i-1}(t)} & * & 1\\
&&\ddots & \ddots & \ddots\\
&&&u_{n-i}^{n-1}(t)\frac{\tau_{n-i+1}^{i-1}(t)}{\tau_{n-i}^{i-1}(t)}\\
&&&&0 & * & \ddots\\
&&&&&\ddots  & * & 1\\
&&&&&  & 0 & 1\\
\end{array}\right],$$
resulting in the following description of the lower-diagonal entries of $H_{j+1,j}^i(t)$:
$$H_{j+1,j}^i(t) = \left\{
\begin{array}{cl}
u_j^{i+j-1}(t) \frac{\tau_{j+1}^{i-1}(t)}{\tau_{j}^{i-1}(t)} & \text{~~if~~}j\leq n-i\\
0 & \text{~~if~~}j> n-i
\end{array}
\right..$$
Thus, from Equation (\ref{eqn:Tiplusone}), we have
\begin{equation} \label{eqn:paramDynam}
    u_j^{i+j-1}(t+1) = H_{j+1,j}^{(i)}(t) \frac{\tau_{j-1}^i(t)}{\tau_j^i(t)}
    = u_j^{i+j-1}(t)\frac{\tau_{j+1}^{i-1}(t)}{\tau_j^{i-1}(t)}\frac{\tau_{j-1}^i(t)}{\tau_j^i(t)}.
\end{equation}
This completes the proof of Equations (\ref{tauf}) and (\ref{taus}), and therefore Theorem \ref{thm:wp}.

\section{O'Connell's ODE's}\label{appendixb}

\begin{thm}\label{thm:oconnellsetup}
Suppose $(R^{(i)})_{i=0}^{n-1}$ and $(T_j(s))_{j=1}^{n-1}$ satisfy Equations (\ref{TjandRjminusone}) and (\ref{Rjandjminusone}), along with the initialization $R^{(0)}=\epsilon_\Lambda$. By writing the positive Lusztig coordinates in the following way
$$u_i^j=e^{x_{i+1}^{j+1}-x_i^j},~~~~1\leq i\leq j\leq n-1,$$
there exist constants $(C_h)_{h=0}^{n-2}$ such that 
$$\dot{x}_i^{i+h}=C_h+\lambda_i-u_{i-1}^{i-1}+\sum_{l=1}^h \left(u_i^{i+h-l}-u_{i-1}^{i+h-l}\right).$$
\end{thm}
\begin{proof}
We first cascade Equation (\ref{inductiveReqn}) to obtain the following equation for $R_{i,i}^{(j)}$:
$$R_{i,i}^{(j)}=\lambda_i+\sum_{l=1}^j \left(u_i^{j+i-l}-u_{i-1}^{j+i-l-1}\right),$$
recalling that $R_{i,i}^{(0)}=\lambda_i$ for each $1\leq i\leq n$. By defining $h=j-1$, we now have the equations coming from (\ref{inductiveueqn}) and (\ref{inductiveReqn}):
\begin{eqnarray}
    \dfrac{\dot{u}_i^{i+h}}{u_i^{i+h}} &=& R_{i+1,i+1}^{(h)}-R_{i,i}^{(h)}+u_{i-1}^{i+h-1}-u_i^{i+h},~~~0\leq h\leq n-2,~~~1\leq i\leq n-h-1\label{uiheqn}\\
    R_{i,i}^{(h)}&=&\lambda_i+\sum_{l=1}^h \left(u_i^{i+h-l}-u_{i-1}^{i+h-l-1}\right),~~~0\leq h\leq n-2,~~~1\leq i\leq n.\label{Riheqn}
\end{eqnarray}
Combining Equations (\ref{uiheqn}) and (\ref{Riheqn}), and noting that $\dfrac{\dot{u}_i^{i+h}}{u_i^{i+h}}=\dot{x}_{i+1}^{i+h+1}-\dot{x}_{i}^{i+h}$, one obtains the following equation
$$\dot{x}_{i+1}^{i+h+1}-\dot{x}_{i}^{i+h}
=
\lambda_{i+1}+\sum_{l=1}^h \left(u_{i+1}^{i+h-l+1}-u_{i}^{i+h-l}\right)
- \lambda_i-\sum_{l=1}^h \left(u_i^{i+h-l}-u_{i-1}^{i+h-l-1}\right)
+u_{i-1}^{i+h-1}-u_i^{i+h}$$
which can be rewritten as
$$\dot{x}_{i+1}^{i+h+1}-\lambda_{i+1}+\sum_{l=1}^h \left(u_{i}^{i+h-l}-u_{i+1}^{i+h-l+1}\right)+u_i^{i+h}
=
\dot{x}_{i}^{i+h}- \lambda_i+\sum_{l=1}^h \left(u_{i-1}^{i+h-l-1}-u_i^{i+h-l}\right)
+u_{i-1}^{i+h-1}.
$$
Consequently, each side of the above equation is constant in $i$. If we name this constant $C_h$, we see that
$$\dot{x}_{i}^{i+h}- \lambda_i+\sum_{l=1}^h \left(u_{i-1}^{i+h-l-1}-u_i^{i+h-l}\right)
+u_{i-1}^{i+h-1}=C_h$$
for each $1\leq i\leq n-1-h$. Thus,
\begin{align*}
    \dot{x}_i^{i+h}
    &=C_h+\lambda_i-u_{i-1}^{i+h-1}-\sum_{l=1}^h u_{i-1}^{i+h-l-1}+\sum_{l=1}^h u_i^{i+h-l}\\
    &=C_h+\lambda_i-u_{i-1}^{i-1}-\sum_{l=0}^{h-1} u_{i-1}^{i+h-l-1}+\sum_{l=1}^h u_i^{i+h-l}\\
    &=C_h+\lambda_i-u_{i-1}^{i-1}+\sum_{l=1}^h \left(u_i^{i+h-l}-u_{i-1}^{i+h-l}\right).
\end{align*}
\end{proof}

\begin{cor}\label{cor:oconnellsetup}
 The equations in Theorem \ref{thm:oconnellsetup} can now be reduced to Equation 6.4 in \cite{bib:o}, with the addition of some scalar addition along \textit{`diagonals'} $(\dot{x}_i^{i+h})_{1\leq i\leq n-1-h}$:
 \begin{eqnarray}
 &~~~~~~~~~~~~\dot{x}_1^1=C_0+\lambda_1, \label{eq:refoc1}\\
& \dot{x}_1^m=C_{m-1}-C_{m-2}+\dot{x}_1^{m-1}+e^{x_2^m-x_1^{m-1}},~~~~~\dot{x}_m^m=C_0+\lambda_m-e^{x_m^m-x_{m-1}^{m-1}},~~~2\leq m\leq n,\label{eq:refoc2}\\
&~~~~~\dot{x}_i^m=C_{m-i}-C_{m-i-1}+\dot{x}_i^{m-1}+e^{x_{i+1}^m-x_i^{m-1}}-e^{x_i^m-x_{i-1}^{m-1}},~~~~1<i<m\leq n.\label{eq:refoc3}
 \end{eqnarray}
\end{cor}

\begin{proof}
Applying Theorem \ref{thm:oconnellsetup} with $h=0$, one immediately retrieves Equation (\ref{eq:refoc1}) and the second equation of (\ref{eq:refoc2}):
$$\dot{x}_i^{i}=C_0+\lambda_i-u_{i-1}^{i-1}=\left\{\begin{array}{cc}
C_0+\lambda_1 & \text{~~if~}i=1\\
C_0+\lambda_i-u_{i-1}^{i-1} & \text{~~if~}2\leq i\leq n
\end{array}\right..$$
If $h>0$, consider $\dot{x}_i^{i+h}-\dot{x}_i^{i+h-1}$:
\begin{align*}
    \dot{x}_i^{i+h}-\dot{x}_i^{i+h-1}
    &=
    C_h+\lambda_i-u_{i-1}^{i-1}+\sum_{l=1}^h \left(u_i^{i+h-l}-u_{i-1}^{i+h-l}\right)
    -C_{h-1}-\lambda_i+u_{i-1}^{i-1}-\sum_{l=1}^{h-1} \left(u_i^{i+h-1-l}-u_{i-1}^{i+h-1-l}\right)\\
    &=C_h-C_{h-1}+u_i^{i+h-1}-u_{i-1}^{i+h-1}+\sum_{l=2}^h \left(u_i^{i+h-l}-u_{i-1}^{i+h-l}\right)-\sum_{l=2}^{h} \left(u_i^{i+h-l}-u_{i-1}^{i+h-l}\right)\\
    &=C_h-C_{h-1}+e^{x_{i+1}^{i+h}-x_i^{i+h-1}}-e^{x_i^{i+h}-x_{i-1}^{i+h-2}}.
\end{align*}
Setting $h=m-i$, this can now be rewritten as 
$$\dot{x}_i^m=C_{m-i}-C_{m-i-1}+\dot{x}_i^{m-1}+e^{x_{i+1}^m-x_i^{m-1}}-e^{x_i^m-x_{i-1}^{m-1}},$$
which is Equation (\ref{eq:refoc3}), and reduces to the first equation of (\ref{eq:refoc2}) by recalling that $u_0^{m-1}=e^{x_1^{m}-x_{0}^{m-1}}=0$.
\end{proof}


\newpage
\section*{Glossary}
\addcontentsline{toc}{section}{Glossary}%
\begin{figure}[H]
\centering\renewcommand{\arraystretch}{1.5}
\begin{tabular}{cl}
\textbf{Symbol} & \multicolumn{1}{c}{\textbf{Meaning}}\\
$\mathfrak{b}_+$ & The Lie subalgebra of $\mathfrak{gl}_n$ of upper triangular $n\times n$ matrices\\
$\mathfrak{b}_-$ & The Lie subalgebra of $\mathfrak{gl}_n$ of lower triangular $n\times n$ matrices\\
$\mathfrak{n}_+$ & The Lie subalgebra of $\mathfrak{gl}_n$ of upper nilpotent $n\times n$ matrices\\
$\mathfrak{n}_-$ & The Lie subalgebra of $\mathfrak{gl}_n$ of lower nilpotent $n\times n$ matrices\\
$\pi_+$ & The Lie algebra projection of $\mathfrak{gl}_n=\mathfrak{b}_+\oplus \mathfrak{n}_-$ onto $\mathfrak{b}_+$ \\
$\pi_-$ & The Lie algebra projection of $\mathfrak{gl}_n=\mathfrak{b}_+\oplus \mathfrak{n}_-$ onto $\mathfrak{n}_-$\\
$B_+$ & The Lie subgroup of $\text{GL}_n$ of invertible upper triangular matrices.\\
$N_-$ & The Lie subgroup of $\text{GL}_n$ of lower unipotent triangular matrices.\\
$\Pi_+$ and $\Pi_-$ & For $g=nb\in N_-B_+\subset \text{GL}_n$, $\Pi_-(g)=n\in N_-$ and $\Pi_+(g)=b\in B_+$.\\
$\epsilon$ & The $n\times n$ matrix with ones on its superdiagonal, and zeroes everywhere else.\\
$\epsilon_\Lambda$ & For a tuple of eigenvalues $\Lambda=(\lambda_1,\lambda_2,\ldots,\lambda_n)$, this is $\epsilon+\text{Diag}(\Lambda)$\\
$\mathcal{H}$ & The space of Hessenberg matrices, \textit{i.e.}, $\mathcal{H}=\epsilon + \mathfrak{b}_-$\\
$\mathcal{B}$ & The set of upper bidiagonal matrices with ones on the superdiagonal.\\
$\mathcal{B}^{>0}$ & The set of matrices in $\mathcal{B}$ with positive diagonal entries.\\
$N_-^{>0}$& The set of totally positive lower unipotent matrices\\
$\mathcal{H}^{>0}$& The totally positive Hessenberg matrices: $\mathcal{H}^{>0}=N_-^{>0}\times \mathcal{B}^{>0}$\\
$s$& Continuous time\\
$t$& Discrete time\\
$\ell_i(a)$ & The matrix $\mathbb{I}_n+aE_{i+1,i}$, where $\mathbb{I}_n$ is the $n\times n$ identity matrix, and $E_{i+1,i}$ is\\& the matrix with one in the $(i+1,i)$ entry and zero everywhere else.\\
$T_i(u)$& For $u=(u_i^j)_{1\leq i\leq j\leq n-1}$, $T_i(u)=\ell_1(u_1^i)\ell_2(u_2^{i+1})\cdots \ell_{n-i}(u_{n-i}^{n-1})$\\
$\tau_j$ & The function on $n\times n$ matrices returning the principal $j\times j$ minor.
\end{tabular}
\end{figure}

\newpage



\end{document}